\documentclass[a4paper,UKenglish]{lipics-v2016}
\usepackage{microtype}
\author[1]{Bassel Mannaa}
\author[2]{Rasmus Ejlers Møgelberg}

\affil[1,2]{Department of Computer Science, IT University of Copenhagen, Copenhagen, Denmark\\
\texttt{basm@itu.dk,mogel@itu.dk}}

\title{The clocks they are adjunctions \protect\\ \Large{Denotational semantics for Clocked Type Theory}}
\titlerunning{The clocks they are adjunctions}

\date{\today}

\authorrunning{B Mannaa, R E Møgelberg}
\Copyright{}

\subjclass{Dummy classification -- please refer to \url{http://www.acm.org/about/class/ccs98-html}}

\ArticleNo{}

\usepackage{amsmath}
\usepackage{amssymb} 
\usepackage{ifxetex}
\usepackage[T1]{fontenc}
\usepackage[utf8]{inputenc}
\usepackage[cal=boondox]{mathalfa}
\usepackage{stmaryrd}

\usepackage{url}
\usepackage{authblk}

\usepackage{tikz-cd}
\newenvironment{diagram}{\begin{tikzcd}[sep=large]}{\end{tikzcd}}

\theoremstyle{plain}

\usepackage{amsthm}

\newtheorem{proposition}[theorem]{Proposition}
\newenvironment{proofof}[1]
{\begin{proof}[Proof of {#1}]}
{\end{proof}}
\newenvironment{proofsketch}
{\begin{proof}[Proof (sketch)]}
{\end{proof}}

\usepackage{mathpartir}

\newcommand\sym[1]{\mathsf{#1}}

\usepackage{tikz-cd}

\newcommand{\cat}[1]{\mathcal{#1}}

\newcommand{\id}{\mathsf{id}}

\newcommand{\set}{\mathsf{Set}}

\newcommand{\inv}[1]{#1^{-1}}

\newcommand{\catT}{\mathbb{T}}
\renewcommand{\L}{\mathsf{L}}
\newcommand{\R}{\mathsf{R}}
\newcommand{\transp}[1]{\overline{#1}}

\newcommand{\tick}{\mathrm{tick}}

\newcommand{\opcat}[1]{{{#1}^{\mathrm{op}}}}

\newcommand{\triple}[3]{(#1;#2;#3)}
\newcommand{\timeobj}[2]{(#1;#2)}

\newcommand{\grtotal}{\mathsf{GR}}
\newcommand{\gr}[1]{\mathsf{GR}[#1]}

\newcommand{\Fam}{\mathsf{Fam}}

\newcommand{\typ}{\mathsf{Ty}}
\newcommand{\trm}{\mathsf{Tm}}
\newcommand{\cwftm}[3]{#1 \vdash #2 : #3}
\newcommand{\cwfty}[2]{#1 \vdash #2}

\newcommand{\p}[0]{\mathsf{p}}
\newcommand{\q}[0]{\mathsf{q}}
\newcommand{\compr}[2]{#1.#2}
\newcommand{\cpair}[2]{\la #1, #2\ra}
\newcommand{\isocomp}[2]{\zeta_{{#1},{#2}}}

\newcommand{\pair}[2]{\left(#1,#2\right)}

\newcommand{\Nat}[0]{\mathsf{N}}

\newcommand{\la}[0]{\langle}
\newcommand{\ra}[0]{\rangle}

\newcommand{\dprod}[3]{\ensuremath{{\textstyle\prod\left(#1 : #2\right) . #3}}}

\newcommand{\subst}[2]{[#2/#1]}

\newcommand{\idty}[3]{#2 =_{#1} #3}

\newcommand{\later}{\triangleright}
\newcommand{\tearlier}{\blacktriangleleft}
\newcommand{\tlater}{\blacktriangleright}
\newcommand{\cv}{\mathrm{CV}}
\newcommand{\clk}{\mathrm{Clk}}

\newcommand{\dfix}{\mathsf{dfix}}

\newcommand{\nxt}{\mathsf{next}}
\newcommand{\pret}[1]{{\llbracket #1 \rrbracket}}

\newcommand\tabs[2]{\lambda (#1 : #2).}
\newcommand\tickc{\diamond}
\newcommand\tapp[2][\tickA]{#2\,[#1] }
\newcommand\tappc[1]{\tapp[\tickc]{#1}}
\newcommand{\tickA}{\alpha}
\newcommand{\tickB}{\beta}
\newcommand\latbind[2]{{\triangleright}\, (#1:#2) .}
\newcommand\toksubst[3][\kappa]{\left[#2/#3\right]}

\newcommand{\clocktype}{\mathrm{Clock}}

\newcommand\Str{\sym{Str}}
\newcommand\gStr[1][\kappa]{\sym{Str}^{#1}}

\newcommand{\cons}[2]{\mathrel{#1::#2}}

\newcommand{\basicsub}[2]{[#1\mapsto #2]}
\newcommand{\subex}[3]{#1\basicsub{#2}{#3}}

\newcommand{\cirr}[1][\kappa]{\sym{cirr}^{#1}}
\newcommand{\tirr}[1][\kappa]{\sym{tirr}^{#1}}
\newcommand{\pfix}[1][\kappa]{\sym{pfix}^{#1}}

\newcommand{\nats}{\mathbb{N}}

\newcommand{\ru}[2]{\dfrac{\begin{array}[b]{@{}c@{}} #1 \end{array}}{#2}}

\newcommand{\of}{{:}} 

\newcommand\hastype[4][\Delta]{
#2 \vdash_{#1} #3: #4
}
\newcommand\wfcxt[2][\Delta]{#2 \vdash_{#1}}
\newcommand\istype[3][\Delta]{
\ensuremath{#2 \vdash_{#1} #3 \, \operatorname{type}}
}
\newcommand\istypeshort[3][\Delta]{
\ensuremath{#2 \vdash_{#1} #3}
}
\newcommand\jeqjud[5][\Delta]{
#2 \vdash_{#1} #3 \jeq #4: #5
}

\newcommand{\jeq}{\equiv}
\newcommand{\jeqty}{\equiv}

\newcommand{\defeq}{\mathbin{\overset{\textsf{def}}{=}}}

\newcommand{\clott}{\textsf{CloTT}}
\newcommand{\gdtt}{\textsf{GDTT}}

\begin{document}

\maketitle
\abstract{
Clocked Type Theory (\clott) is a type theory for guarded recursion useful for programming with
coinductive types, allowing productivity to be encoded in types, and for reasoning about advanced 
programming language features using an abstract form of step-indexing. \clott\ has previously
been shown to enjoy a number of 
syntactic properties including strong normalisation, canonicity and decidability of type checking.
In this paper we present a denotational semantics for \clott\ useful, e.g., for studying future extensions 
of \clott\ with constructions such as path types. 

The main challenge for constructing this model is to model the notion of ticks used in \clott\ for coinductive
reasoning about coinductive types. We build on a category previously used to model guarded recursion, but in
this category there is no object of ticks, so tick-assumptions in a context can not be modelled using standard tools.
Instead we show how ticks can be modelled using adjoint functors, and how to model the tick constant
using a semantic substitution. 
}

\section{Introduction}

In recent years a number of extensions of Martin-L{\"o}f type theory~\cite{MartinLof:84} have been proposed to 
enhance the expressiveness or usability of the type theory. The most famous of these is Homotopy Type Theory~\cite{hottbook},
but other directions include the related Cubical Type Theory~\cite{CTT}, 
FreshMLTT~\cite{FreshMLTT}, a type theory with name abstraction
based on nominal sets, and Type Theory in Color~\cite{bernardy2015presheaf} 
for internalising relational parametricity in type theory. Many
of these extensions use denotational semantics to argue for consistency and to inspire constructions in the language.

This paper is part of a project to extend type theory with guarded recursion~\cite{Nakano:Modality}, a variant of 
recursion that uses a modal type operator $\later$ (pronounced `later') 
to preserve consistency of the logical reading of type theory.
The type $\later A$ should be read as classifying data of type $A$ available one time step from now, and comes with
a map $\nxt : A \to \later A$ and a fixed point operator mapping a function $f : \later A \to A$ to a fixed point for 
$f \circ \nxt$. This, in combination with \emph{guarded recursive types}, i.e., types where the recursive variable
is guarded by a $\later$, e.g., $\gStr[] \jeqty \nats \times \later \gStr[]$ gives a powerful type theory in which 
operational models of combinations of advanced programming language features such as higher-order
store~\cite{Birkedal-et-al:topos-of-trees} and nondeterminism~\cite{Bizjak-et-al:countable-nondet-internal} can be modelled
using an abstract form of step-indexing~\cite{Appel:M01}. 
Combining this with a notion of clocks, indexing the $\later$ operator with
clock names, and universal quantification over clocks, one can encode coinduction using guarded recursion,
allowing productivity~\cite{coquand1993infinite} of coinductive definitions to be encoded in types~\cite{atkey13icfp}.

The most recent type theory with guarded recursion is Clocked Type Theory (\clott)~\cite{bahr2017clocks}, 
which introduces the notion of ticks 
on a clock. Ticks are evidence that time has passed and can be used to unpack elements of type $\later A$ to elements
of $A$. In fact, in \clott, $\later A$ is a special form of function type from ticks to $A$. The combination of ticks and clocks
in \clott\ can be used for coinductive reasoning about coinductive types, by encoding the \emph{delayed substitutions} 
of~\cite{GDTT}. 

Bahr et al~\cite{bahr2017clocks} have shown that \clott\ can be given a reduction semantics satisfying 
strong normalisation, confluence and
canonicity. This establishes that productivity can indeed be encoded in types: For a closed term $t$ of stream type, 
the $n$'th element can be computed in finite time. These syntactic results also imply soundness of the type theory. 
However, these results have only been established for a core type theory without, e.g., identity types, and
the arguments can be difficult to extend to larger calculi. In particular, we are interested in extending \clott\ with path
types as in Guarded Cubical Type Theory~\cite{GCTT}. Therefore a denotational model of \clott\ 
can be useful, and this paper presents such a model. 

The work presented here builds on a number of existing models for guarded recursion. The most basic such, modelling the
single clock case, is the topos of trees model~\cite{Birkedal-et-al:topos-of-trees}, 
in which a closed type is modelled as a family of sets 
$X_n$ indexed by natural numbers $n$, together with restriction maps of the form $X_{n+1} \to X_n$ for every $n$. In other
words, a type is a presheaf over the ordered natural numbers. In this model $\later$ is modelled as $(\later X)_0 = 1$ and 
$(\later X)_{n+1} = X_n$ and guarded recursion reduces to natural number recursion. The guarded recursive 
type $\gStr[]$ mentioned above can be modelled in the topos of trees as $\gStr[](n) = 1\times \nats^n$. 

Bizjak and M{\o}gelberg~\cite{GDTTmodel} recently extended this model to the case of many clocks, using a category
$\set^\catT$ of covariant presheaves over a category $\catT$ of time objects, i.e., pairs of a finite set $X$ and a map 
$X \to \nats$. In this model, universal quantification over clocks is modelled by constructing an object in the topos of 
trees and taking the limit of that. For example, taking the limits over the object $\gStr[]$ gives the usual coinductive type 
of streams over natural numbers. 

The main challenge when adapting the model of~\cite{GDTTmodel} to \clott\ is to model ticks, which were not present
in the language modelled in~\cite{GDTTmodel}. In particular, how does one model tick assumptions of the form
$\tickA : \kappa$ in a context, when there appears to be no object of ticks in the model to be used as the
denotation of the clock $\kappa$. In this paper we observe that these assumptions can be modelled using a left adjoint
$\tearlier^\kappa$ to the functor $\tlater^\kappa$ used in~\cite{GDTTmodel} to model $\later^\kappa$ the delay modality
associated to the clock $\kappa$. Precisely we model context extension as $\pret{\Gamma, \tickA : \kappa} = 
\tearlier^\kappa\pret\Gamma$. To clarify what is needed to model ticks, we focus on a fragment of \clott\ called
the \emph{tick calculus} capturing just the interaction of ticks with dependent types. We show that the tick calculus
can be modelled soundly in a category with family~\cite{dybjer1996} (a standard notion of model for dependent type theory), 
with an adjunction $\L \dashv \R$ of endofunctors on the underlying category, 
for which the right adjoint lifts to types and terms, and there is a natural transformation from $\L$ to the identity. This appears 
to be a general pattern seen also in the model of fresh name abstraction of FreshMLTT~\cite{PITTS201519}
and dependent path types
in cubical type theory~\cite{CTT}. 
Similarly challenging is how to model the special tick constant $\tickc$. Since there is no object of ticks, there is no element
corresponding to $\tickc$ either. Still, we shall see that there exists a semantic substitution of $\tickc$ for a tick variable
that can be used to model application of terms to $\tickc$.  

The paper is organised as follows: 
The tick calculus and its model theory are introduced in  
Section~\ref{section:delayed substitution and ticks}. Section~\ref{sec:clott} introduces \clott, omitting
guarded recursive types and universes, which we leave for future work. Section~\ref{sec:basic:model} presents the
basics of the model, in particular the presheaf category $\set^\catT$ and the adjunction $\tearlier^\kappa \dashv 
\tlater^\kappa$. The presence of ticks in contexts leads to a non-standard notion of substitutions, and we study the 
syntax and semantics of these in Section~\ref{sec:substitutions}. Sections~\ref{sec:forall} and~\ref{sec:tickc} 
extend the model with universal quantification over clocks and $\tickc$, respectively. Finally, Section~\ref{sec:cirr} verifies
the important clock irrelevance axiom, and Section~\ref{sec:conclusion} concludes and discusses future work.

\section{A tick calculus}
\label{section:delayed substitution and ticks}

Before introducing \clott\ we focus on a fragment to explain the notion of ticks and how to model these. To motivate 
ticks, consider the notion of applicative functor from functional programming~\cite{mcbride2008applicative}: 
a type former $\later$ with maps
$A \to \later A$ and $\later(A \to B) \to \later A \to \later B$ satisfying a number of equations that we shall not recall. 
These maps can be used 
for programming with the constructor $\later$, but for reasoning in a dependent type theory, one needs an extension of these
to dependent function types. 
For example, in guarded recursion one can prove a theorem $X$ by constructing a map $\later X \to X$ and taking its fixed
point in $X$. If the theorem is that a property holds for all elements in a type of guarded streams satisfying 
$\gStr[] \jeqty \nats \times \later \gStr[]$, then $X$ will be of the form $\dprod{xs}{\gStr[]}P$. To apply the 
(essentially coinductive) assumption of type $\later\dprod {xs}{\gStr[]}P$ to the tail of a stream, which has type
$\later\gStr[]$ we need an extension of the applicative functor action. 

What should the type of such an extension be? Given $a: \later A$ and $f: \later (\dprod xAB)$ 
the application of $f$ to $a$ should be something of the form $\later B\subst x{??}$. If we think of $\later$ as a delay,
intuitively $a$ is a value of type $A$ delayed by one time, and the $??$ should be the value delivered
by $a$ one time step from now. Ticks are evidence that time has passed, and they allow us to talk about values delivered 
in the future.

 
The \emph{tick calculus} is the extension of dependent type theory with the following four rules
\begin{gather*}
\ru{\Gamma\vdash}{\Gamma,\tickA\of \tick\vdash} \qquad
\ru{\Gamma,\tickA\of\tick\vdash A}{\Gamma\vdash \later(\tickA\of\tick)A}\qquad \ru{\Gamma,\tickA\of\tick\vdash  t:A}{\Gamma\vdash \lambda(\tickA\of\tick) t:\later(\tickA\of\tick)A} \qquad
\ru{\Gamma\vdash t: \later(\tickA\of\tick) A}{\Gamma,\tickB\of\tick, \Gamma'\vdash \tapp[\tickB] t : A[\tickB/\tickA]} 
\end{gather*}
An assumption of the form $\tickA \of \tick$ in a context is an assumption that one time step has passed, and $\tickA$ 
is the evidence of this. Variables on the right-hand side of such an assumption should be thought of as arriving one time step later than those on the left. Ticks can be abstracted in terms and types, so that the type constructor $\later$ now comes
with evidence that time has passed that can be used in its scope. The type $\later(\tickA\of\tick)A$ can be thought
of as a form of dependent function type over ticks, which we abbreviate to $\later A$ if $\tickA$ does not occur
free in $A$. 
The elimination rule states that if a term $t$ can be typed as $\later(\tickA\of\tick)A$ before the arrival 
of tick $\tickB$, $t$ can be opened using $\tickB$ to give something of type $A[\tickB/\tickA]$. Note that the 
causality restriction in the typing rule prevents a term like 
$\lambda x .  \lambda(\tickA\of\tick) . \tapp{\tapp x} : \later\later A \to \later A$ being well typed; a tick can only 
be used to unpack the same term once. The context $\Gamma'$ in the elimination rule ensures that typing rules are
closed under weakening, also for ticks. 
Note that the clock object $\tick$ is not a type.

The equality theory is likewise extended with the usual $\beta$ and $\eta$ rules: 
\begin{align*}
 \lambda(\tickA\of \tick) \tapp[\tickB]t & = t[\tickB/\tickA] & \lambda(\tickA\of\tick) (\tapp t) & = t
\end{align*}
As stated, the tick calculus should be understood as an extension of standard dependent type theory. In particular one 
can add dependent sums and function types with standard rules. Variables can be introduced from anywhere in 
the context, also past ticks. 

We can now type the dependent applicative structure as 
\begin{align*}
 \lambda( x \of A) \lambda(\tickA \of\tick) x & \,\of\, A \to \later A \\
 \lambda f  \lambda y  \lambda(\tickA \of\tick)  \tapp f(\tapp y) &\,\of\, 
 \later\left(\dprod xAB\right) \to \dprod y{\later A}{\later(\tickA \of \tick) . B\subst x{\tapp y}}
\end{align*}

For a small example on how ticks in combination with the fixed point operator 
$\dfix : (\later X \to X) \to \later X$ 
can be used to reason about guarded recursive data, let 
$\gStr[] \jeqty \nats \times \later \gStr[]$ be the type of guarded recursive streams mentioned above, and
suppose $x\of \Nat \vdash P(x)$ is a family to be thought of as a predicate on $\Nat$ (where 
$\cons x{xs}$ is the pairing of $x$ and $xs$). A lifting of $P$ to streams
would be another guarded recursive type $y\of \Str \vdash \hat{P}(y)$ satisfying 
$\hat{P}(\cons x{xs}) \jeqty P(x) \times \later (\alpha \of \tick) \hat{P}(xs\,[\alpha])$. If $p: \Pi(x\of\Nat) P$ is a proof of $P$ 
we would expect that also $\Pi(y\of \Str) \hat{P}$ can be proved, and indeed this can be done as follows. 
Consider first 
\begin{align*}
 f & : \later (\Pi(y\of \Str) \hat{P}) \to \Pi(y\of \Str) \hat{P} \\
 f\, q \, (\cons x{xs}) & \defeq \cons{p(x)}{\lambda(\tickA \of \tick)\tapp[\tickA]q(\tapp[\tickA]{xs})}
\end{align*}
Then $f(\dfix(f))$ has the desired type. 

More generally, ticks can be used to encode \cite{bahr2017clocks} the \emph{delayed substitutions} of \cite{GDTT}, 
which have been used to reason coinductively about coinductive data. For more examples of reasoning using these 
see~\cite{GDTT}. For reasons of space, we will not model general guarded recursive types in this paper, but see 
Section~\ref{sec:basic:model} for how to model the types used above. 
%
%

\subsection{Modelling ticks using adjunctions}
\label{externalcwf}


We now describe a notion of model for the tick calculus. It is based on the notion of category with families (CwF) \cite{dybjer1996},
which is a standard notion of model of dependent type theory. Recall that a CwF is a pair $(\cat{C},T)$ such that 
$\cat C$ is a category with a distinguished terminal object and $T : \opcat{\cat C} \to \Fam(\set)$ is a functor
together with a comprehension map to be recalled below. The functor $T$ associates
to every object $\Gamma$ in $\cat C$ a map $T(\Gamma) : \trm(\Gamma) \to \typ(\Gamma)$ and to every morphism
$\gamma : \Delta \to \Gamma$ maps $\typ(\gamma): \typ(\Gamma) \to \typ(\Delta)$ and 
$\trm(\gamma): \trm(\Gamma) \to \trm(\Delta)$ such that $T(\Delta) \circ \trm(\gamma) = \typ(\gamma) \circ T(\Gamma)$. 
Following standard conventions, we write $\cwfty\Gamma A$ to 
mean $A \in \typ(\Gamma)$ and $\cwftm \Gamma tA$ to mean $t \in T(\Gamma)^{-1}(A)$, and we write 
$\cwfty\Delta{A[\gamma]}$ for $\typ(\gamma)(A)$ when $\cwfty\Gamma A$, 
and likewise $\cwftm\Delta{t[\gamma]}{A[\gamma]}$ for $\trm(\gamma)(t)$ when $\cwftm \Gamma tA$. 
We refer to the objects of $\cat C$ as contexts, morphisms as substitutions, elements of
$\typ(\Gamma)$ as types and elements of $\trm(\Gamma)$ as terms. 

Comprehension associates to each $\cwfty{\Gamma}A$ a context $\compr \Gamma A$, a substitution 
$\p_A : \compr \Gamma A \to \Gamma$ and a term $\cwftm{\compr\Gamma A}{\q_A}{A[\p_A]}$, such that for every
$\gamma : \Delta \to \Gamma$, and $\cwftm\Delta t{A[\gamma]}$ there exists a unique substitution 
$\cpair \gamma t : \Delta \to \compr \Gamma A$ such that  $\p_A \circ\cpair \gamma t = \gamma$ and 
$\q_A[\cpair \gamma t] = t$. 

To model the tick calculus we need an operation $\L$ modelling the extension of a context with a tick, plus an 
operation $\R$ modelling
$\later$. In the simply typed setting, $\R$ would be a right adjoint to context extension, but for dependent types this is 
not quite so, since these operations work on different objects (contexts and types respectively). In the model we consider
in this paper, the right adjoint does exist as an operation on contexts, but also extends to types and terms in the sense of the 
following definition.
\begin{definition}
\label{def:cwfa}
 Let $(\cat{C},T)$ be a CwF and let $\R : \cat C \to \cat C$ be a functor. An \emph{extension of $\R$ to types and terms}
 is a pair of operations on types and term presented here in the form of rules
 \begin{gather*}
\ru{\cwfty\Gamma A}{\cwfty{\R \Gamma}{\R A}} \qquad \ru{\cwftm\Gamma tA}{\cwftm{\R \Gamma}{\R t}{\R A}} 
\end{gather*}
commuting with substitutions in the sense that $(\R A)[\R \gamma] = \R (A[\gamma])$ and $(\R t)[\R \gamma] = \R (t[\gamma])$
hold for all substitutions $\gamma$, and commuting with comprehension in the sense that there exists an operation
associating to each $\cwfty\Gamma A$ a morphism $\isocomp{\Gamma}{A}:\R\Gamma.\R A \rightarrow \R(\Gamma.A)$ 
in $\cat C$ inverse to $\la \R \p_A, \R \q_A\ra$.
A \emph{CwF with adjunction} is a pair of adjoint endofunctors 
$\L \dashv \R : \cat C \to \cat C$ with an extension of $\R$ to types and terms. 
\end{definition}

Given a CwF with adjunction, one can define an operation mapping types $\cwfty{\L\Gamma}A$ to types 
$\cwfty\Gamma{\R_\Gamma A}$ defined as $\R_\Gamma A = (\R A)[\eta]$ where $\eta$ is the unit of the 
adjunction. 

\begin{lemma}
\label{lem:bijectivecorresp}
 There is a bijective correspondence between terms $\cwftm{\L\Gamma}aA$ and terms $\cwftm{\Gamma}b{\R_\Gamma A}$
 for which we write $\transp{(-)}$ for both directions where $\cwftm{\Gamma}{\transp{a}}{\R_\Gamma A}$ is given by $\transp{a} = (\R a)[\eta]$ and $\cwftm{\L \Gamma}{\transp{b}}A$ is given by $\transp{b} = \q_A[\epsilon \circ \L (\isocomp{\L\Gamma}{A} \circ \la \eta,b\ra)]$. Moreover, if $\gamma : \Delta \to \Gamma$, $\cwftm{\L\Gamma}aA$ 
 and $\cwftm{\Gamma}b{\R_\Gamma A}$ then 
\begin{align*}
(\R_\Gamma A)[\gamma] & = R_\Delta (A[\L \gamma]) & \transp{a[\L \gamma]} & = \transp{a}[\gamma] & 
\transp{b[\gamma]} & = \transp{b}[\L\gamma]
\end{align*}
\end{lemma}

%
%
%
%
%
%

\subsection{Interpretation}
\label{sec:tick:calc:interp}

The notion of CwF with adjunction is almost sufficient for modelling the tick calculus, but to interpret tick weakening, we
will assume given a natural transformation $\p_\L : \L \to \id_{\cat C}$. Defining 
\[
\pret{\Gamma,\alpha:\tick\vdash} = \L \pret{\Gamma} 
\]
$\p_\L$ allows us to define a context projection $\p_{\Gamma'} : \pret{\Gamma, \Gamma'\vdash} \to \pret{\Gamma\vdash}$ 
by induction on $\Gamma'$ using $\p_\L$ in the case of tick variables. We can then define the rest of the interpretation as
\begin{align*}
 \pret{\Gamma, x : A, \Gamma' \vdash x : A} & = \q_A[\p_{\Gamma'}] & 
 \pret{\Gamma\vdash \later(\alpha\of \tick) A} & = \R_\pret{\Gamma}\pret{A} \\
\pret{\Gamma\vdash \lambda(\alpha\of \tick) t} & =  \transp{\pret{t}} & 
\pret{\Gamma,\alpha'\of\tick,\Gamma'\vdash \tapp[\tickA']t} & =  \transp{\pret{t}}[ \p_{\pret{\Gamma'}}]
\end{align*}

\begin{proposition}
 The above interpretation of the tick calculus into a CwF with adjunction and tick weakening $\p_\L$ is sound. 
\end{proposition}

%
%
%
%
%

\section{Clocked Type Theory}
\label{sec:clott}

Clocked Type Theory (\clott) is an extension of the tick calculus with guarded recursion and multiple clocks. Rather than
having a global notion of time as in the tick calculus, ticks are associated with clocks and clocks can be assumed and 
universally quantified. Judgements have a separate context of clock variables $\Delta$, for example, the typing judgement has
the form $\hastype{\Gamma}tA$, where $\Delta$ is a set of clock variables $\kappa_1,\dots, \kappa_n$. The clock context
can be thought of as a context of assumptions of the form $\kappa_1 : \clocktype, \dots, \kappa_n : \clocktype$ 
that appear to the left of the assumptions of $\Gamma$, except that $\clocktype$ is not a type. There are no operations 
for forming clocks, only clock variables. It is often convenient to have a single clock constant $\kappa_0$ and this
can be added by working in a context of a single clock variable.  

\begin{figure*}[tbp]
\begin{center}
\textbf{Type formation rules} 
\begin{mathpar}
  \inferrule*
  {\istype{\Gamma,\tickA:\kappa}{A}\\ \kappa \in \Delta}
  {\istype{\Gamma}{\latbind{\tickA}{\kappa} A}}
  \and
  \inferrule*
  {\istype[\Delta, \kappa]{\Gamma}{A}\\
    \wfcxt{\Gamma}}
  {\istype{\Gamma}{\forall \kappa . A}}
\end{mathpar}
\textbf{Typing rules}
\begin{mathpar}
  \inferrule*
  {\hastype[\Delta, \kappa]{\Gamma}{t}{A}\\
    \wfcxt{\Gamma}}
  {\hastype{\Gamma}{\Lambda\kappa. t}{\forall \kappa . A}}
  \and
  \inferrule*
  {\hastype{\Gamma}{t}{\forall \kappa . A}\\
    \kappa'\in\Delta}
  {\hastype{\Gamma}{t [\kappa']}{A \subst{\kappa}{\kappa'}}}
  \and
  \inferrule*
  {\hastype{\Gamma,\tickA:\kappa}{t}{A} \\ \kappa \in \Delta}
  {\hastype{\Gamma}{\tabs{\tickA}{\kappa} t}{\latbind{\tickA}{\kappa} A}}
  \and
  \inferrule*
  {\hastype{\Gamma}{t}{\latbind{\tickA}{\kappa} A}\\\wfcxt{\Gamma,\tickA':\kappa,\Gamma'}}
  {\hastype{\Gamma,\tickA': \kappa,\Gamma'}{\tapp[\tickA'] t}{A\toksubst{\tickA'}{\tickA}}}
  \and
  \inferrule*
  {\hastype[\Delta,\kappa]{\Gamma}{t}{\latbind{\tickA}{\kappa} A}\\ \wfcxt{\Gamma}\\\kappa'\in\Delta}
  {\hastype{\Gamma}{\tappc{(t\,\subst{\kappa}{\kappa'})}}{A\subst{\kappa}{\kappa'}\toksubst{\tickc}{\tickA}}} 
  \and
  \inferrule*
  {\hastype{\Gamma}{t}{\later^\kappa A \to A}}
  {\hastype{\Gamma}{\dfix^\kappa\,t}{\later^\kappa A}}
\end{mathpar}
\textbf{Judgemental equality}
  \begin{align*}
    (\Lambda\kappa.t)[\kappa']&\jeq  t \subst{\kappa'}\kappa &
    (\Lambda \kappa. t [\kappa]) &\jeq t & 
   \tapp{(\tabs{\tickA'}{\kappa} t)} &\jeq t\toksubst{\tickA}{\tickA'} \\ 
    \tabs{\tickA}{\kappa} (\tapp t) &\jeq t & 
    \tappc{(\dfix^\kappa \, t)} &\jeq t\,(\dfix^\kappa\,t) &
  \end{align*}
\caption{Selected typing and judgemental equality rules of Clocked Type Theory.}
\label{fig:clott:typing}
\end{center}
\end{figure*}

The rules for typing judgements and judgemental equality are given in Figure~\ref{fig:clott:typing}. These should 
be seen as an extension of a dependent type theory with dependent function and sum types, as well as 
extensional identity types. The rules for these are completely standard (ignoring the clock context), and thus are omitted from
the figure. We write $\jeq$ for judgemental equality and $\idty Atu$ for identity types. The model will also model the 
\emph{identity reflection} rule
\[
\inferrule*{\hastype{\Gamma}{p}{\idty Atu}}
{\jeqjud{\Gamma}tuA}
\]
of extensional type theory. 

The guarded fixed point operator $\dfix$ is useful in combination with guarded recursive types. Suppose for example that
we have a type of natural numbers $\nats$ and a type of guarded recursive streams $\gStr$ satisfying 
$\gStr \jeqty \nats \times \later^\kappa \gStr$. One can then
use $\dfix$ for recursive programming with guarded streams, e.g., when defining a constant stream of zeros as
$\dfix^\kappa (\lambda x. \pair 0x)$. The type of $\dfix$ ensures that only productive recursive definitions are typeable, e.g.,
$\dfix^\kappa (\lambda x . x)$ is not. 

The tick constant $\tickc$ gives a way to execute a delayed computation $t$ of type $\later^\kappa A$ to compute a value of
type $A$. In particular, if $t$ is a fixed point, application to the tick constant unfolds the fixed point once. This explains the need
to name ticks in \clott: substitution of $\tickc$ for a tick variable $\tickA$ in a term allows for all fixed points applied to $\tickA$ 
in the term to be unfolded. In particular, the names of ticks are crucial for the strong normalisation result for \clott\ in \cite{bahr2017clocks}. 

To ensure productivity, application of $\tickc$ must be restricted. In particular a term such as 
$\dfix^\kappa(\lambda x : \gStr . \tappc x)$ should not be well typed. The typing rule for application to the tick constant 
ensures this by assuming that the clock $\kappa$ associated to the delay is not free in the context of the term $t$. For example,
the rule 
\[
  \inferrule*
  {\hastype[\Delta,\kappa]{\Gamma}{t}{\latbind{\tickA}{\kappa} A}\\ \wfcxt{\Gamma}}
  {\hastype[\Delta,\kappa]{\Gamma}{\tappc{t}}{A\toksubst{\tickc}{\tickA}}} 
\]
is admissible, which can be proved using weakening lemma for the clock variable context. This rule, however, is not closed
under variable substitution, which is the motivation for the more general rule of Figure~\ref{fig:clott:typing}. The typing
rule is a bit unusual, in that it involves substitution in the term in the conclusion. We shall see in Section~\ref{sec:welldef}
that this causes extra proof obligations for welldefinedness of the denotational semantics.

Universal quantification over clocks allow for coinductive types to be encoded using guarded recursive types~\cite{atkey13icfp}. 
For example
$\Str \defeq \forall\kappa . \gStr$ is a coinductive type of streams. The head and tail maps $\sym{hd} : \Str \to \nats$ and 
$\sym{tl} : \Str \to \Str$ can be defined as
\begin{align*}
 \sym{hd}(xs) & \defeq \pi_1(xs[\kappa_0]) &
 \sym{tl}(xs) & \defeq  \Lambda \kappa . (\tappc{(\pi_2(xs[\kappa]))})
\end{align*}
using the clock constant $\kappa_0$. 

Finally we recall the \emph{clock irrelevance axiom}
\begin{equation} \label{eq:cirr}
\inferrule*
{\hastype{\Gamma}{t}{\forall\kappa . A} \\ \istype{\Gamma}A}
{\hastype\Gamma{\cirr t}{\forall{\kappa'} . \forall{\kappa''} . \idty{A}{t [\kappa']}{t [\kappa'']}}}
\end{equation}
crucial for correctness of the encoding of coinductive types~\cite{atkey13icfp}. Note that the hypothesis implies that 
$\kappa$ is not free  in $A$. This rule can be used to prove that $\forall\kappa . A$ is isomorphic 
to $A$ if $\kappa$ is not free in $A$. Likewise the \emph{tick irrelevance axiom} 
\begin{equation}\label{eq:tirr}
\inferrule*
{\hastype{\Gamma}{t}{\later^\kappa A}}
{\tirr t :  \latbind{\tickA}{\kappa}  \latbind{\tickA'}{\kappa} \idty{A}{\tapp{t}}{\tapp[\tickA'] t}}
\end{equation}
states that the identity of ticks is irrelevant for the equality theory, despite being crucial for the reduction semantics. 
Tick irrelevance implies fixed point unfolding
\[
\inferrule*
{\hastype[\Delta, \kappa]{\Gamma}{f}{\later^\kappa A \to A} \\ \wfcxt{\Gamma} \\ \kappa'\in \Delta}
{\hastype{\Gamma}{\pfix[\kappa'] \, f\subst{\kappa}{\kappa'}}{\latbind{\tickA}{\kappa}{\idty{A}{\tapp{(\dfix^{\kappa'} f\subst{\kappa}{\kappa'})}}{(f(\dfix^{\kappa'} f))\subst{\kappa}{\kappa'}}}}}
\]

The type theory \clott\ as defined in \cite{bahr2017clocks} also has guarded recursive types and a universe. 
We leave these for future work, see Section~\ref{sec:conclusion}.

\section{Presheaf semantics}
\label{sec:basic:model}

The setting for the denotational semantics of \clott\ is a category of covariant presheaves over a category $\catT$ of time objects.
This category has previously been used to give a model of \gdtt~\cite{GDTTmodel}. 

We will assume given a countably infinite set $\cv$ of (semantic) clock variables, for which we use $\lambda, \lambda',\dots$ to range over. 
A \emph{time object} is a pair $\timeobj\Theta\vartheta$ where $\Theta$ is a finite subset of $\cv$ and 
$\vartheta:\Theta\rightarrow \nats$ is a map giving the number of ticks left on each clock in $\Theta$. We will write the finite sets
$\Theta$ as lists writing e.g., $\Theta, \lambda$ for $\Theta \cup\{\lambda\}$ and $\subex\vartheta\lambda n$ for the extension
of $\vartheta$ to $\Theta,\lambda$, or indeed for the update of $\vartheta$, if $\vartheta$ is already defined on $\lambda$. 
A morphism $\timeobj\Theta\vartheta \to \timeobj{\Theta'}{\vartheta'}$ is a function $\tau:\Theta \rightarrow \Theta'$ such that 
$\vartheta'\tau\leq \vartheta$ in the pointwise order. The inequality allows for time to pass in a morphism, but morphisms can
also synchronise clocks in $\Theta$ by mapping them to the same clock in $\Theta'$, or introduce new clocks if $\tau$ is not
surjective. Define $\grtotal$ to be the category $\set^{\catT}$ of covariant presheaves
on $\catT$. The topos of trees can be seen as a restriction of this where time objects always have a single clock. 


%

The category $\grtotal$ contains a special object of clocks, given by the first projection $\clk \timeobj \Theta \vartheta = \Theta$. 
%
If $\Delta$ is a set, one can form the object $\clk^\Delta$ as 
$\clk^\Delta \timeobj\Theta\vartheta = \Theta^\Delta$. Let $\catT_\Delta$ be the category of elements of $\clk^\Delta$, i.e.,
the objects are triples $\triple\Theta\vartheta f$ where $\timeobj\Theta\vartheta \in \mathbb{T}$ and 
$f:\Delta \rightarrow \Theta$ and a morphism $\tau : \triple\Theta\vartheta f \rightarrow \triple{\Theta'}{\vartheta'}{f'}$ 
is  a morphism $\tau:\timeobj\Theta\vartheta \to \timeobj{\Theta'}{\vartheta'}$ such that $\tau\circ f = f'$. 
A clock context $\Delta$ will be
interpreted as $\clk^\Delta$ and contexts, types and terms in clock context $\Delta$ will be modelled in the category 
$\gr\Delta \defeq \set^{\catT_\Delta}$ of covariant presheaves over $\catT_\Delta$. If $F$ is a covariant presheaf  
over $\gr\Delta$ and $\tau : \triple\Theta\vartheta f \rightarrow \triple{\Theta'}{\vartheta'}{f'}$ and 
$x \in F\triple\Theta\vartheta f$ we will write $\tau\cdot x$ for $F(\tau)(x) \in F\triple{\Theta'}{\vartheta'}{f'}$. 
%
%
%

To describe the model of \clott, we start by fixing a clock context $\Delta$ and modelling the fragment of \clott\ excluding 
universal quantification over clocks and the tick constant $\tickc$. The resulting fragment is a version of the tick calculus
with one notion of tick for each clock $\kappa$ in $\Delta$. To model this, we need the structure of a CwF with adjunction on $\gr\Delta$ for each $\kappa$ in $\Delta$. Recall that, like any presheaf category, $\gr\Delta$ can be equipped with 
the structure of a CwF where contexts are objects, types in context $\Gamma$ are presheaves over the elements of $\Gamma$
and terms are sections. Precisely, a type over $\Gamma$ is a mapping associating a set $A(\gamma)$ to each 
$\gamma \in \Gamma \triple\Theta\vartheta f$ and to each 
$\tau : \triple\Theta\vartheta f \rightarrow \triple{\Theta'}{\vartheta'}{f'}$ a mapping 
$\tau\cdot(-) : A(\gamma) \to A(\tau\cdot\gamma)$ such that $\id\cdot x = x$ and $(\rho\tau)\cdot x = \rho\cdot(\tau\cdot x)$ 
for all $x, \tau$ and $\rho$. A term is a mapping associating to each $\gamma$ an element $t(\gamma) \in A(\gamma)$ such
that $t(\tau\cdot \gamma) = \tau\cdot t(\gamma)$. We often make the underlying $\catT_\Delta$ object explicit writing
$t_{\triple\Theta\vartheta f}(\gamma)$.  

As an example of a model of a type, recall the type of guarded streams satisfying 
$\gStr[\kappa] \jeqty \nats \times \later \gStr[\kappa]$ from Section~\ref{sec:clott}.
This is a closed type in a clock context $\Delta$ (assuming $\kappa \in \Delta$), and so will be interpreted as a presheaf 
in $\gr\Delta$ defined as
$\pret{\Str^\kappa} \triple \Theta \vartheta f =\nats^{\vartheta(f(\kappa))+1} \times \{\ast\}$.
We will assume that the products in this associate to the right, so that this is the type of tuples of the form
$(n_{\vartheta(f(\kappa))},(\dots, (n_0, \ast))\dots)$. This is needed to model the equality 
$\gStr[\kappa] \jeqty \nats \times \later \gStr[\kappa]$, rather than just an isomorphism of types.
Given a predicate $x\of \Nat\vdash P$, we can lift it to a predicate 
$y:\gStr[\kappa] \vdash \hat{P}$ satisfying $\hat{P}(x:xs) \jeqty P(x) \times \latbind \tickA\kappa{\hat{P}(xs\,[\tickA])}$ 
as in Section~\ref{section:delayed substitution and ticks}, by defining 
\[
\pret{\hat{P}}_{\triple \Theta \vartheta f}(n_{\vartheta(f(\kappa))},(\dots, (n_0, \ast))\dots)
= \{(x_{\vartheta(f(\kappa))},(\dots, (x_0, \ast))\dots) \mid \forall i . x_i \in \pret P_{\triple \Theta \vartheta f}(n_i) \}
\]
It is a simple calculation (using the definitions below) that these interpretations model the type equalities mentioned above.

\subsection{Adjunction structure on $\gr\Delta$}

For the adjunction, recall that in the topos of trees the functor $\tlater$ is defined as 
$(\tlater F)(n + 1) = Fn$ and $(\tlater F)(0) = \{\ast\}$. This has a left adjoint $\tearlier$ defined as $(\tearlier F) n = F (n+1)$. 
The right adjoint generalises in a straight forward way to
the multiclock setting of \clott: If $F$ is in $\gr\Delta$, define
\begin{align*}
& (\tlater^\kappa F) \triple\Theta\vartheta f = 
\begin{cases} 
F \triple{\Theta}{\vartheta[f(\kappa)-]}f &\vartheta(f(\kappa))>0\\
\{\ast\}& \text{otherwise}
\end{cases}
\end{align*}
where $\vartheta[f(\kappa)-] (f(\kappa))= \vartheta(f(\kappa)) - 1$ and $\vartheta[f(\kappa)-](\lambda) = \vartheta(\lambda)$ 
for $\lambda \neq f(\kappa)$.
This is the same definition as used in the \gdtt\ model of~\cite{GDTTmodel}.

\begin{lemma} \label{lem:later:types:terms}
 The functor $\tlater^\kappa$ extends to types and terms.
%
\end{lemma}
\begin{proof}
We just give the definitions. 
For $\gamma \in (\tlater^\kappa \pret{\Gamma}) \triple \Theta \vartheta f$  define
\[(\tlater^\kappa \pret{A})_{\triple \Theta \vartheta f}(\gamma) = 
            \begin{cases} 
                      \{\ast\} & \vartheta(f(\kappa)) = 0\\
                      \pret{A}_{\triple \Theta {\vartheta[f(\kappa)-]} f}(\gamma) & \text{otherwise}
                      \end{cases}
                      \]
The case for terms is similar. 

The isomorphism $\isocomp{\Gamma}{A}$ is given by
$
{\isocomp{\Gamma}{A}}_{\triple \Theta \vartheta f} = 
\begin{cases}
\la \ast,\ast\ra \mapsto \ast & \vartheta(f(\kappa)) = 0\\
\id & \text{otherwise}
\end{cases}
$
\end{proof}
At first sight it would seem that one can define a left adjoint to the above functor given by 
$(\tearlier^\kappa F)\triple{\Theta}{\vartheta}{f} = F\triple\Theta{\vartheta[f(\kappa) +]}{f}$, 
where $\vartheta[f(\kappa) +]$ is defined similarly to
$\vartheta[f(\kappa)-]$.
Unfortunately, $\tearlier^\kappa F$ so described is not a presheaf because it has no well-defined action on maps since a map 
$\tau: \triple{\Theta}{\vartheta}f \rightarrow \triple{\Theta'}{\vartheta'}{f'}$ does not necessarily induce a map 
$\triple{\Theta}{\vartheta[f(\kappa)+]}f \to \triple{\Theta'}{\vartheta'[f'(\kappa)+]}{f'}$: If 
$\tau(f(\kappa)) = \tau(\lambda)$ there is no guarantee that 
$\vartheta'[f'(\kappa)+](\tau(\lambda)) \leq \vartheta[f(\kappa)+](\lambda)$.

To get the correct description of the left adjoint consider the set $f^{-1}(f(\kappa)) \subseteq\Delta$ of syntactic
clocks synchronised with $\kappa$ by $f$. Given a morphism 
$\tau : \triple{\Theta}{\vartheta}{f} \to \triple{\Theta'}{\vartheta'}{f'}$, more clocks can be synchronised 
with $\kappa$ by $f'$ than $f$, but never fewer. If we think of time as flowing in the direction of morphisms, 
the left adjoint must take into account all the possible ways that $\kappa$ could have been
synchronised with fewer syntactic clocks ``in the past''. Such a past is given by a subset $X \subset f^{-1}(f(\kappa))$
such that $\kappa \in X$. 

%
%

\begin{lemma} \label{lem:earlier}
The functor $\tlater^\kappa$ has a left adjoint $\tearlier^\kappa$ given by
\begin{align*}
\tearlier^\kappa F \triple\Theta\vartheta f  & = \displaystyle\coprod_{\kappa\in X \subset f^{-1}(f(\kappa))} F \triple\Theta\vartheta f [X,\kappa+] 
\end{align*}
where 
$\triple\Theta\vartheta f[X,\kappa+] = \triple{\Theta,\#_\Theta}{\vartheta[\#_\Theta\mapsto \vartheta(f(\kappa))+1]}{f[X\mapsto \#_\Theta]}$
for $\#_\Theta$ a chosen clock name fresh for $\Theta$. 
\end{lemma}

Finally, the projection $\p_{\tearlier^\kappa} : \tearlier^\kappa \to \id$ maps an element
$\pair X\gamma$ in $\tearlier^\kappa F \triple\Theta\vartheta f$ to $\chi \cdot \gamma$ where 
\[
  \chi : \triple{\Theta}{\vartheta}f[X,\kappa+] \to \triple\Theta\vartheta f
\]
is defined as $\chi(\#_\Theta) = f(\kappa)$ and $\chi(\lambda) = \lambda$ for $\lambda\in \Theta$. 
Collectively, Lemmas~\ref{lem:later:types:terms} and \ref{lem:earlier} together with the projection $\p_{\tearlier^\kappa}$
state that for each $\kappa$, $\gr\Delta$ carries the structure of a model of the tick calculus. This is enough to model 
the tick abstractions and applications of \clott. 

The adjoint correspondent  $\transp{\p_{\tearlier^\kappa}}:\id \rightarrow \tlater^\kappa$ to $\p_{\tearlier^\kappa}$
maps an element $\gamma \in F\triple\Theta\vartheta f$ to its restriction in $F \triple\Theta {\vartheta[f(\kappa)-]} f$.
This is the map referred to as $\nxt$ in~\cite{GDTTmodel}. 
Moreover, a simple calculation shows that the interpretation of $\later^\kappa A$, i.e., $\latbind\tickA\kappa A$ for $\tickA$ not 
free in $A$, is the same as in~\cite{GDTTmodel}, namely
\[
\pret{\istype{\Gamma}{\later^\kappa A}}_{\triple{\Theta}\vartheta f}(\gamma) = 
\pret{A}_{\triple {\Theta}{\vartheta[f(\kappa)-]} f} (\gamma|_{\triple {\Theta}{\vartheta[f(\kappa)-]} f})
\]
We can thus define the interpretation of $\dfix$ as in~\cite{GDTTmodel} by induction on $\vartheta(f(\kappa))$:
\begin{align*}
\pret{\dfix\,t}_{\triple\Theta\vartheta f}(\gamma) = 
\begin{cases}
\ast & \vartheta(f(\kappa)) = 0\\
\pret{t\;(\dfix\;t)}_{\triple\Theta{\vartheta[f(\kappa)-]}f}(\gamma|_{\triple\Theta{\vartheta[f(\kappa)-]}f}) & \text{Otherwise}
\end{cases}
\end{align*}

Finally we note soundness of the tick irrelevance axiom (\ref{eq:tirr}).

\begin{proposition}
\label{prop:tirr}
If $\hastype{\Gamma}{t}{\later^\kappa A}$ then 
$\pret{\hastype{\Gamma, \tickA : \kappa, \tickA' : \kappa}{\tapp t}{A}} =
\pret{\hastype{\Gamma, \tickA : \kappa, \tickA' : \kappa}{\tapp[\tickA']t}{A}}$
\end{proposition}

%

\section{Substitutions}
\label{sec:substitutions}

Having described the interpretation of the fragment of \clott\ that lives within a fixed clock context $\Delta$ it remains to 
describe the interpretation of the universal quantification over clocks and of the tick constant $\tickc$. 
Quantification over clocks can be seen
as a dependent product over a type of clocks, and should therefore be modelled as a right adjoint to weakening in the
clock context. Weakening is an example of a substitution and, as we shall see, the tick constant $\tickc$ will
also be modelled using a form of substitution. We therefore first study substitutions, which are non-standard in \clott\
because of the two contexts, and because of the unusual typing rules for ticks. 

\subsection{Syntactic substitutions}

A syntactic substitution from $\wfcxt{\Gamma}$ to $\wfcxt[\Delta']{\Gamma'}$  is a pair $\pair\nu\sigma$
of a substitution $\nu$ of clocks for clocks and a substitution $\sigma$ of terms for variables and ticks for ticks variables. 
Substitutions are formed according to the following rules. 
\begin{itemize}
\item If $\nu : \Delta' \to \Delta$ is a map of sets, then $\pair\nu\cdot : \wfcxt{\Gamma} \to \wfcxt[\Delta']{\cdot}$
\item If $\pair\nu\sigma : \wfcxt{\Gamma} \to \wfcxt[\Delta']{\Gamma'}$ and 
$\hastype\Gamma{t}{A\pair\nu\sigma}$ then 
$\pair\nu{\subex\sigma xt} : \wfcxt{\Gamma} \to \wfcxt[\Delta']{\Gamma', x : A}$
\item If $\pair\nu\sigma : \wfcxt{\Gamma} \to \wfcxt[\Delta']{\Gamma'}$ and $\wfcxt{\Gamma, \tickA : \nu(\kappa), \Gamma''}$ 
and $\wfcxt[\Delta']{\Gamma', \tickB : \kappa}$ are welformed then 
$\pair\nu{\subex\sigma\tickB\tickA} :\wfcxt{\Gamma, \tickA : \nu(\kappa), \Gamma''} \to
\wfcxt[\Delta']{\Gamma', \tickB : \kappa}$
\item If $\pair\nu\sigma : \wfcxt\Gamma \to \wfcxt[\Delta']{\Gamma'}$, $\kappa \notin \Delta'$ and $\kappa' \in \Delta$, then
\[
\pair{\subex\nu\kappa{\nu(\kappa')}}{\subex\sigma{\tickA}\tickc} : \wfcxt\Gamma \to \wfcxt[\Delta',\kappa]{\Gamma',\tickA : \kappa}
\]
\end{itemize}
Here $A\pair\nu\sigma$ is the result of substituting $A$ along $\pair\nu\sigma$ which is defined in the standard way. 

\subsection{Semantic substitutions}
\label{sec:semantic:subst}


The clock substitution $\nu$ gives rise to a functor $\catT_{\Delta} \to \catT_{\Delta'}$ mapping an object 
$\triple\Theta\vartheta f$ to $\triple\Theta\vartheta{f\nu}$, and this induces a functor $\nu^* : \gr{\Delta'} \to \gr\Delta$
by $(\nu^*F)\triple\Theta\vartheta f = F\triple\Theta\vartheta{f\nu}$. This functor extends to a morphism of 
CwFs~\cite{dybjer1996}, in particular it maps a type $A$ in context $\Gamma$ in 
the CwF structure of $\gr{\Delta'}$ to a type $\nu^*A$ in context $\nu^*\Gamma$ the CwF structure of $\gr{\Delta}$, and
likewise for terms. For example, $(\nu^*A)(\gamma)$ for $\gamma\in \nu^*\Gamma\triple\Theta\vartheta f =
\Gamma\triple\Theta\vartheta{f\nu}$ is defined as $A\gamma$. 
Moreover, this map commutes on the nose with comprehension and substitution. For example,
if $A$ is a type in context $\Gamma$ in $\gr{\Delta'}$ and $\gamma : \Gamma' \to \Gamma$, then 
$(\nu^*A) [\nu^*\gamma] = \nu^*(A[\gamma])$. Moreover, it commutes with $\tlater$ in the following sense.

\begin{lemma}
 If $\nu : \Delta'\to \Delta$ and $\kappa\in \Delta'$ then $\nu^*\circ\tlater^\kappa = \tlater^{\nu(\kappa)}\circ \nu^*$.
\end{lemma}

The interpretation of a substitution $\pair\nu\sigma$ is a morphism
\[
\pret{\pair\nu\sigma} : \pret{\wfcxt\Gamma} \to \nu^*\pret{\wfcxt[\Delta']\Gamma'}
\]
in $\gr{\Delta}$, which we will define below. 
But first we state the substitution lemma, which must be proved by induction on terms and types
simultaneously with the definition of the interpretation, as is standard for models of dependent type
theory.  

%
%

\begin{lemma} \label{lem:substitution}
Let $\pair{\nu}{\sigma} : \wfcxt[\Delta]{\Gamma} \to \wfcxt[\Delta']{\Gamma'}$ be a substitution and let
${\Gamma'}\vdash_{\Delta'}  J$ be a judgement of a wellformed type or a typing judgement. Then
$\pret{ J\pair\nu\sigma} = (\nu^*\pret{ J})[\pret{\pair\nu\sigma}]$. 
\end{lemma}

The main difficulty for defining the interpretation of substitutions is that 
the operator $\tearlier^\kappa$ does not commute with clock substitutions in the sense that $\nu^* (\tearlier^\kappa \Gamma)$ is not necessarily equal to $\tearlier^{\nu(\kappa)} (\nu^* \Gamma)$.
However, we can define a map in the relevant direction:
\begin{align*}
e^{\kappa, \nu}_\Gamma  = \epsilon_{\nu^*(\tearlier^\kappa \Gamma)}^{\nu(\kappa)} \circ \tearlier^{\nu(\kappa)} \nu^*(\eta_\Gamma^\kappa) 
&\, :\, \tearlier^{\nu(\kappa)} (\nu^*\Gamma) \to \nu^*(\tearlier^\kappa \Gamma)
\end{align*}
where $\eta_\Gamma^\kappa : \Gamma \to \tlater^\kappa\tearlier^\kappa \Gamma$ is the unit of the adjunction and
\[\epsilon_{\nu^*(\tearlier^\kappa \Gamma)}^{\nu(\kappa)} \, : \,
\tearlier^{\nu(\kappa)}\tlater^{\nu(\kappa)}\nu^*(\tearlier^\kappa \Gamma) \to \nu^*(\tearlier^\kappa \Gamma)\]
is the counit. The composition type checks since 
$\nu^*\tlater^\kappa = \tlater^{\nu(\kappa)} \nu^*$.

The map $e^{\kappa, \nu}$ has a simple description in the model: $e^{\kappa, \nu}_\Gamma$ maps an element $(X,\gamma)$ 
in
\begin{align*}
 \tearlier^{\nu(\kappa)} (\nu^* \Gamma) \triple \Theta \vartheta f 
 &= \coprod_{\substack{\nu(k)\in X \\ f(X) = f(\nu(\kappa))}} (\nu^* \Gamma) \triple {\Theta,\#_\Theta} {\vartheta[\#_\Theta \mapsto \vartheta((f\circ \nu)(\kappa))+1]} {f[X\mapsto \#_\Theta]}\\ 
&= \coprod_{\substack{\nu(k)\in X \\ f(X) = f(\nu(\kappa))}} \Gamma \triple {\Theta,\#_\Theta} {\vartheta[\#_\Theta \mapsto \vartheta((f\circ \nu)(\kappa))+1]} {f\circ \nu [\inv{\nu} X\mapsto \#_\Theta]}
\end{align*}
to $(\inv{\nu} X,\gamma)$ in the set $\nu^*(\tearlier^\kappa\Gamma) \triple \Theta \vartheta f$ which equals
\begin{align*}
   \tearlier^\kappa\Gamma \triple \Theta \vartheta {f\circ \nu} 
&=\coprod_{\substack{\kappa\in Y\\ f(\nu(Y)) = f(\nu(\kappa))}} \Gamma \triple {\Theta,\#_\Theta} {\vartheta[\#_\Theta \mapsto \vartheta((f\circ \nu)(\kappa))+1]} {f\circ \nu [Y\mapsto \#_\Theta]}
\end{align*}

The interpretation of substitutions is defined as
\begin{align*}
  \pret\cdot & = x \mapsto \star \\
  \pret{\pair\nu{\subex\sigma xt}} & = \cpair{\pret{\pair\nu\sigma}}{\pret t} \\
  \pret{\pair\nu{\subex\sigma\tickB\tickA}} & = e^\kappa_{\pret{\Gamma'}}\circ\tearlier^{\nu(\kappa)}{\pret{\pair\nu\sigma}} \circ \p_{\Gamma''} \\
  \pret{\pair{\subex\nu\kappa{\nu(\kappa')}}{\subex\sigma{\tickA}\tickc}} 
  & = \nu^*{\pret{\pair{\basicsub{\kappa}{\kappa'}}{\basicsub{\tickA}{\tickc}}}}\circ \pret{\pair\nu\sigma}
\end{align*}
where we have assumed the types as in the rules for forming substitutions and $\p_{\Gamma''}$ is the context 
projection defined as in Section~\ref{sec:tick:calc:interp}. The last case uses
$\pret{\pair{\basicsub{\kappa}{\kappa'}}{\basicsub{\tickA}{\tickc}}}$ which will be defined in Section~\ref{sec:tickc} below.

The rest of this section is a sketch proof of the substitution lemma for the fragment of \clott\ modelled so far, i.e.,
excluding quantification over clocks and $\tickc$. As we extend the interpretation we will also extend the proof of the
substitution lemma. 

 The proof is by induction over judgements and is simultaneous with the definition of the interpretation of substitutions.
 The cases of standard dependent type theory (dependent functions and sums, identity types) can be essentially reduced to the 
 standard proof of the substitution lemma for dependent type theory in presheaf models as follows
 (although the non-standard notion of substitution requires some new lemmas). Since $\gr\Delta$ is
 the object of presheaves over $\catT_\Delta$, the category of elements of $\pret\Delta = \clk^\Delta$, given a context 
 $\wfcxt\Gamma$ one can form the comprehension $\compr{\pret\Delta}{\pret\Gamma}$ as an object of $\grtotal$. Types 
 and terms in context $\compr{\pret\Delta}{\pret\Gamma}$ in the CwF structure associated to $\grtotal$ are then
 in bijective correspondence with those over $\pret\Gamma$ in the CwF structure of $\gr\Delta$. A substitution
 $\pair\nu\sigma : \wfcxt\Gamma \to \wfcxt[\Delta']{\Gamma'}$ induces a morphism 
 $\compr{\pret\Delta}{\pret\Gamma} \to \compr{\pret{\Delta'}}{\pret{\Gamma'}}$ in $\grtotal$ and the substitution defined 
 above corresponds to substitution along this map. Thus the interpretation of the standard type theoretic constructions
 are the same as the standard ones in the presheaf model $\grtotal$, and the corresponding cases of the substitution lemma 
 can be proved similarly to the standard proof of the substitution lemma for presheaf models of type theory. 
  
 The cases corresponding to the constructions from the tick calculus follow from the following two equations.
\begin{align}
 \tlater^{\nu(\kappa)}e^{\kappa,\nu} \circ \eta^{\nu(\kappa)}_{\nu^*} & = \nu^*(\eta^\kappa) 
 && :  \nu^* \to \tlater^{\nu(\kappa)} \nu^* \tearlier^\kappa = \nu^* \tlater^\kappa \tearlier^\kappa \label{eq:e:comm:eta} \\
 \nu^*(\epsilon^\kappa) \circ e^{\kappa,\nu}_{\tlater^\kappa} & = \epsilon^{\nu(\kappa)}_{\nu^*}
 && :  \tearlier^{\nu(\kappa)}\nu^* \tlater^\kappa  = \tearlier^{\nu(\kappa)}\tlater^{\nu(\kappa)}\nu^* \to \nu^* \label{eq:e:comm:eps}
\end{align}

For example,  the case of $\latbind \tickA\kappa A$ can be proved as follows (writing $\pret\sigma$ for $\pret{\pair\nu\sigma}$)
\begin{align*}
 \nu^*\pret{\istypeshort[\Delta']{\Gamma'}{\latbind \tickA\kappa A}}[\pret\sigma] & = 
 \nu^*\left(\left(\tlater^{\kappa}\pret{\istypeshort[\Delta']{\Gamma', \tickA : \kappa}{A}}\right)[\eta^\kappa_{\pret{\Gamma'}}]\right)[\pret\sigma] \\
& = \left(\tlater^{\nu(\kappa)}\nu^*\pret{\istypeshort[\Delta']{\Gamma', \tickA : \kappa}{A}}\right)[\nu^*\eta^\kappa_{\pret{\Gamma'}}][\pret\sigma] \\
& = \left(\tlater^{\nu(\kappa)}\nu^*\pret{\istypeshort[\Delta']{\Gamma', \tickA : \kappa}{A}}\right)[\tlater^{\nu(\kappa)}e^\kappa \circ \eta^{\nu(\kappa)}_{\nu^*\pret{\Gamma'}} \circ \pret\sigma] \\
& = \left(\tlater^{\nu(\kappa)}\nu^*\pret{\istypeshort[\Delta']{\Gamma', \tickA : \kappa}{A}}\right)[\tlater^{\nu(\kappa)}(e^\kappa \circ \tearlier^{\nu(\kappa)}\pret\sigma)\circ \eta^{\nu(\kappa)}_{\nu^*\pret{\Gamma'}}] \\
& = \left(\tlater^{\nu(\kappa)}\left(\nu^*\pret{\istypeshort[\Delta']{\Gamma', \tickA : \kappa}{A}}\right)[e^\kappa \circ \tearlier^{\nu(\kappa)}\pret\sigma]\right) [\eta^{\nu(\kappa)}_{\nu^*\pret{\Gamma'}}] \\
& = \left(\tlater^{\nu(\kappa)}\pret{\istypeshort{\Gamma, \tickB : \nu(\kappa)}{A\pair{\nu}{\subex\sigma\tickA\tickB}}}\right) [\eta^{\nu(\kappa)}_{\nu^*\pret{\Gamma'}}] \\
& = \pret{\istypeshort{\Gamma}{\latbind \tickB{\nu(\kappa)}{(A\pair{\nu}{\subex\sigma\tickA\tickB})}}} \\
& = \pret{\istypeshort{\Gamma}{(\latbind \tickA{\kappa}A)\pair{\nu}\sigma}} 
\end{align*}

%

\section{Interpretation of clock quantification} 
\label{sec:forall}

Universal quantification over clocks should be modelled as a right adjoint to the semantic correspondent to clock weakening. 
Syntactically, clock weakening 
from context $\wfcxt[\Delta]{\Gamma}$ to $\wfcxt[\Delta, \kappa]{\Gamma}$ corresponds to the 
substitution $\pair i{\id_\Gamma}$, where $i$ is the inclusion of $\Delta$ into $\Delta, \kappa$. 
Recall from Section~\ref{sec:substitutions} 
that types and terms in context $\pret{\wfcxt[\Delta]\Gamma}$ in the CwF structure of 
$\gr\Delta$ correspond to types and terms in context $\compr{\pret\Delta}{\pret{\wfcxt[\Delta]\Gamma}}$ in $\grtotal$, and recall
that $\pret\Delta = \clk^\Delta$.  
By the substitution lemma, clock weakening from context $\wfcxt[\Delta]{\Gamma}$ to $\wfcxt[\Delta, \kappa]{\Gamma}$
corresponds to substitution along the composition
\[
\compr{\pret{\Delta, \kappa}}{\pret{\wfcxt[\Delta, \kappa]{\Gamma}}} 
\xrightarrow{\compr{\pret{\Delta, \kappa}}\pret{\pair i{\id_\Gamma}}}
\compr{\pret{\Delta, \kappa}}{i^*\pret{\wfcxt[\Delta]{\Gamma}}} \to
\compr{\pret{\Delta}}{\pret{\wfcxt[\Delta]{\Gamma}}} 
\]
All such substitutions have right adjoints, but to get a simple description of this (and to satisfy the substitution lemma), we
will give a concrete description which can be briefly described as follows: To model $\forall\kappa . A$, open
a fresh semantic clock $\#$, map $\kappa$ to $\#$ 
and take the limit of $\pret A$ as $\#$ ranges over all natural numbers. To type this 
description we need the following lemma.

\begin{lemma} \label{lem:(i,id):invertible}
 Let $\triple\Theta\vartheta f$ be an object of $\gr\Delta$, let $\#$ be fresh for $\Theta$ and let 
 $\iota : \Theta \to \Theta, \#$ be the inclusion.
%
 The component of $\pret{\pair{i}{\id_\Gamma}}$ at $\triple{\Theta, \#}{\vartheta\basicsub\# n}{f\basicsub{\kappa}{\#}}$ 
 is an isomorphism
\[
  \pret{\wfcxt[\Delta, \kappa]{\Gamma}}_{\triple{\Theta, \#}{\vartheta\basicsub\# n}{f\basicsub{\kappa}{\#}}} \to 
  i^* \pret{\wfcxt[\Delta]{\Gamma}}_{\triple{\Theta, \#}{\vartheta\basicsub\# n}{f\basicsub{\kappa}{\#}}}
  = \pret{\wfcxt[\Delta]{\Gamma}}_{\triple{\Theta, \#}{\vartheta\basicsub\# n}{\iota f}}
\] 
If $\kappa'\in \Delta$, the inverse to 
$\pret{\pair{i}{\id_\Gamma}}_{\triple{\Theta, \#}{\vartheta\basicsub\# n}{f\basicsub{\kappa}{\#}}}$ is given by
the component of 
\begin{align*}
i^*(\pret{\pair{\id_\Delta\basicsub{\kappa}{\kappa'}}{\id_\Gamma}}) : 
i^* \pret{\wfcxt[\Delta]{\Gamma}} \to 
i^*(\id_\Delta\basicsub{\kappa}{\kappa'})^* \pret{\wfcxt[\Delta,\kappa]{\Gamma}}
= \pret{\wfcxt[\Delta, \kappa]{\Gamma}}
\end{align*}
at ${\triple{\Theta, \#}{\vartheta\basicsub\# n}{f\basicsub{\kappa}{\#}}})$. 
\end{lemma}

We can now define the interpretation of universal quantification over clocks: 
\begin{align*}
&\pret{\istype{\Gamma}{\forall\kappa . A}}_{\triple\Theta\vartheta f}(\gamma) \defeq \\
& \{ (\omega_n) \in \displaystyle\prod_{n\in \nats} {\pret{\istype[\Delta, \kappa]{\Gamma}A}}_{{\triple{\Theta, \#}{\vartheta\basicsub\# n}{\subex f{\kappa}{\#}}}} (\inv{\pret{\pair{i}{\id_\Gamma}}}(\iota\cdot \gamma)) 
  \mid \forall n . \omega_n = (\omega_{n+1})|_{n} \}
\end{align*}
Here, by assumption $\gamma \in \pret{\wfcxt[\Delta]\Gamma}_{\triple\Theta\vartheta f}(\gamma)$ and so 
\[
\iota \cdot \gamma \in \pret{\wfcxt[\Delta]\Gamma}_{\triple{\Theta, \#}{\vartheta\basicsub{\#}n} {\iota f}} 
= i^*\pret{\wfcxt[\Delta]\Gamma}_{\triple{\Theta, \#}{\vartheta\basicsub{\#}n} {f\basicsub{\kappa}{\#}}} 
\] 
which means that $\inv{\pret{\pair{i}{\id_\Gamma}}}(\iota\cdot \gamma) \in 
\pret{\wfcxt[\Delta, \kappa]\Gamma}_{\triple{\Theta, \#}{\vartheta\basicsub{\#}n} {f\basicsub{\kappa}{\#}}}$
and thus the type of each $\omega_n$ is a welldefined set. In the condition for the families, $(\omega_{n+1})|_{n}$
is the restriction of $\omega_{n+1}$ along the map ${\triple{\Theta, \#}{\vartheta\basicsub{\#}{n+1}} {f\basicsub{\kappa}{\#}}} 
\to {\triple{\Theta, \#}{\vartheta\basicsub{\#}n} {f\basicsub{\kappa}{\#}}}$ given by the identity on $\Theta, \#$. 
For the interpretation of terms define 
\begin{align*}
\pret{\hastype{\Gamma}{\Lambda \kappa . t}{\forall\kappa . A}}_{\triple\Theta\vartheta f} 
& = (\pret{\hastype[\Delta,\kappa]{\Gamma}{t}{A}}_{{\triple{\Theta, \#}{\vartheta\basicsub\# n}{\iota f}}} (\inv{\pret{\pair{i}{\id_\Gamma}}}(\iota\cdot \gamma)))_n \\
\pret{\hastype{\Gamma}{t [\kappa']}{A \subst{\kappa}{\kappa'}}}_{\triple\Theta\vartheta f} 
& = \basicsub{\#}{f(\kappa')}\cdot(\pret{\hastype{\Gamma}{t}{\forall\kappa .A}}_{\triple\Theta\vartheta f} (\gamma))_{\vartheta(f(\kappa'))}
\end{align*}
To see that the latter type checks, note first that 
\[
\pret{\hastype{\Gamma}{t}{\forall\kappa .A}}_{\triple\Theta\vartheta f} (\gamma))_{\vartheta(f(\kappa'))}
\in 
{\pret{\istype[\Delta, \kappa]{\Gamma}A}}_{{\triple{\Theta, \#}{\vartheta\basicsub\#{\vartheta(f(\kappa'))}}{\subex f{\kappa}{\#}}}} (\inv{\pret{\pair{i}{\id_\Gamma}}}(\iota\cdot \gamma)) 
\]
and since $\basicsub{\#}{f(\kappa')} : {\triple{\Theta, \#}{\vartheta\basicsub\#{\vartheta(f(\kappa'))}}{\subex f\kappa\#}} \to 
{\triple{\Theta}{\vartheta}{\subex f\kappa{f(\kappa')}}}$, the right hand side of the definition is in 
\begin{align*}
 & {\pret{\istype[\Delta, \kappa]{\Gamma}A}}_{\triple{\Theta}{\vartheta}{\subex f\kappa{f(\kappa')}}} (\basicsub{\#}{f(\kappa')} \cdot\inv{\pret{\pair{i}{\id_\Gamma}}}(\iota\cdot \gamma))  \\
 = &  
(\subex{\id_\Delta}{\kappa}{\kappa'})^*(\pret{\istype[\Delta, \kappa]{\Gamma}A})_{\triple{\Theta}{\vartheta}{f}} (\basicsub{\#}{f(\kappa')} \cdot\inv{\pret{\pair{i}{\id_\Gamma}}}(\iota\cdot \gamma))  \\
 = &  
(\subex{\id_\Delta}{\kappa}{\kappa'})^*(\pret{\istype[\Delta, \kappa]{\Gamma}A})_{\triple{\Theta}{\vartheta}{f}} (\basicsub{\#}{f(\kappa')} \cdot{\pret{\pair{\subex{\id_\Delta}{\kappa}{\kappa'}}{\id_\Gamma}}}(\iota\cdot \gamma))  \\
 = &  
 (\subex{\id_\Delta}{\kappa}{\kappa'})^*(\pret{\istype[\Delta, \kappa]{\Gamma}A})_{\triple{\Theta}{\vartheta}{f}} ({\pret{\pair{\subex{\id_\Delta}{\kappa}{\kappa'}}{\id_\Gamma}}}(\basicsub{\#}{f(\kappa')} \cdot\iota\cdot \gamma))  \\
 = &  
(\subex{\id_\Delta}{\kappa}{\kappa'})^*(\pret{\istype[\Delta, \kappa]{\Gamma}A})_{\triple{\Theta}{\vartheta}{f}} ({\pret{\pair{\subex{\id_\Delta}{\kappa}{\kappa'}}{\id_\Gamma}}}(\gamma))  \\
 = &  
\pret{\istype[\Delta]{\Gamma}{A\subst{\kappa}{\kappa'}}}_{\triple{\Theta}{\vartheta}{f}} (\gamma) 
\end{align*}
where the last equality is by the substitution lemma. 

\begin{lemma} \label{lem:beta:eta:clock:quant}
 The $\beta$ and $\eta$ rules for universal quantification over clocks are sound for the interpretation. 
\end{lemma}

\section{Interpretation of $\tickc$}
\label{sec:tickc}

To interpret the rule for the tick constant $\tickc$, we define a substitution 
\[
\pret{\pair{[\kappa\mapsto \kappa']}{\basicsub{\tickA}\tickc}} : \pret{\wfcxt{\Gamma}} \to 
\basicsub\kappa{\kappa'}^*\pret{\wfcxt[\Delta,\kappa]{\Gamma, \tickA : \kappa}}
\]
for every context $\wfcxt\Gamma$ with $\kappa \notin \Delta$. 
The interpretation of application to $\tickc$ can then be defined as 
\begin{align*}
 \pret{\hastype{\Gamma}{\tappc{t\subst{\kappa}{\kappa'}}}{A\subst{\kappa}{\kappa'}\subst{\tickA}{\tickc}}}
 & =  (\basicsub\kappa{\kappa'}^*\pret{\hastype[\Delta,\kappa]{\Gamma, \tickA : \kappa}{t}{A}})
 [\pret{\pair{[\kappa\mapsto \kappa']}{\basicsub{\tickA}\tickc}}]
\end{align*}
which has type
\[
(\basicsub\kappa{\kappa'}^*\pret{\istype[\Delta,\kappa]{\Gamma, \tickA : \kappa}{A}})
 [\pret{\pair{[\kappa\mapsto \kappa']}{\basicsub{\tickA}\tickc}}]
\]
which equals the required 
$\pret{\istype[\Delta,\kappa]{\Gamma, \tickA : \kappa}{A\subst{\kappa}{\kappa'}\subst{\tickA}{\tickc}}}$
by the substitution lemma. 

Suppose $\gamma \in \pret{\wfcxt\Gamma}_{\triple\Theta\vartheta f}$. We define
\[
  \pret{\pair{[\kappa\mapsto \kappa']}{\basicsub{\tickA}\tickc}}(\gamma) 
  \defeq \pair{\{\kappa\}}{\inv{\pret{\pair i{\id_\Gamma}}}(\iota\cdot \gamma))}
\]
where $\iota : {\triple\Theta\vartheta f} \to {\triple{\Theta, \#}{\subex\vartheta{\#}{n+1}}f}$ is the inclusion
for $n = \vartheta(f(\kappa'))$. We must show that this defines an element of 
$\basicsub\kappa{\kappa'}^*\pret{\wfcxt[\Delta,\kappa]{\Gamma, \tickA : \kappa}}_{\triple\Theta\vartheta f}$. To see this,
note first that 
\[
 \iota\cdot \gamma \in \pret{\wfcxt\Gamma}_{\triple{\Theta, \#}{\subex\vartheta{\#}{n+1}}f}
 = i^*\pret{\wfcxt\Gamma}_{\triple{\Theta, \#}{\subex\vartheta{\#}{n+1}}{\subex f\kappa\#}}
\]
so by Lemma~\ref{lem:(i,id):invertible}, 
\[
 \inv{\pret{\pair i{\id_\Gamma}}}(\iota\cdot \gamma) \in
 \pret{\wfcxt[\Delta,\kappa]\Gamma}_{\triple{\Theta, \#}{\subex\vartheta{\#}{n+1}}{\subex f\kappa\#}}
\]
and 
\begin{align*}
 \pair{\{\kappa\}}{\inv{\pret{\pair i{\id_\Gamma}}}(\iota\cdot \gamma))} 
 & \in 
 \displaystyle\coprod_{\kappa\in X \subset f^{-1}(f(\kappa))}
 \pret{\wfcxt[\Delta,\kappa]\Gamma}_{\triple{\Theta, \#}{\subex\vartheta{\#}{n+1}}{\subex f\kappa{f(\kappa')}\basicsub{X}{\#}}} \\
 & = \pret{\wfcxt[\Delta,\kappa]{\Gamma,\tickA : \kappa}}_{\triple{\Theta}{\vartheta}{\subex f\kappa{f(\kappa')}}} \\
 & = \basicsub{\kappa}{\kappa'}^*\pret{\wfcxt[\Delta,\kappa]{\Gamma,\tickA : \kappa}}_{\triple{\Theta}{\vartheta}{f}}
\end{align*}

We note that this satisfies the equality for $\tickc$ in Figure~\ref{fig:clott:typing}. 

\begin{lemma}
\label{lem:dfixapptick}
 The interpretations of $\tappc{(\dfix^{\kappa'} \, t)}$ and $t\,(\dfix^{\kappa'}\,t)$ are equal. 
\end{lemma}

\subsection{Welldefinedness}
\label{sec:welldef}

As mentioned in Section~\ref{sec:clott} the unusual typing rule for $\tappc t$ introduces a problem of welldefinedness of the
interpretation: If $t$ is a term, a proof of the typing judgement $\hastype\Gamma{\tappc t}{A\subst{\tickA}{\tickc}}$
consists of a term $s$ such that $s\subst \kappa{\kappa'} = t$, a type $B$ such that $B\subst \kappa{\kappa'} = A$ 
and a proof of a typing judgement $\hastype[\Delta,\kappa]\Gamma{s}B$. In general there may be different possible
choices of $s$ and $B$, but the next lemma states that the interpretation of the term $\tappc t$ is independent of this 
choice. This means that the interpretation of a welltyped term is a welldefined object, independent of the choice of
typing derivation.

\begin{proposition}
 If $\hastype[\Delta,\kappa]\Gamma{s}B$ and $\hastype[\Delta,\kappa]\Gamma{u}C$ are such that 
 $\hastype[\Delta]\Gamma{s\subst \kappa{\kappa'}}{B\subst \kappa{\kappa'}}$ is equal to
 $\hastype[\Delta]\Gamma{u\subst \kappa{\kappa'}}{C\subst \kappa{\kappa'}}$ then 
 \[
 \basicsub{\kappa}{\kappa'}^*\pret{\tapp s}[\pret{\pair{[\kappa\mapsto \kappa']}{\basicsub{\tickA}\tickc}}] =
 \basicsub{\kappa}{\kappa'}^*\pret{\tapp u}[\pret{\pair{[\kappa\mapsto \kappa']}{\basicsub{\tickA}\tickc}}]
 \]
\end{proposition}

\begin{proof}
 The assumption implies that also 
 \[
 \hastype[\Delta]{\Gamma, \tickA : \kappa'}{(\tapp{s})\subst \kappa{\kappa'}}{B\subst \kappa{\kappa'}}
 \qquad\text{ and } \qquad 
 \hastype[\Delta]{\Gamma, \tickA : \kappa'}{(\tapp{u})\subst \kappa{\kappa'}}{B\subst \kappa{\kappa'}}
 \]
 are equal, and so by the substitution lemma
 \[
   (\basicsub\kappa{\kappa'}\pret{\tapp s})[\pret{\pair{i}{\subex{\id_{\Gamma}}\tickA\tickA}}] 
   = (\basicsub\kappa{\kappa'}\pret{\tapp u})[\pret{\pair{i}{\subex{\id_{\Gamma}}\tickA\tickA}}]
 \]
 Now, if $\gamma\in \pret{\wfcxt\Gamma}_{\triple\Theta\vartheta f}$ then 
\begin{align*}
  \basicsub{\kappa}{\kappa'}^*\pret{\tapp s}[\pret{\pair{[\kappa\mapsto \kappa']}{\basicsub{\tickA}\tickc}}] (\gamma)
 & = \basicsub{\kappa}{\kappa'}^*\pret{\tapp s}\pair{\{\kappa\}}{\inv{\pret{\pair i{\id_\Gamma}}}(\iota\cdot \gamma)} \\
 & = \basicsub{\kappa}{\kappa'}^*\pret{\tapp s}\pair{\{\kappa\}}{\pret{\pair {\basicsub\kappa{\kappa'}}{\id_\Gamma}}(\iota\cdot \gamma)} \\
 & = \basicsub{\kappa}{\kappa'}^*\pret{\tapp s}\pret{\pair{\basicsub\kappa{\kappa'}}{\subex{\id_{\Gamma}}\tickA\tickA}}\pair{\{\kappa\}}{(\iota\cdot \gamma)} \\
 & = \basicsub{\kappa}{\kappa'}^*\pret{\tapp u}\pret{\pair{\basicsub\kappa{\kappa'}}{\subex{\id_{\Gamma}}\tickA\tickA}}\pair{\{\kappa\}}{(\iota\cdot \gamma)} \\
 & = \basicsub{\kappa}{\kappa'}^*\pret{\tapp u}[\pret{\pair{[\kappa\mapsto \kappa']}{\basicsub{\tickA}\tickc}}] (\gamma) \qedhere
\end{align*}
\end{proof}

\section{Clock irrelevance}
\label{sec:cirr}

We now show how to model the clock irrelevance axiom (\ref{eq:cirr}). For this it suffices to show that if 
$\hastype{\Gamma}t{\forall\kappa . A}$, $\istype\Gamma A$ and $\kappa', \kappa''\in \Delta$ then 
\[
\pret{\hastype\Gamma{t[\kappa']}A}
= \pret{\hastype\Gamma{t[\kappa'']}A}
\]
Recall that for $\gamma\in \pret{\wfcxt\Gamma}_{\triple\Theta\vartheta f}$ 
\[
  \pret{\hastype\Gamma{t[\kappa']}A}_{\triple\Theta\vartheta f}(\gamma) = 
  \basicsub{\#}{f(\kappa')}\cdot(\pret{\hastype{\Gamma}t{\forall\kappa . A}}_{\triple\Theta\vartheta f}(\gamma))_{n}
\]
where $n = \vartheta(f(\kappa'))$. 
Here the element $(\pret{\hastype{\Gamma}t{\forall\kappa . A}}_{\triple\Theta\vartheta f}(\gamma))_{n}$
lives in 
\begin{align*}
 & \pret{\istype[\Delta,\kappa]{\Gamma}{A}}_{\triple{\Theta, \#}{\subex{\vartheta}{\#}n}{\subex f\kappa\#}}
 (\inv{\pret{\pair{i}{\id_\Gamma}}} (\iota \cdot \gamma)) \\
 & = i^*(\pret{\istype[\Delta]{\Gamma}{A}})_{\triple{\Theta, \#}{\subex{\vartheta}{\#}n}{\subex f\kappa\#}}
 (\pret{\pair{i}{\id_\Gamma}}(\inv{\pret{\pair{i}{\id_\Gamma}}} (\iota \cdot \gamma))) \\
 & = \pret{\istype[\Delta]{\Gamma}{A}}_{\triple{\Theta, \#}{\subex{\vartheta}{\#}n}{f}}
 (\iota \cdot \gamma) 
\end{align*}
Clock irrelevance will follow from the following lemma.

\begin{lemma} \label{lem:iota:dot:iso}
 Suppose $\istype{\Gamma}A$ and $\gamma\in \pret{\wfcxt\Gamma}_{\triple\Theta\vartheta f}$. The map 
 \[
 \iota\cdot (-) : \pret{\istype{\Gamma}A}_{\triple\Theta\vartheta f} 
 \to \pret{\istype{\Gamma}A}_{\triple{\Theta,\#}{\subex{\vartheta}{\#}{n}}f}(\iota\cdot \gamma)
 \]
is an isomorphism.
\end{lemma}

In particular, there is an element $x$ such that 
$\iota\cdot x = (\pret{\hastype{\Gamma}t{\forall\kappa . A}}_{\triple\Theta\vartheta f}(\gamma))_{n}$
and so 
\begin{align*}
 \pret{\hastype\Gamma{t[\kappa']}A}_{\triple\Theta\vartheta f}(\gamma) &  = 
 \basicsub{\#}{f(\kappa')}\cdot \iota \cdot x = x
\end{align*}
Likewise $\pret{\hastype\Gamma{t[\kappa'']}A}_{\triple\Theta\vartheta f}(\gamma) = x$ proving clock irrelevance. 

Lemma~\ref{lem:iota:dot:iso} can be proved by induction on $A$ using the techniques of~\cite{GDTTmodel}. 
In particular,~\cite{GDTTmodel} proves
that the statement of the lemma is equivalent to the statement that the context projection map 
 $\pret{\wfcxt{\Gamma, x : A}} \to \pret{\wfcxt\Gamma}$
(considered as a morphism in $\grtotal$) is \emph{orthogonal} to all objects of the form $y(\lambda, n)$ where $y$ is the yoneda
embedding. Here, orthogonality means that for all squares as below, there exists a unique filling diagonal:
\[
  \begin{diagram}
    y(\lambda, n) \times X \ar{r}{} \ar{d}{\pi} & \pret{\wfcxt{\Gamma, x : A}}\ar{d}{}\\
    X \ar[dotted]{ur}[description]{} \ar{r}{} & \pret{\wfcxt{\Gamma}}
  \end{diagram}
\]
where $\pi$ is the second projection. In particular, this implies that the condition is closed under $\Pi$- and $\Sigma$-types, and 
substitutions, see~\cite{GDTTmodel} for details. This proof can be easily extended to the cases of $\latbind\tickA\kappa A$ 
and $\forall\kappa . A$.

\section{Conclusion and future work}
\label{sec:conclusion}

We have constructed a model of \clott\ modelling ticks on clocks using an adjunction where the right adjoint extends
to types and terms. 
The description of the left adjoint $\tearlier^\kappa$ is fairly heavy to work with, but by abstracting away the required properties
needed to model the tick calculus, the model can be described in reasonable space. 

Future work includes extending the model to universes, which we expect to be easy using the universes
constructed in~\cite{GDTTmodel}. As noted there the axiom of clock irrelevance forces universes 
to be indexed by syntactic clock contexts. Fortunately, we can model a notion of universe polymorphism in the clock context: 
inclusions of clock contexts induce inclusion of universes, and these commute with type constructions on the nose. These
results, however, must be adapted to \clott, in particular the code for the $\later$  type constructor should be extended
to the tick abstracting generalisation used in \clott. We expect this to be a simple adaptation. Using universes, guarded 
recursive types can be encoded~\cite{birkedal2013intensional}, indeed these can also be added as primitive, given that
many of them exist in the model~\cite{GDTTmodel}. 

Our motivation for constructing this model is to study extensions of \clott. In particular, we would like to extend \clott\ with
path types as in~\cite{GCTT}. This requires an adaptation of the model to the cubical setting, using $\catT$ indexed
families of cubical sets~\cite{CTT} rather than just sets. 

\subparagraph*{Acknowledgements.} 

This work was supported by DFF-Research Project 1 Grant no.
4002-00442, from The Danish Council for Independent Research for the Natural Sciences (FNU) and by 
a research grant (13156) from VILLUM FONDEN.

\bibliography{paper.bib}

\newpage
\appendix
\section{Proof appendix for reviewers}

This appendix contains proofs for the reviewers. It is not part of the main submission. 

\section{Soundness for the tick calculus}

We omit the proof of this as it is a fragment of \clott. Instead we give the full proof of soundness of \clott. The 
proof of soundness for this fragment is sufficiently abstract that it also applies to the abstract notion of model for the 
tick calculus. 

\section{Dependent right adjoint types in $\gr\Delta$}

Let $\tau:\triple \Theta \vartheta f \to \triple {\Theta'} {\vartheta'} {f'}$ we write $\tau[\kappa-]$ for the map $\triple \Theta {\vartheta[f(\kappa)-]} f \to \triple {\Theta'} {\vartheta'[f(\kappa)-]} {f'}$ with the same action on $\Theta$ as $\tau$. 

We also define for $\kappa \in X\subset \Delta$ the map $\tau[X,\kappa+]: \triple \Theta \vartheta f [X,\kappa+] \to \triple {\Theta'} {\vartheta'} f [X,\kappa+]$ where $\tau[X,\kappa+](\#_\Theta) = \#_{\Theta'}$ and $\tau [X,\kappa+]|_\Theta = \tau$.

\begin{proofof}{Lemma~\ref{lem:earlier}}
 Note first that the definition of the lemma extends to a presheaf structure: If 
 $\tau:\triple\Theta\vartheta f \to \triple{\Theta'}{\vartheta'}{f'}$ and $X\subseteq f^{-1} (f(\kappa))$ such that  $\kappa \in X$ then 
 also $X \subset f'^{-1}(f'(\kappa))$. We can then define 
\[
  \tearlier^\kappa F(\tau) (X,\gamma) = (X, \tau[X,\kappa+]\cdot \gamma)
\]

For $\varphi:F\rightarrow G$, the map $\tearlier^\kappa \varphi:\tearlier^\kappa F \rightarrow \tearlier^\kappa G$ is given by  
\[(\tearlier^\kappa \varphi)_{\triple \Theta \vartheta  f} (X,\gamma) = (X,\varphi_{{\triple \Theta \vartheta f}[X,\kappa+]} (\gamma))\]

Before we describe the unit $\eta:\id \rightarrow \tlater^\kappa \tearlier^\kappa$ and the counit $\epsilon^\kappa:\tearlier^\kappa \tlater^\kappa \rightarrow \id$ of the adjunction we spell out the description of $\tlater^\kappa \tearlier^\kappa F$ and $\tearlier^\kappa \tlater^\kappa F$.

If $\vartheta(f(\kappa)) > 0$ then 
\begin{align*}(\tlater^\kappa \tearlier^\kappa F) \triple \Theta \vartheta f &= \tearlier^\kappa F \triple \Theta{\vartheta[f(\kappa)-]} f \\
&=  \displaystyle	\coprod_{\kappa \in X\subset \inv{f}(f(\kappa))} F \triple {\Theta,\#_\Theta} {\vartheta [f(\kappa)-],\#_\Theta \mapsto \vartheta(f(\kappa))} {f[X\mapsto \#_\Theta]} 
\end{align*}
If $\vartheta(f(\kappa)) = 0$ then $(\tlater^\kappa \tearlier^\kappa F) \triple \Theta \vartheta f = \{\ast\}$.

On the other hand
\begin{align*}
\tearlier^\kappa \tlater^\kappa F \triple \Theta \vartheta f  &= \displaystyle \coprod_{\kappa \in X\subset \inv{f}(f(\kappa))}\tlater^\kappa F \triple {\Theta,\#_\Theta} {\vartheta, \#_\Theta \mapsto \vartheta(f(\kappa))+1} {f[X\mapsto \#_\Theta]}\\
&= \displaystyle \coprod_{\kappa \in X\subset \inv{f}(f(\kappa))} F \triple {\Theta,\#_\Theta} {\vartheta,\#_\Theta \mapsto \vartheta(f(\kappa))} {f[X\mapsto \#_\Theta]}
\end{align*}

For $\gamma \in F\triple \Theta \vartheta f$ we then have 
$\eta^{\kappa}_F \gamma = \begin{cases} \ast & \vartheta(f(\kappa))=0 \\
(f^{-1}(f(\kappa)), s^{\kappa}_{\triple \Theta \vartheta f} \cdot \gamma) &\text{otherwise}
\end{cases}$
\\
where $s^\kappa_{\triple \Theta \vartheta f}: \triple \Theta \vartheta f \to \triple {\Theta,\#_\Theta} {\vartheta[f(\kappa)-],\#_\Theta \mapsto \vartheta(f(\kappa))} {f[f^{-1}(f(\kappa))\mapsto \#_\Theta]}$ is the map
\begin{align*}
s^\kappa_{\triple \Theta \vartheta f} (\lambda) = 
\begin{cases}
\#_\Theta & \lambda = f(\kappa)\\
\lambda & \text{otherwise}
\end{cases} && \text{When }\vartheta(f(\kappa)) >0 
\end{align*}

For $\gamma \in \tearlier^\kappa \tlater^\kappa F\triple \Theta \vartheta f$ we define $\epsilon^{\kappa}_F (X,\gamma)=r^{X,\kappa}_{\triple \Theta \vartheta f}\cdot \gamma$, where 
\[r^{X,\kappa}_{\triple \Theta \vartheta f}: \triple {\Theta,\#_\Theta}{\vartheta, \#_\Theta \mapsto f(\kappa)}{f[X\mapsto \#_\Theta]} \rightarrow \triple \Theta \vartheta f\] 
is given by
\begin{align*}
r^{X,\kappa}_{\triple \Theta \vartheta f} (\lambda) = 
\begin{cases}
f(\kappa) & \lambda = \#_\Theta\\
\lambda   & \text{otherwise}
\end{cases}
\end{align*}

It remains to show the two triangle equalities:
\begin{align}
\tlater^\kappa \epsilon^{\kappa}_F \circ \eta^\kappa_{\tlater^\kappa F} = \id_{\tlater^\kappa F}\label{eq:adj1}\\
\epsilon^\kappa_{\tearlier^\kappa F} \circ \tearlier^\kappa \eta^\kappa_{F} = \id_{\tearlier^\kappa \Gamma}\label{eq:adj2}
\end{align}

To show~(\ref{eq:adj1}) we first we prove that when $\vartheta(f(\kappa)) > 0$ we have that $ (s^\kappa)_{\triple \Theta \vartheta f}[\kappa-]$ followed by $r^{\inv{f}(f(\kappa)), \kappa}_{\triple \Theta {\vartheta[f(\kappa)-]} f}$ is the identity.
\begin{align*}
r^{\inv{f}(f(\kappa)), \kappa,}_{\triple \Theta {\vartheta[f(\kappa)-]} f} (s^\kappa_{\triple \Theta \vartheta f}[\kappa-](\lambda)) = 
\begin{cases}
r^{\inv{f}(f(\kappa)),\kappa}_{\triple \Theta {\vartheta[f(\kappa)-]} f}(\#_\Theta) = f(\kappa) & \lambda = f(\kappa)\\
\\
r^{\inv{f}(f(\kappa)),\kappa}_{\triple \Theta {\vartheta[f(\kappa)-]} f}(\lambda) = \lambda & \text{Otherwise}
\end{cases}
\end{align*}

Suppose $\gamma \in \tlater^\kappa F \triple \Theta \vartheta f$. If $\vartheta(f(\kappa))=0$ the codomain of the maps is a 
singleton and thus the equality is trivial. 
 Otherwise we have
\begin{align*}
(\tlater^\kappa \epsilon^k_\Gamma)_{\triple \Theta \vartheta f} ((\eta^k_{\tlater^\kappa \Gamma})_{\triple \Theta \vartheta f} (\gamma)) &= (\tlater^\kappa \epsilon^k_\Gamma)_{\triple \Theta \vartheta f}(\inv{f}(f(\kappa)),(\tlater^\kappa \Gamma)(s^\kappa_{\triple \Theta \vartheta f}) (\gamma)) \\
&=(\epsilon^{\kappa}_\Gamma)_{\triple \Theta {\vartheta[f(\kappa)-]} f}(\inv{f}(f(\kappa)),(\Gamma(s^\kappa_{\triple \Theta \vartheta f}[\kappa-]) (\gamma)) \\
&=\Gamma(r^{\inv{f}(f(\kappa)),\kappa}_{\triple \Theta {\vartheta[f(\kappa)-]} f}) (\Gamma(s^\kappa_{\triple \Theta \vartheta f}[\kappa-]) (\gamma)) \\
& = \gamma
\end{align*}

To show~(\ref{eq:adj2}) first we prove that for any $X\subseteq(\inv{f}(f(\kappa)))$ with $\kappa\in X$ we have that $s^\kappa_{\triple \Theta \vartheta f [X,\kappa+]}$ followed by $r^{X,\kappa}_{\triple \Theta \vartheta f} [X,\kappa+]$ is the identity on $\triple \Theta \vartheta f [X,\kappa+]$. We have
\begin{align*}
r_{\triple \Theta \vartheta f}^{X,\kappa}[X,\kappa+](s^\kappa_{\triple \Theta \vartheta f [X,\kappa+]}(\lambda)) = 
\begin{cases}
r^{X,\kappa}_{\triple \Theta \vartheta f}[X,\kappa+] (\#_{\Theta,\#_\Theta}) = \#_\Theta & \lambda = f(\kappa)=\#_\Theta\\
r[X,\kappa+]_{\triple \Theta \vartheta f}(\lambda)= \lambda & \text{Otherwise}
\end{cases}
\end{align*}

Let $(X,\gamma) \in \tearlier^\kappa \Gamma \triple \Theta \vartheta f$ we have
\begin{align*}
(\epsilon^k_{\tearlier^\kappa \Gamma})_{\triple \Theta \vartheta f} ((\tearlier^\kappa \eta^k_{\Gamma})_{\triple \Theta \vartheta f} (X,\gamma)) & = (\epsilon^k_{\tearlier^\kappa \Gamma})_{\triple \Theta \vartheta f}(X, (\eta^\kappa_\Gamma)_{\triple \Theta \vartheta f [X,\kappa+]}(\gamma))\\
& = \tearlier^\kappa \Gamma(r^{X,\kappa}_{\triple \Theta \vartheta f}) ((\eta^\kappa_\Gamma)_{\triple \Theta \vartheta f [X,\kappa+]}(\gamma)) \\
& = \tearlier^\kappa \Gamma(r^{X,\kappa}_{\triple \Theta \vartheta f}) (X, (\Gamma(s^\kappa_{\triple \Theta \vartheta f[X,\kappa+]}) (\gamma)) \\
& = (X, (\Gamma(r^{X,\kappa}_{\triple \Theta \vartheta f}[X,\kappa+])(\Gamma (s^\kappa_{\triple \Theta \vartheta f[X,\kappa+]}) (\gamma))))\\
& =(X, \gamma)\qedhere
\end{align*}
\end{proofof}

\subsection{Tick irrelevance}
\label{sec:tirr}

In this section we give a proof of Proposition~\ref{prop:tirr}. 
%
%
%
%
First we will need the following lemmas.


\begin{lemma}
\label{lem:laterofnextisnextoflater}
Let $\varphi:F\rightarrow \tlater^\kappa G$. We have
\[\tlater^\kappa \transp{\p_{\tearlier^\kappa}} \circ \varphi = \tlater^\kappa \varphi \circ \transp{\p_{\tearlier^\kappa}}\]
Where $\transp{\p_{\tearlier^\kappa}}$ on the left has the type $G\rightarrow \tlater^\kappa G$ and $\transp{\p_{\tearlier^\kappa}}$ on the right has the type $F\rightarrow \tlater^\kappa F$.
\end{lemma}
\begin{proof}
The statement follows from the observation that $(\transp{\p_{\tearlier^\kappa}})_{\tlater^\kappa G}$ is the same as $\tlater^\kappa (\transp{\p_{\tearlier^\kappa}})_{G}$.

At $\triple \Theta \vartheta f$ the map  $(\transp{\p_{\tearlier^\kappa}})_{\tlater^\kappa G}$ is given by
\[
(\tlater^\kappa G)(\id_\Theta:\triple \Theta \vartheta f \to \triple \Theta {\vartheta[f(\kappa)-]} f = G(\id_\Theta: \triple \Theta {\vartheta[f(\kappa)-]} f \to \triple \Theta {\vartheta[f(\kappa)--]} f)
\]
On the other hand at $\triple \Theta \vartheta f$ the map $\tlater^\kappa (\transp{\p_{\tearlier^\kappa}})_{G}$ is given by
\[
((\transp{\p_{\tearlier^\kappa}})_{G})_{\triple \Theta {\vartheta[f(\kappa)-]} f} = G(\id_\Theta: \triple \Theta {\vartheta[f(\kappa)-]} f \to \triple \Theta {\vartheta[f(\kappa)--]} f)
\]
Let $\gamma \in F(\Theta,\vartheta,f)$ then 
\begin{align*}
(\tlater^\kappa (\transp{\p_{\tearlier^\kappa}})_G)_{\triple \Theta \vartheta f} (\varphi_{\triple \Theta \vartheta f} \gamma) & = ((\transp{\p_{\tearlier^\kappa}})_{\tlater^\kappa G})_{\triple \Theta \vartheta f} (\varphi_{\triple \Theta \vartheta f} \gamma) \\
&= (\varphi_{\triple \Theta \vartheta f} \gamma)|_{\triple {\Theta[f(\kappa)-]} \vartheta f}
\end{align*}
On the other hand 
\begin{align*}
(\tlater^\kappa \varphi)_{\triple \Theta \vartheta f}(((\transp{\p_{\tearlier^\kappa}})_F)_{\triple \Theta \vartheta f} \gamma) &=(\tlater^\kappa \varphi)_{\triple \Theta \vartheta f} (\gamma |_{ \triple \Theta {\vartheta[f(\kappa))-]} f}) \\
&= \varphi_{(\Theta,\vartheta[\vartheta(f(\kappa)) -])} (\gamma|_{ \triple \Theta {\vartheta[f(\kappa))-]} f})
\end{align*}
The statement follows by naturality of $\varphi$.
\end{proof}

\begin{lemma}
\label{lem:tirr}
If $\hastype{\Gamma}{t}{\later(\alpha \of\kappa) A}$ then 
\[
\pret{\hastype{\Gamma}{\lambda(\tickA'\of \kappa) t}{\later^\kappa \later^\kappa A}} =
\pret{\hastype{\Gamma}{\lambda(\tickA' \of \kappa)(\lambda(\tickA\of \kappa) \tapp[\tickA'] t)}{\later^\kappa \later^\kappa A}}
\] 
\end{lemma}
\begin{proof}
By the substitution lemma and Lemma~\ref{lem:syn:proj:nu} (below)
\[\pret{\hastype{\Gamma}{\lambda(\tickA'\of \kappa) t}{\later^\kappa \later^\kappa A}} = \transp{\pret{t}[\p_{\tearlier^\kappa]}]}\]
which by definition is equal to
\begin{align*}
(\tlater^\kappa(\pret{t} [\p_{\tearlier^\kappa}]))[\eta] = (\tlater^\kappa(\pret{t})) [\tlater^\kappa \p_{\tearlier^\kappa} \circ \eta] = (\tlater^\kappa(\pret{t})) [\transp{\p_{\tearlier^\kappa}}]
\end{align*}

On the other hand by Lemma~\ref{lem:bijectivecorresp} we have  
\[\pret{\hastype{\Gamma, \tickA' \of \kappa}{\tapp[\tickA'] t}{A}} = \q_\pret{A}[\epsilon \circ \tearlier^\kappa(\isocomp {\tearlier^\kappa \pret{\Gamma}}{\pret{A}} \circ \la \eta, \pret{t}\ra)]\]
By the substitution lemma and Lemma~\ref{lem:syn:proj:nu} (below), we have then that 
\begin{align*}
\pret{\hastype{\Gamma, \tickA' \of \kappa, \tickA \of \kappa}{\tapp[\tickA'] t}{A}} 
& = \q_\pret{A}[\epsilon \circ \tearlier^\kappa(\isocomp{\tearlier^\kappa \pret{\Gamma}}{\pret{A}} \circ \la \eta, \pret{t}\ra)][\p_{\tearlier^\kappa}] 
\end{align*}
and so
\begin{align*}
 \pret{\hastype{\Gamma, \tickA' \of \kappa}{\lambda(\tickA\of\kappa)\tapp[\tickA'] t}{A}} 
& = (\tlater^\kappa (\q_\pret{A}[\epsilon \circ \tearlier^\kappa(\isocomp{\tearlier^\kappa \pret{\Gamma}}{\pret{A}} \circ \la \eta, \pret{t}\ra)][\p_{\tearlier^\kappa}]))[\eta] 
\end{align*}
We then get
\[\pret{\hastype{\Gamma}{\lambda(\tickA' \of \kappa)\lambda(\tickA\of \kappa) \tapp[\tickA'] t }{\later^\kappa \later^\kappa A}} = (\tlater^\kappa ((\tlater^\kappa (\q_{\pret{A}}[\epsilon \circ \tearlier^\kappa (\isocomp{\tearlier^\kappa \pret{\Gamma}}{\pret{A}} \circ \la \eta,\pret{t}\ra)] [\p_{\tearlier^\kappa}]))[\eta]))[\eta]\]
But for any $u$ and $\sigma$ we have $\tlater^\kappa (u[\sigma]) = (\tlater^\kappa u) [\tlater^\kappa \sigma]$, hence
\begin{align*}
(\tlater^\kappa (\q_{\pret{A}}[\epsilon \circ& \tearlier^\kappa (\isocomp {\tearlier^\kappa \pret{\Gamma}}{\pret{A}}\circ \la \eta,\pret{t}\ra)] [\p_{\tearlier^\kappa}]))[\eta] \\
&= (\tlater^\kappa \q_{\pret{A}})[\tlater^\kappa \epsilon \circ \tlater^\kappa \tearlier^\kappa (\isocomp {\tearlier^\kappa \pret{\Gamma}}{\pret{A}} \circ \la \eta,\pret{t}\ra)][\tlater^\kappa \p_{\tearlier^\kappa}][\eta]\\
&=(\tlater^\kappa \q_{\pret{A}})[\tlater^\kappa \epsilon \circ \tlater^\kappa \tearlier^\kappa (\isocomp {\tearlier^\kappa \pret{\Gamma}}{\pret{A}} \circ \la \eta,\pret{t}\ra)][\tlater^\kappa \p_{\tearlier^\kappa} \circ \eta]\\
&=(\tlater^\kappa \q_{\pret{A}})[\tlater^\kappa \epsilon \circ \tlater^\kappa \tearlier^\kappa (\isocomp {\tearlier^\kappa \pret{\Gamma}}{\pret{A}} \circ \la \eta,\pret{t}\ra) \circ \transp{\p_{\tearlier^\kappa}}] \\
&=(\tlater^\kappa \q_{\pret{A}})[\transp{\p_{\tearlier^\kappa}} \circ \epsilon \circ  \tearlier^\kappa (\isocomp {\tearlier^\kappa \pret{\Gamma}}{\pret{A}} \circ \la \eta,\pret{t}\ra)] 
\end{align*}
In the last equality we have used naturality of $\transp{\p_{\tearlier^\kappa}}$, which is an easy consequence of the 
concrete description in Section~\ref{sec:basic:model}. 
Thus
\begin{align*}
\pret{\hastype{\Gamma}{\lambda(\tickA' \of \kappa)\lambda(\tickA\of \kappa) \tapp[\tickA'] t }{\later^\kappa \later^\kappa A}}& = (\tlater^\kappa ( (\tlater^\kappa \q_{\pret{A}})[\transp{\p_{\tearlier^\kappa}} \circ \epsilon \circ  \tearlier^\kappa (\isocomp {\tearlier^\kappa \pret{\Gamma}}{\pret{A}} \circ \la \eta,\pret{t}\ra)] )) [\eta]\\
&= (\tlater^\kappa \tlater^\kappa \q_{\pret{A}}) [\tlater^\kappa \transp{\p_{\tearlier^\kappa}} \circ \tlater^\kappa \epsilon \circ  \tlater^\kappa \tearlier^\kappa (\isocomp {\tearlier^\kappa \pret{\Gamma}}{\pret{A}} \circ \la \eta,\pret{t}\ra)\circ \eta]\\
&=(\tlater^\kappa \tlater^\kappa \q_{\pret{A}}) [\tlater^\kappa \transp{\p_{\tearlier^\kappa}} \circ \tlater^\kappa \epsilon\circ \eta \circ (\isocomp {\tearlier^\kappa \pret{\Gamma}}{\pret{A}} \circ \la \eta,\pret{t}\ra)]\\
&=(\tlater^\kappa \tlater^\kappa \q_{\pret{A}}) [\tlater^\kappa \transp{\p_{\tearlier^\kappa}} \circ (\isocomp {\tearlier^\kappa \pret{\Gamma}}{\pret{A}} \circ \la \eta,\pret{t}\ra)]
\end{align*}
by Lemma~\ref{lem:laterofnextisnextoflater} we have then that
\begin{align*}
\pret{\hastype{\Gamma}{\lambda(\tickA' \of \kappa)\lambda(\tickA\of \kappa) \tapp[\tickA'] t }{\later^\kappa \later^\kappa A}}
&= (\tlater^\kappa \tlater^\kappa \q_{\pret{A}}) [\tlater^\kappa (\isocomp {\tearlier^\kappa \pret{\Gamma}}{\pret{A}} \circ \la \eta,\pret{t}\ra)] [\transp{\p_{\tearlier^\kappa}}] \\
&= (\tlater^\kappa (\tlater^\kappa \q_{\pret{A}} [(\isocomp {\tearlier^\kappa \pret{\Gamma}}{\pret{A}} \circ \la \eta,\pret{t}\ra)]))[\transp{\p_{\tearlier^\kappa}}]
\end{align*}

But $(\tlater^\kappa \q_{\pret{A}})[\isocomp {\tearlier^\kappa \pret{\Gamma}}{\pret{A}}] = \q_{\tlater^\kappa A}$ by Definition~\ref{def:cwfa}, thus
\[(\tlater^\kappa (\tlater^\kappa \q_{\pret{A}} [(\isocomp {\tearlier^\kappa \pret{\Gamma}}{\pret{A}} \circ \la \eta,\pret{t}\ra)]))[\transp{\p_{\tearlier^\kappa}}] = \tlater^\kappa \pret{t} [\transp{\p_{\tearlier^\kappa}}] \]
which we saw above equals $\pret{\hastype{\Gamma}{\lambda(\tickA'\of \kappa) t}{\later^\kappa \later^\kappa A}}$. 
\end{proof}

Finally we can prove Proposition~\ref{prop:tirr}.

\begin{proofof}{Proposition~\ref{prop:tirr}}
By Lemma~\ref{lem:tirr} we have that if $\hastype{\Gamma}{t}{\later(\alpha \of\kappa) A}$ then 
\[
\pret{\hastype{\Gamma}{\lambda(\tickA'\of \kappa) t}{\later^\kappa \later^\kappa A}} =
\pret{\hastype{\Gamma}{\lambda(\tickA' \of \kappa)\lambda(\tickA\of \kappa) \tapp[\tickA'] t }{\later^\kappa \later^\kappa A}}
\] 
By double application we then get that
\[
\pret{\hastype{\Gamma, \tickA : \kappa, \tickA' : \kappa}{\tapp[\tickA']t}{A}}
 = \pret{\hastype{\Gamma, \tickA : \kappa, \tickA' : \kappa}{\tapp t}{A}}
\qedhere\] 
\end{proofof}

\section{Substitutions}

In this section we present the most interesting cases of the substitution lemma, leaving the case of universal quantification
over clocks to Appendix~\ref{app:forall}. We also prove the formulas (\ref{eq:e:comm:eta}) and (\ref{eq:e:comm:eps})
mentioned in Section~\ref{sec:substitutions}. 

We need the following lemma.

\begin{lemma}\label{lem:e:nat}
 If $\nu : \Delta'' \to \Delta'$ and $\mu : \Delta' \to \Delta$ and $\kappa \in \Delta''$ then 
 \[
 \mu^*(e^{\nu,\kappa}) \circ e^{\mu, \nu(\kappa)} = e^{\mu\nu,\kappa} : \,
 \tearlier^{\mu\nu(\kappa)}\mu^*\nu^* \to \mu^*\nu^*\tearlier^\kappa
 \]
 Moreover, $e^{\id,\kappa}$ is the identity. 
\end{lemma}

\begin{proof}
\begin{align*}
 \mu^*(e^{\nu,\kappa}) \circ e^{\mu, \nu(\kappa)} &= \mu^* \epsilon^{\nu(\kappa)}_{\nu^* \tearlier^\kappa} \circ \mu^* \tearlier^{\nu(\kappa)} \nu^*(\eta^\kappa) \circ \epsilon^{\mu\nu(\kappa)}_{\mu^* \tearlier^{\nu(\kappa)}} \circ \tearlier^{\mu\nu(\kappa)} \mu^* (\eta^{\nu(\kappa)}_{\nu^*})\\
 &=\mu^* \epsilon^{\nu(\kappa)}_{\nu^* \tearlier^\kappa} \circ \epsilon^{\mu\nu(\kappa))}_{\mu^* \tearlier^{\nu(\kappa)} \nu^* \tlater^\kappa \tearlier^\kappa} \circ \tearlier^{\mu\nu(\kappa)} \tlater^{\mu\nu(\kappa)} \mu^* \tearlier^{\nu(\kappa)} \nu^*(\eta) \circ \tearlier^{\mu\nu(\kappa)} \mu^* (\eta^{\nu(\kappa)}_{\nu^*})\\
 &=\mu^* \epsilon^{\nu(\kappa)}_{\nu^* \tearlier^\kappa} \circ \epsilon^{\mu\nu(\kappa)}_{\mu^* \tearlier^{\nu(\kappa)} \nu^* \tlater^\kappa \tearlier^\kappa} \circ \tearlier^{\mu\nu(\kappa)}\mu^*( \tlater^{\nu(\kappa)} \tearlier^{\nu(\kappa)} \nu^*(\eta) \circ (\eta^{\nu(\kappa)}_{\nu^*}))\\
 &=\mu^* \epsilon^{\nu(\kappa)}_{\nu^* \tearlier^\kappa} \circ \epsilon^{\mu\nu(\kappa)}_{\mu^* \tearlier^{\nu(\kappa)} \nu^* \tlater^\kappa \tearlier^\kappa} \circ \tearlier^{\mu\nu(\kappa)}\mu^*( \eta^{\nu(\kappa)}_{\nu^* \tlater^\kappa \tearlier^\kappa} \circ \nu^* \eta)
\end{align*}
Where the second equality follows by naturality of $\epsilon$, the third by commutation of $\mu^*$ with $\tlater$, and the fourth by naturality of $\eta$. By the definition
\[e^{\nu(\kappa),\mu}_{\nu^* \tlater^\kappa \tearlier^\kappa} = \epsilon^{\mu\nu(\kappa)}_{\mu^* \tearlier^{\nu(\kappa)} \nu^* \tlater^\kappa \tearlier^\kappa} \circ \tearlier^{\mu\nu(\kappa)} \mu^* (\eta^{\nu(\kappa)}_{\nu^* \tlater^\kappa \tearlier^\kappa})\]
We have then
\begin{align*}
 \mu^*(e^{\nu,\kappa}) \circ e^{\mu, \nu(\kappa)} &= \mu^* \epsilon^{\nu(\kappa)}_{\nu^* \tearlier^\kappa} \circ e^{\nu(\kappa),\mu}_{\nu^* \tlater^\kappa \tearlier^\kappa} \circ \tearlier^{\mu\nu(\kappa)}\mu^*(\nu^* \eta)
\end{align*}
By Formula~(\ref{eq:e:comm:eps}) we get
\begin{align*}
\mu^*(e^{\nu,\kappa}) \circ e^{\mu, \nu(\kappa)} &= \epsilon^{\mu\nu(\kappa)}_{(\mu \nu)^*\tearlier^\kappa} \circ \tearlier^{\mu\nu(\kappa)}\mu^*(\nu^* \eta)
\end{align*}
Which by the definition is equal to $e^{\kappa,\mu \nu}$

The second statement, i.e.~$e^{\id,\kappa}$ is the identity, follows from the adjunction equations.
\end{proof}
%
%

\begin{lemma}
 The formulas (\ref{eq:e:comm:eta}) and (\ref{eq:e:comm:eps}) in the main paper hold.
\end{lemma}

\begin{proof}
Equation~(\ref{eq:e:comm:eta}) holds because
\begin{align*}
 \tlater^{\nu(\kappa)}e^\kappa \circ \eta^{\nu(\kappa)}_{\nu^*} &= \tlater^{\nu(\kappa)} \epsilon_{\nu^* \tearlier^\kappa }^{\nu(\kappa)} \circ \tlater^{\nu(\kappa)} \tearlier^{\nu(\kappa)} \nu^*(\eta^\kappa)  \circ \eta^{\nu(\kappa)}_{\nu^*}\\
 &= \tlater^{\nu(\kappa)} \epsilon_{\nu^*\tearlier^\kappa }^{\nu(\kappa)} \circ \eta^{\nu(\kappa)}_{\nu^*\tlater^{\kappa} \tearlier^{\kappa}} \circ \nu^*(\eta^\kappa) \\
 &= \tlater^{\nu(\kappa)} \epsilon_{\nu^*\tearlier^\kappa }^{\nu(\kappa)} \circ \eta^{\nu(\kappa)}_{\tlater^{\nu(\kappa)}  \nu^* \tearlier^{\kappa}} \circ \nu^*(\eta^\kappa) \\
 &= \id_{\tlater^{\nu(\kappa)} \nu^* \tearlier^\kappa} \circ \nu^*(\eta^\kappa)\\
 &=  \nu^*(\eta^\kappa)
\end{align*}

Equation~(\ref{eq:e:comm:eps}) holds because by naturality of the counit
\[\epsilon^{\nu(\kappa)}_{\nu^*} \circ \tearlier^{\nu(\kappa)} \tlater^{\nu(\kappa)} (\nu^* (\epsilon^\kappa)) = \nu^* (\epsilon^\kappa) \circ \epsilon^{\nu(\kappa)}_{\nu^* \tearlier^\kappa \tlater^\kappa}\]
postcomposing the left hand side with with $\tearlier^{\nu(\kappa)} \nu^*(\eta^\kappa_{\tlater^\kappa})$ we get
\begin{align*}
\epsilon^{\nu(\kappa)}_{\nu^*} \circ \tearlier^{\nu(\kappa)} \tlater^{\nu(\kappa)} (\nu^* (\epsilon^\kappa)) \circ \tearlier^{\nu(\kappa)} \nu^*(\eta^\kappa_{\tlater^\kappa})&=\epsilon^{\nu(\kappa)}_{\nu^*} \circ \tearlier^{\nu(\kappa)} \nu^* (\tlater^\kappa \epsilon^\kappa) \circ \tearlier^{\nu(\kappa)} \nu^*(\eta^\kappa_{\tlater^\kappa})\\
&= \epsilon^{\nu(\kappa)}_{\nu^*} \circ \tearlier^{\nu(\kappa)} \nu^* (\tlater^\kappa \epsilon^\kappa \circ \eta^\kappa_{\tlater^\kappa})\\
&=\epsilon^{\nu(\kappa)}_{\nu^*} \circ \tearlier^{\nu(\kappa)} \nu^* (\id_{\tlater^\kappa})\\
&=\epsilon^{\nu(\kappa)}_{\nu^*}
\end{align*}
postcomposing the right hand side with $\tearlier^{\nu(\kappa)} \nu^*(\eta^\kappa_{\tlater^\kappa})$ we get
\[\nu^*(\epsilon^\kappa) \circ  \epsilon_{\nu^* \tearlier^\kappa  \tlater^\kappa}^{\nu(\kappa)} \circ \tearlier^{\nu(\kappa)} \nu^*(\eta^\kappa_{\tlater^\kappa}) = \nu^*(\epsilon^\kappa) \circ e^{\kappa,\nu}_{\tlater^\kappa}\]
\end{proof}

\subsection{Projections}

\newcommand{\proj}[2]{\mathsf{pr}^{#1}_{#2}}

Suppose $\wfcxt\Gamma$ and suppose $\Delta'\subseteq \Delta$ and $\Gamma'$ a subcontext of $\Gamma$, i.e., a 
sublist of $\Gamma$ such that $\wfcxt[\Delta']{\Gamma'}$. There is then a projection 
$\pair i{\proj{\Gamma}{\Gamma'}} : \wfcxt\Gamma \to \wfcxt[\Delta']{\Gamma'}$  where $i$ is the inclusion 
and $\proj{\Gamma}{\Gamma'}$ is defined by induction on $\Gamma'$:
\begin{align*}
 \proj{\Gamma}{\Gamma', x:A} & = \subex{\proj{\Gamma}{\Gamma'}}xx \\
 \proj{\Gamma}{\Gamma', \tickA : \kappa} & = \subex{\proj{\Gamma}{\Gamma'}}\tickA\tickA
\end{align*}
The aim of this subsection is to prove the substitution lemma in the special case of projections. This must be done 
before the proof of the general substitution lemma. 

We first prove a series of lemmas.

\begin{lemma}\label{lem:clocksubstpreservesnext}
Let $\nu:\Delta'\to \Delta$ and $\kappa \in \Delta'$. Then $\nu^*(\transp{\p_{\tearlier^\kappa}}) = \transp{\p_{\tearlier^{\nu(\kappa)} \nu^*}}$.
\end{lemma}
\begin{proof}
Recall that $\transp{\p_{\tearlier^\kappa}}$ is the transformation $\id \to \tlater^\kappa$ mapping $\gamma \in F\triple \Theta \vartheta f$ to its restriction in $F\triple \Theta {\vartheta[f(\kappa)-]} f$. We have then that $\nu^*(\transp{\p_{\tearlier^\kappa}}):\nu^* \to \nu^* \tlater^\kappa = \tlater^{\nu(\kappa)} \nu^*$ is the transformation mapping $\gamma \in \nu^* F\triple \Theta \vartheta f = F\triple \Theta \vartheta {f\nu}$ to its restriction in $\nu^*(\tlater^\kappa F) \triple \Theta \vartheta f = F\triple \Theta {\vartheta[f\nu(\kappa)-]} {f\nu}$. On the other hand $\transp{\p_{\tearlier^{\nu(\kappa)} \nu^*}}$ is the transformation $\nu^* \to \tlater^{\nu(\kappa)} \nu^*$ mapping $\gamma \in \nu^*F \triple \Theta \vartheta f = F \triple \Theta \vartheta {f\nu}$ to its restriction $\nu^*F \triple \Theta {\vartheta[f(\nu(\kappa))-]} f = F\triple \Theta {\vartheta[f\nu(\kappa)-]} {f\nu}$. Thus the two transformations are equal.
\end{proof}

\begin{lemma} \label{lem:proj:e}
 Let $\nu : \Delta' \to \Delta$ and $\kappa\in \Delta'$. Then 
 $\nu^*(\p_{\tearlier^\kappa})\circ e^{\nu,\kappa} = \p_{\tearlier^{\nu(\kappa)} \nu^*}$
\end{lemma}

\begin{proof} By the definition of $\transp{\p_{\tearlier^\kappa}}$
\begin{align*}
\nu^*(\p_{\tearlier^\kappa})\circ e^{\nu,\kappa} = \nu^*(\epsilon^\kappa \circ \tearlier^\kappa \transp{\p_{\tearlier^\kappa}}) \circ e^{\nu,\kappa}
\end{align*}
By naturality of $e$ 
\begin{align*}
\nu^*(\epsilon^\kappa \circ \tearlier^\kappa \transp{\p_{\tearlier^\kappa}}) \circ e^{\nu,\kappa} = \nu^* (\epsilon^\kappa) \circ e^{\nu,\kappa}_{\tlater^\kappa} \circ \tearlier^{\nu(\kappa)} \nu^* (\transp{\p_{\tearlier^\kappa}})
\end{align*}
By (\ref{eq:e:comm:eps})
\begin{align*}
\nu^*(\epsilon^\kappa \circ \tearlier^\kappa \transp{\p_{\tearlier^\kappa}}) \circ e^{\nu,\kappa} = \epsilon^{\nu(\kappa)}_{\nu^*} \circ \tearlier^{\nu(\kappa)} \nu^* (\transp{\p_{\tearlier^\kappa}})
\end{align*}
By Lemma~\ref{lem:clocksubstpreservesnext} we have
\begin{align*}
\epsilon^{\nu(\kappa)}_{\nu^*} \circ \tearlier^{\nu(\kappa)} \nu^* (\transp{\p_{\tearlier^\kappa}}) = \epsilon^{\nu(\kappa)}_{\nu^*} \circ \tearlier^{\nu(\kappa)} \transp{\p_{\tearlier^{\nu(\kappa)} \nu^*}}
\end{align*}
which by the definition is equal to $\p_{\tearlier^{\nu(\kappa)} \nu^*}$.
\end{proof}



\begin{lemma} \label{lem:syn:proj:nu}
 Let $\pair i{\proj\Gamma{\Gamma'}} : \wfcxt[\Delta]{\Gamma} \to \wfcxt[\Delta']{\Gamma'}$ be a context projection
 and let $\wfcxt[\Delta]{\Gamma, \Gamma_1}$. Then $\pret{\pair i{\proj{\Gamma, \Gamma_1}{\Gamma'}}} 
 = \pair i{\proj\Gamma{\Gamma'}} \circ \p_{\Gamma_1}$. If, moreover, $\wfcxt[\Delta']{\Gamma', \Gamma_1}$
 then also $\pret{\pair i{\proj{\Gamma, \Gamma_1}{\Gamma'}}} 
 = i^*(\p_{\Gamma_1})\pret{\pair i{\proj{\Gamma, \Gamma_1}{\Gamma', \Gamma_1}}}$.
\end{lemma}


\begin{proof}
 The proof of the first equality is by induction on $\Gamma'$. In the case of $\Gamma' = \Gamma'', x: A$, 
 we get 
\begin{align*}
 \pret{\pair i{\proj{\Gamma, \Gamma_1}{\Gamma'}}} 
 & = \cpair{\pret{\pair i{\proj{\Gamma, \Gamma_1}{\Gamma''}}}}{\q[\pair i{\pret{\proj{\Gamma, \Gamma_1}{\Gamma''}}}]} \\
 & = \cpair{\pret{\pair i{\proj{\Gamma}{\Gamma''}}}\circ \p_{\Gamma_1}}{\q[\pret{\pair i{\proj{\Gamma}{\Gamma''}}}\circ \p_{\Gamma_1}]} \\
 & = \cpair{\pret{\pair i{\proj{\Gamma}{\Gamma''}}}}{\q[\pret{\pair i{\proj{\Gamma}{\Gamma''}}}]} \circ \p_{\Gamma_1}\\
 & = \pret{\pair i{\proj{\Gamma}{\Gamma'}}} \circ \p_{\Gamma_1}
\end{align*}
In the case of $\Gamma' = \Gamma'', \tickA : \kappa$, the context $\Gamma$ must be of the form 
$\Gamma_0, \tickA : \kappa, \Gamma_2$ and 
\begin{align*}
 \pret{\pair i{\proj{\Gamma, \Gamma_1}{\Gamma'}}} 
 & = e^{i, \kappa} \circ \tearlier^{\kappa}(\pret{\pair i{\proj{\Gamma_0}{\Gamma''}}}) \circ \p_{\Gamma_2, \Gamma_1} \\
 & = e^{i, \kappa} \circ \tearlier^{\kappa}(\pret{\pair i{\proj{\Gamma_0}{\Gamma''}}}) 
 \circ \p_{\Gamma_2} \circ \p_{\Gamma_1} \\
 & = \pret{\pair i{\proj{\Gamma}{\Gamma'}}} \circ \p_{\Gamma_1} 
\end{align*}

The proof of the second equality is also by induction on $\Gamma_1$. 
In the case that $\Gamma_1 = \Gamma_1', x: A$ we get 
\begin{align*}
i^*(\p_{\Gamma_1})\circ \pret{\pair i{\proj{\Gamma, \Gamma_1}{\Gamma', \Gamma_1}}} 
& = i^*(\p_{\Gamma_1})\circ i^*(\p_{A}) \circ
\cpair{\pret{\pair i{\proj{\Gamma, \Gamma_1}{\Gamma', \Gamma_1'}}}}{\q[\pret{\pair i{\proj{\Gamma, \Gamma_1}{\Gamma', \Gamma_1'}}}]}  \\
& = i^*(\p_{\Gamma_1}) \circ \p_{i^*(A)} \circ 
\cpair{\pret{\pair i{\proj{\Gamma, \Gamma_1}{\Gamma', \Gamma_1'}}}}{\q[\pret{\pair i{\proj{\Gamma, \Gamma_1}{\Gamma', \Gamma_1'}}}]}  \\
& = i^*(\p_{\Gamma_1}) \circ 
\pret{\pair i{\proj{\Gamma, \Gamma_1}{\Gamma', \Gamma_1'}}} \\
& =\pret{\pair i{\proj{\Gamma, \Gamma_1}{\Gamma'}}}
\end{align*}
In the case that $\Gamma_1 = \Gamma_1', \tickA : \kappa$ we get 
\begin{align*}
 i^*(\p_{\Gamma_1'})\pret{\pair i{\proj{\Gamma, \Gamma_1}{\Gamma', \Gamma_1}}} 
& = i^*(\p_{\Gamma_1})\circ i^*(\p_{\tearlier^{\kappa}}) \circ e^{i, \kappa} \circ \tearlier^\kappa{\pret{\pair i{\proj{\Gamma, \Gamma_1'}{\Gamma', \Gamma_1'}}}} \\
& = i^*(\p_{\Gamma_1'})\circ \p_{\tearlier^{\kappa}} \circ \tearlier^\kappa{\pret{\pair i{\proj{\Gamma, \Gamma_1'}{\Gamma', \Gamma_1'}}}} \\
& = i^*(\p_{\Gamma_1'}) \circ \pret{\pair i{\proj{\Gamma, \Gamma_1'}{\Gamma', \Gamma_1'}}}
\circ \p_{\tearlier^{\kappa}} \\
& = \pret{\pair i{\proj{\Gamma, \Gamma_1'}{\Gamma'}}} \circ \p_{\tearlier^{\kappa}} \\
& = \pret{\pair i{\proj{\Gamma, \Gamma_1}{\Gamma'}}} 
\end{align*}
Using Lemma~\ref{lem:proj:e} in the second equality and the first statement of the lemma in the last equality.
\end{proof}

\begin{lemma} \label{lem:proj:comp}
 If $\pair i{\proj{\Gamma}{\Gamma'}} : \wfcxt{\Gamma} \to \wfcxt[\Delta']{\Gamma'}$ and 
 $\pair j{\proj{\Gamma'}{\Gamma''}} : \wfcxt[\Delta']{\Gamma'} \to \wfcxt[\Delta'']{\Gamma''}$ then 
 \[
 i^*(\pret{\pair j{\proj{\Gamma'}{\Gamma''}}}) \circ \pret{\pair i{\proj{\Gamma}{\Gamma'}}}
 = \pret{\pair{ij}{\proj\Gamma{\Gamma''}}}
 \]
\end{lemma}

\begin{proof}
 The proof is by induction on $\Gamma''$. If $\Gamma'' = \Gamma''', x : A$ then 
\begin{align*}
 i^*(\pret{\pair j{\proj{\Gamma'}{\Gamma''}}}) \circ \pret{\pair i{\proj{\Gamma}{\Gamma'}}}
 & = i^*(\cpair{\pret{\pair j{\proj{\Gamma'}{\Gamma'''}}}}{\q[\pret{\pair j{\proj{\Gamma'}{\Gamma'''}}}]})
 \circ \pret{\pair i{\proj{\Gamma}{\Gamma'}}} \\
  & = \cpair{i^*(\pret{\pair j{\proj{\Gamma'}{\Gamma'''}}})}{\q[i^*(\pret{\pair j{\proj{\Gamma'}{\Gamma'''}}})]})
 \circ \pret{\pair i{\proj{\Gamma}{\Gamma'}}} \\
  & = \cpair{i^*(\pret{\pair j{\proj{\Gamma'}{\Gamma'''}}})\circ \pret{\pair i{\proj{\Gamma}{\Gamma'}}}}
   {\q[i^*(\pret{\pair j{\proj{\Gamma'}{\Gamma'''}}})\circ \pret{\pair i{\proj{\Gamma}{\Gamma'}}}]} \\
  & = \cpair{\pret{\pair {ij}{\proj{\Gamma}{\Gamma'''}}}}
   {\q[\pret{\pair{ij}{\proj{\Gamma}{\Gamma'''}}}]} \\
  & = \pret{\pair {ij}{\proj{\Gamma}{\Gamma''}}}
\end{align*}

If $\Gamma'' = \Gamma''', \tickA : \kappa$ then $\Gamma'$ must be of the form $\Gamma_0', \tickA : \kappa, \Gamma_1'$
and $\Gamma$ of the form $\Gamma_0, \tickA : \kappa, \Gamma_1$. Then 
\begin{align*}
 i^*(\pret{\pair j{\proj{\Gamma'}{\Gamma''}}}) \circ \pret{\pair i{\proj{\Gamma}{\Gamma'}}} 
 & = i^*(e^{j,\kappa} \circ \tearlier^{\kappa}(\pret{\pair j{\proj{\Gamma'_0}{\Gamma'''}}}) \circ \p_{\Gamma_1'}) 
 \circ \pret{\pair i{\proj{\Gamma}{\Gamma'}}}  \\
 & = i^*e^{j,\kappa} \circ i^*\tearlier^{\kappa}(\pret{\pair j{\proj{\Gamma'_0}{\Gamma'''}}}) \circ i^*\p_{\Gamma_1'}
 \circ \pret{\pair i{\proj{\Gamma}{\Gamma'}}} 
\end{align*}
By Lemma~\ref{lem:syn:proj:nu} this equals
\begin{align*}
  & i^*e^{j,\kappa} \circ i^*\tearlier^{\kappa}(\pret{\pair j{\proj{\Gamma'_0}{\Gamma'''}}})
 \circ \pret{\pair i{\proj{\Gamma}{\Gamma_0', \tickA : \kappa}}} \\
& =   i^*e^{j,\kappa} \circ i^*\tearlier^{\kappa}(\pret{\pair j{\proj{\Gamma'_0}{\Gamma'''}}})
 \circ e^{i, \kappa} \tearlier^{\kappa}\pret{\pair i{\proj{\Gamma_0}{\Gamma_0'}}} \circ \p_{\Gamma_1} \\
 & =   i^*e^{j,\kappa} \circ e^{i, \kappa} \circ \tearlier^{\kappa} i^*\pret{\pair j{\proj{\Gamma'_0}{\Gamma'''}}}
  \circ \tearlier^{\kappa}\pret{\pair i{\proj{\Gamma_0}{\Gamma_0'}}} \circ \p_{\Gamma_1} \\
   & =   i^*e^{j,\kappa} \circ e^{i, \kappa} \circ \tearlier^{\kappa} (i^*\pret{\pair j{\proj{\Gamma'_0}{\Gamma'''}}}
  \circ \pret{\pair i{\proj{\Gamma_0}{\Gamma_0'}}}) \circ \p_{\Gamma_1} \\
     & =   i^*e^{j,\kappa} \circ e^{i, \kappa} \circ \tearlier^{\kappa} \pret{\pair{ij}{\proj{\Gamma_0}{\Gamma'''}}}
   \circ \p_{\Gamma_1} \\
     & =   e^{ij,\kappa} \circ \tearlier^{\kappa} \pret{\pair{ij}{\proj{\Gamma_0}{\Gamma'''}}}
   \circ \p_{\Gamma_1} \\
     & =   \pret{\pair{ij}{\proj{\Gamma_0}{\Gamma''}}}
\end{align*}
using naturality of $e$ and Lemma~\ref{lem:e:nat}.
\end{proof}

We can now sketch the proof of the substitution lemma in the special case of projections. 

\begin{lemma} \label{lem:weakening}
Let $\pair{i}{\proj{\Gamma}{\Gamma'}} : \wfcxt[\Delta]{\Gamma} \to \wfcxt[\Delta']{\Gamma'}$ be a 
projection and ${\Gamma'}\vdash_{\Delta'}  J$ a judgement of a wellformed type or a typing judgement. 
Then $\pret{J} = (i^*\pret{ J})[\pret{\pair i{\proj{\Gamma}{\Gamma'}}}]$.
\end{lemma}

\begin{proofsketch}
 The proof is almost the same of the proof of the general substitution lemma, except in the two places where 
 the general proof appeals to Lemma~\ref{lem:weakening}, which is in the case of variable introduction and 
 in the proof of Lemma~\ref{lem:subst:comp:clock:weak}, and so we here provide proofs of these cases for 
 the case of projections. Note that this proof works because the inductive hypothesis always appeals to 
 another substitution, rather than a more general substitution. For example, in the proof of the case of application
 to $\tickc$ provided in the main paper, the induction hypothesis uses the substitution 
 $\pair{\subex\nu\kappa\kappa}\sigma$, which is a projection if $\pair\nu\sigma$ is.
 
 For Lemma~\ref{lem:subst:comp:clock:weak}, in the case of $\pair\nu\sigma$ being a projection 
 $\pair h{\proj{\Gamma}{\Gamma'}}$ both sides of the equality are equal to 
 $\pret{\pair{hi}{\proj{\Gamma}{\Gamma'}}}$ by Lemma~\ref{lem:proj:comp}. 
 In the case of variable introduction, if $x : A$ is in $\Gamma' $  write 
 $\Gamma' = \Gamma'_0, x:A, \Gamma_1'$ and $\Gamma = \Gamma_0, x:A, \Gamma_1$. Then
\begin{align*}
 i^*\pret{\istypeshort[\Delta']{\Gamma'}{x}}[\pret{\pair{i}{\proj{\Gamma}{\Gamma'}}}] 
 & = i^*(\q[\p_{\Gamma_1'}])[\pret{\pair{i}{\proj{\Gamma}{\Gamma'}}}] \\
 & = \q[i^*\p_{\Gamma_1'} \circ \pret{\pair{i}{\proj{\Gamma}{\Gamma'}}}] \\
 & = \q[\pret{\pair{i}{\proj{\Gamma}{\Gamma_0', x:A}}}] \\
 & = \q[\cpair{\pret{\pair{i}{\proj{\Gamma}{\Gamma_0'}}}}{\q[\p_{\Gamma_1}]}] \\
 & = \q[\p_{\Gamma_1}] \\
 & = \pret{\istypeshort[\Delta]{\Gamma}x}
\end{align*}
using Lemma~\ref{lem:syn:proj:nu} in the third equality. 
\end{proofsketch}

\subsection{Composition of substitutions}

Here we will prove identities of particular compositions of substitutions. The need for these lemmas will only become clearer in the coming subsections. The reader can skip this subsection and consult it later when needed.

\begin{lemma} \label{lem:subst:comp:kappa:to:kappa'}
 If $\pair\nu\sigma : \wfcxt\Gamma \to \wfcxt[\Delta']{\Gamma'}$ and $\kappa'\in \Delta'$ and $\kappa \notin \Delta'$ then 
 \[\pair{\subex\nu\kappa{\nu(\kappa')}}\sigma : \wfcxt\Gamma \to \wfcxt[\Delta', \kappa]{\Gamma'}\]
 is a wellformed substitution and 
 \[
 \pret{\pair{\subex\nu\kappa{\nu(\kappa')}}\sigma}
 = \nu^*(\pret{\pair{\basicsub\kappa{\kappa'}}{\id_{\Gamma'}}}) \circ \pret{\pair\nu\sigma}
 \]
\end{lemma}

\begin{proof}
 The proof is by induction on $\Gamma'$. The case of an empty context is trivial, and the case of extension with ordinary 
 variables is easy, so we focus on the case of $\Gamma' = \Gamma'', \tickA : \kappa''$ for some $\kappa'' \in \Delta'$. 
 There are two cases to consider for $\sigma$: Either $\sigma = \subex\tau\tickA\tickB$ or $\sigma = \subex\tau\tickA\tickc$. 
 
 In the first case $\Gamma$ is of form $\Gamma_0, \tickB : \nu(\kappa''), \Gamma_1$ and 
 $\tau : \wfcxt{\Gamma_0} \to \wfcxt[\Delta']{\Gamma''}$. Then 
\begin{align*}
 \pret{\pair{\subex\nu\kappa{\nu(\kappa')}}\sigma} & = 
 e^{\subex\nu\kappa{\nu(\kappa')}, \kappa''} \circ 
 \tearlier^{\nu(\kappa'')}(\pret{\pair{\subex\nu\kappa{\nu(\kappa')}}\tau}) \circ \p_{\Gamma_1} \\
 & =  e^{\subex\nu\kappa{\nu(\kappa')}, \kappa''} \circ 
 \tearlier^{\nu(\kappa'')}(\nu^*\pret{\pair{\basicsub\kappa{\nu(\kappa')}}{\id_{\Gamma''}}} \circ \pret{\pair{\nu}\tau}) 
 \circ \p_{\Gamma_1}\\
 & =   \nu^*e^{\basicsub\kappa{\kappa'}, \kappa''} \circ e^{\nu, \kappa''} \circ
 \tearlier^{\nu(\kappa'')}(\nu^*\pret{\pair{\basicsub\kappa{\nu(\kappa')}}{\id_{\Gamma''}}} \circ \pret{\pair{\nu}\tau}) 
 \circ \p_{\Gamma_1} \\
  & =   \nu^*e^{\basicsub\kappa{\kappa'}, \kappa''} \circ 
 \nu^*(\tearlier^{\kappa''}\pret{\pair{\basicsub\kappa{\nu(\kappa')}}{\id_{\Gamma''}}}) \circ e^{\nu, \kappa''} \circ
 \tearlier^{\kappa''}(\pret{\pair{\nu}\tau})  \circ \p_{\Gamma_1} \\
  & =   \nu^*(\pret{\pair{\basicsub\kappa{\nu(\kappa')}}{\id_{\Gamma'', \tickA: \kappa"}}})  \circ
 \pret{\pair{\nu}{\subex{\tau}\tickA\tickB}}
\end{align*}

 Finally, in the case where $\sigma = \subex{\tau}\tickA\tickc$, we must have $\Delta' = \Delta'', \kappa''$ and 
$\nu = \subex\mu{\kappa''}{\mu(\kappa''')}$ for some $\mu : \Delta'' \to \Delta$ and $\kappa'''\in\Delta''$. Note that $\kappa'$ is
in $\Delta'$ and so can be equal to $\kappa''$. We need to consider the two cases of $\kappa' = \kappa''$ or not.

If $\kappa'' \neq \kappa'$, then 
\begin{align*}
 \subex\nu\kappa{\nu(\kappa')} & = \subex{\subex\mu{\kappa''}{\mu(\kappa''')}}\kappa{\nu(\kappa')} 
 = \subex{\subex\mu{\kappa}{\mu(\kappa')}}{\kappa''}{\mu(\kappa''')}
\end{align*}
and so 
\begin{align*}
\pret{\pair{\subex\nu\kappa{\nu(\kappa')}}{\sigma}} 
& = \pret{\pair{\subex{\subex\mu{\kappa}{\mu(\kappa')}}{\kappa''}{\mu(\kappa''')}}{\subex{\tau}\tickA\tickc}}  \\
& = (\subex\mu{\kappa}{\mu(\kappa')})^*\pret{\pair{\basicsub{\kappa''}{\kappa'''}}{\basicsub\tickA\tickc}} \circ
\pret{\pair{\subex\mu{\kappa}{\mu(\kappa')}}{\tau}}  \\
& = (\subex\mu{\kappa}{\mu(\kappa')})^*\pret{\pair{\basicsub{\kappa''}{\kappa'''}}{\basicsub\tickA\tickc}} \circ
\mu^*\pret{\pair{\basicsub\kappa{\kappa'}}{\id}} \circ \pret{\pair{\mu}{\tau}}  \\
& = \mu^*\left((\basicsub{\kappa}{\kappa'})^*\pret{\pair{\basicsub{\kappa''}{\kappa'''}}{\basicsub\tickA\tickc}} \circ
\pret{\pair{\basicsub\kappa{\kappa'}}{\id}} \right)\circ \pret{\pair{\mu}{\tau}}  
\end{align*}
using the induction hypothesis in third equality. Lemma~\ref{lem:sub:clock:extension} gives the equality
\begin{align*}
&(\basicsub{\kappa}{\kappa'})^*\pret{\pair{\basicsub{\kappa''}{\kappa'''}}{\basicsub\tickA\tickc}} \circ
\pret{\pair{\basicsub\kappa{\kappa'}}{\id}} \\
& = \basicsub{\kappa''}{\kappa'''}^*\pret{\pair{\basicsub\kappa{\kappa'}}{\id}} \circ 
\pret{\pair{\basicsub{\kappa''}{\kappa'''}}{\basicsub\tickA\tickc}} 
\end{align*}
and so
\begin{align*}
\pret{\pair{\subex\nu\kappa{\nu(\kappa')}}{\sigma}} 
& = \mu^*\left(\basicsub{\kappa''}{\kappa'''}^*\pret{\pair{\basicsub\kappa{\kappa'}}{\id}} \circ 
\pret{\pair{\basicsub{\kappa''}{\kappa'''}}{\basicsub\tickA\tickc}} \right)\circ \pret{\pair{\mu}{\tau}}  \\
& = \mu^*\basicsub{\kappa''}{\kappa'''}^*\pret{\pair{\basicsub\kappa{\kappa'}}{\id}} \circ 
\mu^*\pret{\pair{\basicsub{\kappa''}{\kappa'''}}{\basicsub\tickA\tickc}} \circ \pret{\pair{\mu}{\tau}}  \\
& = (\subex\mu{\kappa''}{\mu(\kappa''')})^*\pret{\pair{\basicsub\kappa{\kappa'}}{\id}} 
\circ \pret{\pair{\subex\mu{\kappa''}{\mu(\kappa''')}}{\subex\tau\tickA\tickc}}  \\
& = \nu^*\pret{\pair{\basicsub\kappa{\kappa'}}{\id}} 
\circ \pret{\pair{\nu}{\sigma}}
\end{align*}
as desired. 

In the case where If $\kappa'' = \kappa'$ the proof is very similar, but we do it for completeness. In this case 
\begin{align*}
 \subex\nu\kappa{\nu(\kappa')} & = \subex{\subex\mu{\kappa''}{\mu(\kappa''')}}\kappa{\nu(\kappa')} 
 = \subex{\subex\mu{\kappa}{\mu(\kappa''')}}{\kappa''}{\mu(\kappa''')}
\end{align*}
and so 
\begin{align*}
\pret{\pair{\subex\nu\kappa{\nu(\kappa')}}{\sigma}} 
& = \pret{\pair{\subex{\subex\mu{\kappa}{\mu(\kappa''')}}{\kappa''}{\mu(\kappa''')}}{\subex{\tau}\tickA\tickc}}  \\
& = (\subex\mu{\kappa}{\mu(\kappa''')})^*\pret{\pair{\basicsub{\kappa''}{\kappa'''}}{\basicsub\tickA\tickc}} \circ
\pret{\pair{\subex\mu{\kappa}{\mu(\kappa''')}}{\tau}}  \\
& = (\subex\mu{\kappa}{\mu(\kappa''')})^*\pret{\pair{\basicsub{\kappa''}{\kappa'''}}{\basicsub\tickA\tickc}} \circ
\mu^*\pret{\pair{\basicsub\kappa{\kappa'''}}{\id}} \circ \pret{\pair{\mu}{\tau}}  \\
& = \mu^*\left((\basicsub{\kappa}{\kappa'''})^*\pret{\pair{\basicsub{\kappa''}{\kappa'''}}{\basicsub\tickA\tickc}} \circ
\pret{\pair{\basicsub\kappa{\kappa'''}}{\id}} \right)\circ \pret{\pair{\mu}{\tau}}  
\end{align*}
using the induction hypothesis in third equality. Lemma~\ref{lem:sub:clock:extension} gives the equality
\begin{align*}
&(\basicsub{\kappa}{\kappa'''})^*\pret{\pair{\basicsub{\kappa''}{\kappa'''}}{\basicsub\tickA\tickc}} \circ
\pret{\pair{\basicsub\kappa{\kappa'''}}{\id}} \\
& = \basicsub{\kappa''}{\kappa'''}^*\pret{\pair{\basicsub\kappa{\kappa'''}}{\id}} \circ 
\pret{\pair{\basicsub{\kappa''}{\kappa'''}}{\basicsub\tickA\tickc}} 
\end{align*}
and so
\begin{align*}
\pret{\pair{\subex\nu\kappa{\nu(\kappa')}}{\sigma}} 
& = \mu^*\left(\basicsub{\kappa''}{\kappa'''}^*\pret{\pair{\basicsub\kappa{\kappa'''}}{\id}} \circ 
\pret{\pair{\basicsub{\kappa''}{\kappa'''}}{\basicsub\tickA\tickc}} \right)\circ \pret{\pair{\mu}{\tau}}  \\
& = \mu^*\basicsub{\kappa''}{\kappa'''}^*\pret{\pair{\basicsub\kappa{\kappa'''}}{\id}} \circ 
\mu^*\pret{\pair{\basicsub{\kappa''}{\kappa'''}}{\basicsub\tickA\tickc}} \circ \pret{\pair{\mu}{\tau}}  \\
& = (\subex\mu{\kappa''}{\mu(\kappa''')})^*\pret{\pair{\basicsub\kappa{\kappa'''}}{\id}} 
\circ \pret{\pair{\subex\mu{\kappa''}{\mu(\kappa''')}}{\subex\tau\tickA\tickc}}  \\
& = \nu^*\pret{\pair{\basicsub\kappa{\kappa'''}}{\id}} 
\circ \pret{\pair{\nu}{\sigma}}
\end{align*}
as desired. This ends the final case of the proof.
\end{proof}

\begin{lemma} \label{lem:tickc:proj}
 Let $\kappa'\in \Delta$ and $\kappa \notin\Delta$, and suppose $\wfcxt\Gamma$. Then 
\[
 \basicsub\kappa{\kappa'}^*(p_{\tearlier^\kappa})\circ
 \pret{\pair{\basicsub\kappa{\kappa'}}{\basicsub\tickA\tickc}} 
 = \pret{\pair{\basicsub\kappa{\kappa'}}{\id_{\Gamma}}}
 : \pret{\wfcxt\Gamma} \to \pret{\wfcxt[\Delta,\kappa]\Gamma}
 \]
\end{lemma}

\begin{proof}
  Let $\gamma \in \pret{\wfcxt\Gamma}_{\triple\Theta\vartheta f}$. Then 
\begin{align*}
 \pret{\pair{\basicsub\kappa{\kappa'}}{\basicsub\tickA\tickc}}_{\triple\Theta\vartheta f} (\gamma) 
 & = \pair{\{\kappa\}}{\inv{\pret{\pair i{\id}}_{\triple{\Theta,\#}{\subex\vartheta\#{n+1}}{\iota f}}}(\iota \cdot \gamma)}\\ 
 & = \pair{\{\kappa\}}{i^*(\pret{\pair{\basicsub\kappa{\kappa'}}{\id}})_{\triple{\Theta,\#}{\subex\vartheta\#{n+1}}{\iota \subex f\kappa \#}}(\iota \cdot \gamma)} \\
 & = \pair{\{\kappa\}}{\pret{\pair{\basicsub\kappa{\kappa'}}{\id}}_{\triple{\Theta,\#}{\subex\vartheta\#{n+1}}{\iota f}}(\iota \cdot \gamma)}
\end{align*}
using Lemma~\ref{lem:(i,id):invertible}. So 
\begin{align*}
 (\basicsub\kappa{\kappa'}^*(p_{\tearlier^\kappa})\circ
 \pret{\pair{\basicsub\kappa{\kappa'}}{\basicsub\tickA\tickc}})_{\triple\Theta\vartheta f} (\gamma) 
 & = \chi \cdot \pret{\pair{\basicsub\kappa{\kappa'}}{\id}}_{\triple{\Theta,\#}{\subex\vartheta\#{n+1}}{\iota f}}(\iota \cdot \gamma)\\
 & =\pret{\pair{\basicsub\kappa{\kappa'}}{\id}}_{\triple{\Theta}{\vartheta}{f}}( \chi \cdot \iota \cdot \gamma) \\
 & =\pret{\pair{\basicsub\kappa{\kappa'}}{\id}}_{\triple{\Theta}{\vartheta}{f}}(\gamma)
\end{align*}
where $\chi$ maps $\#$ to $f(\kappa')$ and fixes everything else. 
\end{proof}

\begin{lemma}\label{lem:weakeningsubstitutions}
Suppose $\pair{\nu}{\sigma} : \wfcxt{\Gamma_0} \to\wfcxt[\Delta']{\Gamma'}$ and let $\wfcxt{\Gamma_0,\Gamma_1}$ then 
\[\pret{\pair{\nu}{\sigma}} = \pret{\pair{\nu}{\sigma}} \circ \p_{\Gamma_1}\]
where $\pair{\nu}{\sigma}$ on the left is the substitution  $\wfcxt{\Gamma_0,\Gamma_1} \to\wfcxt[\Delta']{\Gamma'}$ and on the right is the substitution $\wfcxt{\Gamma_0} \to\wfcxt[\Delta']{\Gamma'}$.
\end{lemma}
\begin{proof}
To avoid confusion we will annotate the syntactic substitution with its source context. The statement of the lemma can then be rewritten as
\[\pret{\pair{\nu}{\sigma}}_{\Gamma_0,\Gamma_1} = \pret{\pair{\nu}{\sigma}}_{\Gamma_0} \circ \p_{\Gamma_1}\]

The proof is by induction on $\Gamma'$. In the case where $\Gamma' = \Gamma'', x:A$ we have that $\sigma = \sigma'[x\mapsto t]$
\begin{align*}
\pret{\pair{\nu}{\sigma}_{\Gamma_0,\Gamma_1}} = \cpair {\pret{\pair{\nu}{\sigma'}}_{\Gamma_0,\Gamma_1}}{\pret{t}_{\Gamma_0,\Gamma_1}}
\end{align*}
By the induction hypothesis
\[\cpair {\pret{\pair{\nu}{\sigma'}}_{\Gamma_0,\Gamma_1}}{\pret{t}_{\Gamma_0,\Gamma_1}} = \cpair {\pret{\pair{\nu}{\sigma'}}_{\Gamma_0} \circ \p_{\Gamma_1}}{\pret{t}_{\Gamma_0,\Gamma_1}}\]

By Lemma~\ref{lem:weakening}, $\pret{t}_{\Gamma_0,\Gamma_1} = \id^*\pret{t}[\pret{\pair{\id} {\proj{\Gamma_0,\Gamma_1}{\Gamma_0}}}]$ and by Lemma~\ref{lem:syn:proj:nu} $\pret{\pair{\id} {\proj{\Gamma_0,\Gamma_1}{\Gamma_0}}}] =  \p_{\Gamma_1}$. We have then that
\[\cpair {\pret{\pair{\nu}{\sigma'}}_{\Gamma_0,\Gamma_1}}{\pret{t}_{\Gamma_0,\Gamma_1}} = \cpair {\pret{\pair{\nu}{\sigma'}}_{\Gamma_0} \circ \p_{\Gamma_1}}{\pret{t}_{\Gamma_0}[\p_{\Gamma_1}]} = \cpair {\pret{\pair{\nu}{\sigma'}}_{\Gamma_0}}{\pret{t}_{\Gamma_0}} \circ \p_{\Gamma_1} =\pret{\pair{\nu}{\sigma}}_{\Gamma_0}\circ \p_{\Gamma_1} \]

In the case where $\Gamma' = \Gamma'',\alpha:\kappa$ we have that either 
\begin{itemize}
\item $\sigma = \sigma'[\alpha \mapsto \beta]$ and $\Gamma_0 = \Gamma_{00},\beta:\nu(\kappa), \Gamma_{01}$
We have then that
\begin{align*}
\pret{\pair{\nu}{\sigma}}_{\Gamma_0,\Gamma_1} &= \pret{\pair{\nu}{\sigma'[\alpha \mapsto \beta]}}_{\Gamma_0,\Gamma_1}\\
&= e^{\nu,\kappa} \circ \tearlier^{\nu(\kappa)} \pret{\pair{\nu}{\sigma'}} \circ \p_{\Gamma_{01},\Gamma_1}\\
&=e^{\nu,\kappa} \circ \tearlier^{\nu(\kappa)} \pret{\pair{\nu}{\sigma'}} \circ \p_{\Gamma_{01}} \circ\p_{\Gamma_1}\\
&=\pret{\pair{\nu}{\sigma}}_{\Gamma_0} \circ \p_{\Gamma_1}
\end{align*}
\item $\nu = \mu[\kappa \mapsto \mu(\kappa')]$  and $\sigma = \sigma'[\alpha \mapsto \tickc]$ and $\pair{\mu}{\sigma'}:\wfcxt{\Gamma_0} \to \wfcxt[\Delta'']{\Gamma''}$ and $\Delta' = \Delta'',\kappa$.
We have then that
\begin{align*}
\pret{\pair{\nu}{\sigma}}_{\Gamma_0,\Gamma_1} &= \pret{\pair{\mu[\kappa\mapsto\mu(\kappa')]}{\sigma'[\alpha \mapsto \tickc]}}_{\Gamma_0,\Gamma_1}  \\
&=\mu^*\pret{\pair{[\kappa\mapsto\kappa'}{\alpha\mapsto \tickc}} \circ \pret{\pair{\mu}{\sigma'}}_{\Gamma_0,\Gamma1}
\end{align*}
By the induction hypothesis $\pret{\pair{\mu}{\sigma'}}_{\Gamma_0,\Gamma1} = \pret{\pair{\mu}{\sigma'}}_{\Gamma_0} \circ \p_{\Gamma_1}$ thus 
\begin{align*}
\pret{\pair{\nu}{\sigma}}_{\Gamma_0,\Gamma_1}  &= \mu^*\pret{\pair{[\kappa\mapsto\kappa'}{\alpha\mapsto \tickc}} \circ \pret{\pair{\mu}{\sigma'}}_{\Gamma_0} \circ \p_{\Gamma_1} \\
&= \pret{\pair{\mu[\kappa \mapsto \mu(\kappa')]}{\sigma'[\alpha \mapsto \tickc]}}_{\Gamma_0} \circ \p_{\Gamma_1} \\
&= \pret{\pair{\nu}{\sigma}} \circ \p_{\Gamma_1}
\end{align*}
\end{itemize}
\end{proof}

\begin{lemma}\label{lem:sub:comp:proj}
 Suppose $\pair{\nu}{\tau} : \wfcxt{\Gamma} \to
 \wfcxt[\Delta']{\Gamma_0'}$ and $\pair{\nu}{\tau\tau'} : \wfcxt{\Gamma} \to
 \wfcxt[\Delta']{\Gamma_0', \Gamma_1'}$ is an extension of $\pair{\nu}{\tau}$. Then
\begin{align*}
 \pret{\pair{\nu}{\tau}}= \nu^*(\p_{\Gamma_1'}) \circ 
 \pret{\pair{\nu}{\tau\tau'}} : \pret{\wfcxt{\Gamma}}
 \to \nu^*\pret{\wfcxt[\Delta']{\Gamma_0'}}
\end{align*}
\end{lemma}

\begin{proof}
The proof is by induction on $\Gamma'_1$. In the case where $\Gamma'_1 = \Gamma''_1, x:A$ we have that $\tau' = \tau''[x\mapsto t]$
\begin{align*}
\nu^*(\p_{\Gamma_1'}) \circ \pret{\pair{\nu}{\tau\tau'}} & = \nu^*(\p_{\Gamma''_1}) \circ \nu^*(\p_A) \circ \pret{\pair{\nu}{\tau\tau''[x\mapsto t]}}\\
&= \nu^*(\p_{\Gamma''_1}) \circ \nu^*(\p_A) \circ \cpair{\pret{\pair{\nu}{\tau\tau''}}}{\pret{t}}\\
&=\nu^*(\p_{\Gamma''_1}) \circ \pret{\pair{\nu}{\tau\tau''}}\\
&=\pret{\pair{\nu}{\tau}}
\end{align*}
In the case where $\Gamma'_1 = \Gamma''_1,\alpha:\kappa$ we have that either 
\begin{itemize}
\item $\tau' = \sigma''[\alpha \mapsto \beta]$ and $\Gamma = \Gamma_0,\beta:\nu(\kappa), \Gamma_{1}$. In which case we have
\begin{align*}
\nu^*(\p_{\Gamma_1'}) \circ \pret{\pair{\nu}{\tau\tau'}} &= \nu^*(\p_{\Gamma_1'}) \circ \pret{\pair{\nu}{\tau\tau''[\alpha \mapsto \beta]}}\\
&= \nu^*(\p_{\Gamma_1'}) \circ e^{\nu,\kappa} \circ \tearlier^{\nu(\kappa)} \pret{\pair{\nu}{\tau\tau''}} \circ \p_{\Gamma_1}\\
&=\nu^*(\p_{\Gamma''_1})\circ \nu^*(\p_{\tearlier^{\kappa}}) \circ e^{\nu,\kappa} \circ \tearlier^{\nu(\kappa)} \pret{\pair{\nu}{\tau\tau''}} \circ \p_{\Gamma_1}\\
&=\nu^*(\p_{\Gamma''_1}) \circ \p_{\tearlier^{\nu(\kappa)}\nu^*} \circ \tearlier^{\nu(\kappa)} \pret{\pair{\nu}{\tau\tau''}} \circ \p_{\Gamma_1}\\
&=\nu^*(\p_{\Gamma''_1}) \circ \pret{\pair{\nu}{\tau\tau''}}\circ \p_{\tearlier^{\nu(\kappa)}\nu^*} \circ \p_{\Gamma_1}
\end{align*}
where the last equality follows by naturality of $\tearlier^\kappa$ and the second to last by Lemma~\ref{lem:proj:e}. By the induction hypothesis we get then that
\[\nu^*(\p_{\Gamma''_1}) \circ \pret{\pair{\nu}{\tau\tau''}}\circ \p_{\tearlier^{\nu(\kappa)}\nu^*} \circ \p_{\Gamma_1} = \pret{\pair{\nu}{\tau}}\circ \p_{\tearlier^{\nu(\kappa)}\nu^*} \circ \p_{\Gamma_1} = \pret{\pair{\nu}{\tau}}\circ \p_{\beta:\nu(\kappa),\Gamma_1} \]
which by Lemma~\ref{lem:weakeningsubstitutions} is equal to $\pret{\pair{\nu}{\tau}}$

\item $\nu = \mu[\kappa \mapsto \mu(\kappa')]$  and $\tau' = \tau''[\alpha \mapsto \tickc]$ and $\pair{\mu}{\tau\tau'}:\wfcxt{\Gamma} \to \wfcxt[\Delta'']{\Gamma'_0, \Gamma''_1}$ and $\Delta' = \Delta'',\kappa$. In which case we have
\begin{align*}
\nu^*(\p_{\Gamma_1'}) \circ \pret{\pair{\nu}{\tau\tau'}} &=\nu^*(\p_{\Gamma''_1})\circ \nu^*(\p_{\tearlier^\kappa}) \circ \pret{\pair{\mu[\kappa \mapsto \mu(\kappa')]}{\tau\tau''[\alpha \mapsto\tickc]}}\\
&=\nu^*(\p_{\Gamma''_1})\circ \nu^*(\p_{\tearlier^\kappa}) \circ {\mu}^*\pret{\pair{[\kappa\mapsto \kappa']}{[\alpha \mapsto \tickc]}} \circ \pret{\pair{\mu}{\tau\tau''}}\\
&= \nu^*(\p_{\Gamma''_1})\circ \mu^*([\kappa \mapsto \kappa']^* (\p_{\tearlier^\kappa})) \circ {\mu}^*\pret{\pair{[\kappa\mapsto \kappa']}{[\alpha \mapsto \tickc]}} \circ \pret{\pair{\mu}{\tau\tau''}}\\
&=\nu^*(\p_{\Gamma''_1})\circ \mu^* ([\kappa \mapsto\kappa']^* (\p_{\tearlier^\kappa})\circ\pret{\pair{[\kappa\mapsto \kappa']}{[\alpha \mapsto \tickc]}}) \circ \pret{\pair{\mu}{\tau\tau''}}
\end{align*}
By Lemma~\ref{lem:tickc:proj} this equal to
\begin{align*}
\nu^*(\p_{\Gamma''_1}) \circ \mu^* \pret{\pair{\basicsub\kappa{\kappa'}}{\id}}\circ \pret{\pair{\mu}{\tau\tau''}} 
&=\mu^*([\kappa\to\kappa']^* (\p_{\Gamma''_1}) \circ \pret{\pair{\basicsub\kappa{\kappa'}}{\id}}) \circ \pret{\pair{\mu}{\tau\tau''}}
\end{align*}
Which by the induction hypothesis is equal to
\begin{align*}
\mu^*(\pret{\pair{\basicsub\kappa{\kappa'}}{\proj{\Gamma'_0,\Gamma''_1}{\Gamma'_0}}})  \circ \pret{\pair{\mu}{\tau\tau''}}
\end{align*}
By Lemma~\ref{lem:syn:proj:nu} this is equal to
\begin{align*}
\mu^*(\pret{\pair{\basicsub\kappa{\kappa'}}{\id}} \circ \p_{\Gamma''_1}) \circ \pret{\pair{\mu}{\tau\tau''}}
&=\mu^*(\p_{\Gamma''_1})\circ \pret{\pair{\mu}{\tau\tau''}}\\
&= \pret{\pair{\mu}{\tau}}
\end{align*}
Where the last equality follows by the induction hypothesis.\qedhere
\end{itemize}
\end{proof}

\begin{lemma} \label{lem:subst:comp:i:kappa:to:kappa'}
 Suppose $\wfcxt\Gamma$ and $\kappa\notin\Delta$, $\kappa'\in \Delta$. Then the composition
  \[
 (\id_\Delta\basicsub{\kappa}{\kappa'})^* \pret{\pair{i}{\id_\Gamma}} \circ \pret{\pair{\id_\Delta\basicsub{\kappa}{\kappa'}}{\id_\Gamma}}
 \]
is the identity.
\end{lemma}

\begin{proof}
 The proof is by induction on $\Gamma$. The case of $\Gamma$ empty is trivial. In the case of 
 $\Gamma = \Gamma', x:A$, there is a projection $p$ such that $\id_{\Gamma} = \subex pxx$. Then
\begin{align*}
 & (\id_\Delta\basicsub{\kappa}{\kappa'})^* \pret{\pair{i}{\id_\Gamma}} \circ \pret{\pair{\id_\Delta\basicsub{\kappa}{\kappa'}}{\id_\Gamma}} \\
 & = (\id_\Delta\basicsub{\kappa}{\kappa'})^* \cpair{\pret{\pair{i}{p}}}\q \circ 
 \cpair{\pret{\pair{\id_\Delta\basicsub{\kappa}{\kappa'}}{p}}}{\q}\\
 & = (\id_\Delta\basicsub{\kappa}{\kappa'})^* \cpair{\pret{\pair{i}{\id_{\Gamma'}}}\circ\p_A}\q \circ 
\cpair{\pret{\pair{\id_\Delta\basicsub{\kappa}{\kappa'}}{\id_{\Gamma'}}}\circ \p_A}{\q}
\end{align*}
where the last of these equalities uses Lemma~\ref{lem:syn:proj:nu}. Since functors of the form $\nu^*$ commute 
with the CwF structure, including comprehension, $\p$ and $\q$, the above can be rewritten to
\begin{align*}
  &  \cpair{(\id_\Delta\basicsub{\kappa}{\kappa'})^*(\pret{\pair{i}{\id_{\Gamma'}}})\circ\p_A}\q \circ 
\cpair{\pret{\pair{\id_\Delta\basicsub{\kappa}{\kappa'}}{\id_{\Gamma'}}}\circ \p_A}{\q} \\
  & =  \cpair{(\id_\Delta\basicsub{\kappa}{\kappa'})^*(\pret{\pair{i}{\id_{\Gamma'}}})\circ
  \pret{\pair{\id_\Delta\basicsub{\kappa}{\kappa'}}{\id_{\Gamma'}}}\circ \p_A}\q
\end{align*}
 which by the induction hypothesis equals $\cpair{\p_A}\q$ which is the identity. 
 
 In the case of $\Gamma = \Gamma', \tickA : \kappa''$ for $\kappa'' \in \Delta$ we get 
\begin{align*}
 & (\id_\Delta\basicsub{\kappa}{\kappa'})^* \pret{\pair{i}{\id_\Gamma}} \circ \pret{\pair{\id_\Delta\basicsub{\kappa}{\kappa'}}{\id_\Gamma}} \\ 
 & = (\id_\Delta\basicsub{\kappa}{\kappa'})^*(e^{i, \kappa''} \circ \tearlier^{\kappa''}(\pret{\pair{i}{\id_{\Gamma'}}})) 
 \circ e^{\id_\Delta\basicsub{\kappa}{\kappa'}, \kappa''} \circ 
 \tearlier^{\kappa''}(\pret{\pair{\id_\Delta\basicsub{\kappa}{\kappa'}}{\id_{\Gamma'}}}) \\ 
 & = (\id_\Delta\basicsub{\kappa}{\kappa'})^*e^{i, \kappa''} \circ 
 (\id_\Delta\basicsub{\kappa}{\kappa'})^*\tearlier^{\kappa''}(\pret{\pair{i}{\id_{\Gamma'}}})
 \circ e^{\id_\Delta\basicsub{\kappa}{\kappa'}, \kappa''} \circ 
 \tearlier^{\kappa''}(\pret{\pair{\id_\Delta\basicsub{\kappa}{\kappa'}}{\id_{\Gamma'}}}) \\ 
 & = (\id_\Delta\basicsub{\kappa}{\kappa'})^*e^{i, \kappa''} \circ e^{\id_\Delta\basicsub{\kappa}{\kappa'}, \kappa''}
 \circ 
 \tearlier^{\kappa''}(\id_\Delta\basicsub{\kappa}{\kappa'})^*(\pret{\pair{i}{\id_{\Gamma'}}})
  \circ 
 \tearlier^{\kappa''}(\pret{\pair{\id_\Delta\basicsub{\kappa}{\kappa'}}{\id_{\Gamma'}}}) \\ 
 & = e^{\id_{\Delta}, \kappa''} \circ 
 \tearlier^{\kappa''}((\id_\Delta\basicsub{\kappa}{\kappa'})^*(\pret{\pair{i}{\id_{\Gamma'}}})
  \circ 
 \pret{\pair{\id_\Delta\basicsub{\kappa}{\kappa'}}{\id_{\Gamma'}}}) 
\end{align*}
using naturality of $e$ in the third equality and Lemma~\ref{lem:e:nat} in the fourth. Now, by the induction hypothesis
and Lemma~\ref{lem:e:nat} once again, the above reduces to the identity.
Note that the above argument also holds for $\kappa'' = \kappa'$. 
\end{proof}

\begin{lemma} \label{lem:weak:sub:clock:extension}
 Let $\pair\nu\sigma : \wfcxt\Gamma\to \wfcxt[\Delta']{\Gamma'}$ be a \clott\ substitution, let $\kappa$ be fresh and 
 $\kappa'\in \Delta$. Let $i: \Delta \to \Delta, \kappa$ and $j: \Delta' \to \Delta', \kappa$ be the inclusion maps. Then under the assumption that $(\subex{\nu}\kappa\kappa)^*(\inv{\pret{\pair{j}{\id_{\Gamma'}}}})\circ i^*\pret{\pair{\nu}{\sigma}} = 
\pret{\pair{\subex\nu\kappa\kappa}{\sigma}}\circ \inv{\pret{\pair{i}{\id_\Gamma}}}$ the two compositions
\[
\pret{\wfcxt\Gamma} \xrightarrow{\pret{\pair{\nu}{\sigma}}} \nu^*\pret{\wfcxt[\Delta']{\Gamma'}} 
\xrightarrow{\nu^*(\pret{\pair{[\kappa\mapsto \kappa']}{\basicsub{\tickA}\tickc}})} 
\nu^*\basicsub{\kappa}{\kappa'}^* \pret{\wfcxt[\Delta',\kappa]{\Gamma', \tickA : \kappa}} 
\]
and $ \basicsub{\kappa}{\nu(\kappa')}^* \pret{\pair{\subex\nu\kappa\kappa}{\subex\sigma\tickA\tickA}} \circ \pret{\pair{[\kappa\mapsto \nu(\kappa')]}{\basicsub{\tickA}\tickc}}$ of type
\[
\pret{\wfcxt\Gamma}\to 
 \basicsub{\kappa}{\nu(\kappa')}^*\pret{\wfcxt[\Delta,\kappa]{\Gamma, \tickA : \kappa}} 
 \to 
 (\basicsub{\kappa}{\nu(\kappa')})^*\subex\nu{\kappa}{\kappa}^* \pret{\wfcxt[\Delta',\kappa]{\Gamma', \tickA : \kappa}} 
\]
are equal. 
\end{lemma}
\begin{proof}
 Suppose $\gamma \in \pret{\wfcxt\Gamma}_{\triple\Theta\vartheta f}$. Let $\iota : \Theta \to \Theta, \#$ be the inclusion, then
\begin{align*}
& \nu^*(\pret{\pair{[\kappa\mapsto \kappa']}{\basicsub{\tickA}\tickc}})_{\triple\Theta\vartheta f} 
\circ \pret{\pair{\nu}{\sigma}}_{\triple\Theta\vartheta f}(\gamma) \\
& = \pret{\pair{[\kappa\mapsto \kappa']}{\basicsub{\tickA}\tickc}}_{\triple\Theta\vartheta{f\nu}}
\circ \pret{\pair{\nu}{\sigma}}_{\triple\Theta\vartheta f}(\gamma) \\
 & = \pair{\{\kappa\}}{\inv{\pret{\pair{j}{\id_{\Gamma'}}}}_{\triple{\Theta, \#}{\subex\vartheta\#{n+1}}{\subex{(f\nu)}\kappa\#}}
 (\iota \cdot \pret{\pair{\nu}{\sigma}}_{\triple\Theta\vartheta f}(\gamma))} \\
 & = \pair{\{\kappa\}}{((\subex\nu\kappa\kappa)^*\inv{\pret{\pair{j}{\id_{\Gamma'}}}})_{\triple{\Theta, \#}{\subex\vartheta\#{n+1}}{\subex{f}\kappa\#}}
 (\iota \cdot \pret{\pair{\nu}{\sigma}}_{\triple\Theta\vartheta f}(\gamma))} \\
 & = \pair{\{\kappa\}}{((\subex\nu\kappa\kappa)^*\inv{\pret{\pair{j}{\id_{\Gamma'}}}})_{\triple{\Theta, \#}{\subex\vartheta\#{n+1}}{\subex{f}\kappa\#}}
 (\pret{\pair{\nu}{\sigma}}_{\triple{\Theta,\#}{\subex\vartheta\kappa\#}{\iota f}}(\iota \cdot \gamma))} \\
 & = \pair{\{\kappa\}}{((\subex\nu\kappa\kappa)^*\inv{\pret{\pair{j}{\id_{\Gamma'}}}})_{\triple{\Theta, \#}{\subex\vartheta\#{n+1}}{\subex{f}\kappa\#}}
 i^*((\pret{\pair{\nu}{\sigma}})_{\triple{\Theta,\#}{\subex\vartheta\kappa\#}{\subex{f}\kappa\#}}(\iota \cdot \gamma))} 
\end{align*}
and $\basicsub{\kappa}{\nu(\kappa')}^* \pret{\pair{\subex\nu\kappa\kappa}{\subex\sigma\tickA\tickA}} \circ \pret{\pair{[\kappa\mapsto \nu(\kappa')]}{\basicsub{\tickA}\tickc}}(\gamma)$ is 
\begin{align*} 
  \pair{\{\kappa\}}{\pret{\pair{\subex\nu\kappa\kappa}{\sigma}}_{\triple{\Theta,\#}{\subex\vartheta\kappa\#}{\subex{f}\kappa\#}}
  (\inv{\pret{\pair{i}{\id_\Gamma}}}_{\triple{\Theta,\#}{\subex\vartheta\kappa\#}{\subex{f}\kappa\#}}(\iota\cdot \gamma))}
\end{align*}
By the assumption $(\subex{\nu}\kappa\kappa)^*(\inv{\pret{\pair{j}{\id_{\Gamma'}}}})\circ i^*\pret{\pair{\nu}{\sigma}} = 
\pret{\pair{\subex\nu\kappa\kappa}{\sigma}}\circ \inv{\pret{\pair{i}{\id_\Gamma}}}$ the two are equal.
\end{proof}

\begin{lemma}\label{lem:subst:comp:clock:weak}
 Let $\pair\nu\sigma : \wfcxt\Gamma\to \wfcxt[\Delta']{\Gamma'}$ be a \clott\ substitution, 
 and let $i : \Delta\to \Delta, \kappa$ and $j : \Delta' \to \Delta', \kappa'$ be inclusions into strictly larger sets. Then
\begin{align*}
  (\subex\nu{\kappa'}\kappa)^*\pret{\pair{j}{\id_{\Gamma'}}} \circ \pret{\pair{\subex\nu{\kappa'}\kappa}{\sigma}} 
  =  i^*\pret{\pair{\nu}{\sigma}} \circ \pret{\pair{i}{\id_\Gamma}}
\end{align*}
\end{lemma}

\begin{proof}
 The proof is by induction on $\Gamma'$. In the case of $\Gamma' = \Gamma'',x:A$ we have that $\sigma = \subex\tau xt$ where $\pair\nu \tau:\wfcxt\Gamma\to \wfcxt[\Delta']{\Gamma''}$
\begin{align*}
 (\subex\nu{\kappa'}\kappa)^*\pret{\pair{j}{\id_{\Gamma'}}} \circ \pret{\pair{\subex\nu{\kappa'}\kappa}{\sigma}} 
 & = (\subex\nu{\kappa'}\kappa)^*\cpair{\pret{\pair{j}{p}}}{\q} \circ \cpair{\pret{\pair{\subex\nu{\kappa'}\kappa}{\tau}}}{\pret t} 
\end{align*}
where $\pair{j}{p} : \wfcxt[\Delta',\kappa']{\Gamma'} \to \wfcxt[\Delta']{\Gamma''}$ is the projection. By 
Lemma~\ref{lem:syn:proj:nu} we know that $\pret{\pair{j}{p}}$ equals $\pret{\pair{j}{\id_{\Gamma''}}}\circ \p_A$, so 
\begin{align*}
 & (\subex\nu{\kappa'}\kappa)^*\pret{\pair{j}{\id_{\Gamma'}}} \circ \pret{\pair{\subex\nu{\kappa'}\kappa}{\sigma}}  \\
 & = (\subex\nu{\kappa'}\kappa)^*\cpair{\pret{\pair{j}{\id_{\Gamma''}}}\circ \p_A}{\q} \circ 
 \cpair{\pret{\pair{\subex\nu{\kappa'}\kappa}{\tau}}}{\pret t} \\
 & = \cpair{(\subex\nu{\kappa'}\kappa)^*\pret{\pair{j}{\id_{\Gamma''}}}\circ 
 \pret{\pair{\subex\nu{\kappa'}\kappa}{\tau}}}{\pret t}  \\
 & = \cpair{ i^*\pret{\pair{\nu}{\tau}} \circ \pret{\pair{i}{\id_{\Gamma}}}}{\pret t} 
\end{align*}
which by Lemma~\ref{lem:weakening} equals
\begin{align*}
 \cpair{ i^*\pret{\pair{\nu}{\tau}} \circ \pret{\pair{i}{\id_{\Gamma}}}}{i^*\pret t[\pret{\pair i{\id_{\Gamma}}}]} 
 & = \cpair{ i^*\pret{\pair{\nu}{\tau}}}{i^*\pret t}\circ\pret{\pair i{\id_{\Gamma}}} \\
 & =  i^*\pret{\pair{\nu}{\sigma}} \circ \pret{\pair{i}{\id_\Gamma}}
\end{align*}

In the case of $\Gamma' = \Gamma'', \tickA:\kappa''$ we prove the statement by case analysis on the form of $\sigma$

\begin{itemize}
\item If $\sigma = \subex\tau{\tickB}{\tickA}$ we have that $\Gamma$ must be of the form $\Gamma_0, \tickB : \nu(\kappa''), \Gamma_1$ such that $\pair\nu\tau : 
\wfcxt{\Gamma_0} \to \wfcxt[\Delta']{\Gamma''}$. Then
\begin{align*}
 & (\subex\nu{\kappa'}\kappa)^*\pret{\pair{j}{\id_{\Gamma'}}} \circ \pret{\pair{\subex\nu{\kappa'}\kappa}{\sigma}} \\
 & = (\subex\nu{\kappa'}\kappa)^*(e^{j,\kappa''} \circ \tearlier^{\kappa''}(\pret{\pair{j}{\id_{\Gamma''}}}) \circ 
 e^{\subex\nu{\kappa'}\kappa, \nu(\kappa'')}\circ \tearlier^{\nu(\kappa'')}(\pret{\pair{\subex\nu{\kappa'}\kappa}{\tau}}) 
 \circ \p_{\Gamma_1} \\
  & = (\subex\nu{\kappa'}\kappa)^*e^{j,\kappa''} \circ  (\subex\nu{\kappa'}\kappa)^*\tearlier^{\kappa''}(\pret{\pair{j}{\id_{\Gamma''}}}) \circ 
 e^{\subex\nu{\kappa'}\kappa, \nu(\kappa'')}\circ \tearlier^{\nu(\kappa'')}(\pret{\pair{\subex\nu{\kappa'}\kappa}{\tau}}) 
 \circ \p_{\Gamma_1} \\
   & = (\subex\nu{\kappa'}\kappa)^*e^{j,\kappa''} \circ e^{\subex\nu{\kappa'}\kappa, \nu(\kappa'')}\circ
    \tearlier^{\nu(\kappa'')}(\subex\nu{\kappa'}\kappa)^*(\pret{\pair{j}{\id_{\Gamma''}}}) \circ 
  \tearlier^{\nu(\kappa'')}(\pret{\pair{\subex\nu{\kappa'}\kappa}{\tau}}) \circ \p_{\Gamma_1} \\
   & = e^{\subex\nu{\kappa'}\kappa j, \kappa''}\circ
    \tearlier^{\nu(\kappa'')}(\subex\nu{\kappa'}\kappa)^*(\pret{\pair{j}{\id_{\Gamma''}}}) \circ 
  \tearlier^{\nu(\kappa'')}(\pret{\pair{\subex\nu{\kappa'}\kappa}{\tau}}) \circ \p_{\Gamma_1}
\end{align*}
the last of these equalities is by Lemma~\ref{lem:e:nat}. Since $\subex\nu{\kappa'}\kappa j = i\nu$ we get 
\begin{align*}
 &(\subex\nu{\kappa'}\kappa)^*\pret{\pair{j}{\id_{\Gamma'}}} \circ \pret{\pair{\subex\nu{\kappa'}\kappa}{\sigma}} \\
   & = e^{i\nu, \kappa''}\circ
    \tearlier^{\nu(\kappa'')}(\subex\nu{\kappa'}\kappa)^*(\pret{\pair{j}{\id_{\Gamma''}}} \circ 
  \pret{\pair{\subex\nu{\kappa'}\kappa}{\tau}}) \circ \p_{\Gamma_1} \\
   & = e^{i\nu, \kappa''}\circ
    \tearlier^{\nu(\kappa'')}(i^*\pret{\pair{\nu}{\tau}} \circ \pret{\pair{i}{\id_{\Gamma_0}}}) 
    \circ \p_{\Gamma_1}
\end{align*}
By another application of Lemma~\ref{lem:e:nat}, the above equals
\begin{align*}
  & i^*(e^{\nu, \kappa''}) \circ e^{i, \nu(\kappa'')}\circ
    \tearlier^{\nu(\kappa'')}i^*\pret{\pair{\nu}{\tau}} \circ \tearlier^{\nu(\kappa'')}\pret{\pair{i}{\id_{\Gamma_0}}}
    \circ \p_{\Gamma_1} \\
 &=   i^*(e^{\nu, \kappa''}) \circ i^*\tearlier^{\nu(\kappa'')}\pret{\pair{\nu}{\tau}} \circ e^{i, \nu(\kappa'')}\circ
     \tearlier^{\nu(\kappa'')}\pret{\pair{i}{\id_{\Gamma_0}}}
    \circ \p_{\Gamma_1} \\
 &=   i^*(e^{\nu, \kappa''}) \circ i^*\tearlier^{\nu(\kappa'')}\pret{\pair{\nu}{\tau}} \circ 
 \pret{\pair{i}{\id_{\Gamma_0, \beta : \nu(\kappa'')}}} \circ \p_{\Gamma_1} 
\end{align*}
By Lemma~\ref{lem:syn:proj:nu} the latter is equal to 
\begin{align*}
 &=   i^*(e^{\nu, \kappa''}) \circ i^*\tearlier^{\nu(\kappa'')}\pret{\pair{\nu}{\tau}} \circ i^*\p_{\Gamma_1} \circ
 \pret{\pair{i}{\id_{\Gamma}}}  \\
& =   i^*\pret{\pair{\nu}{\sigma}} \circ  \pret{\pair{i}{\id_{\Gamma}}}  
\end{align*}
Finishing the proof of the case.

\item If $\sigma = \subex\tau{\tickB}{\tickc}$, the context $\Gamma'$ must be of the form $\Gamma'', \tickA : \kappa''$, 
the clock context $\Delta'$ of the form $\Delta'', \kappa''$, 
the clock substitution $\nu$ of the form $\subex\mu{\kappa''}{\mu(\kappa''')}$, for some $\kappa'''\in \Delta''$
such that $\pair\mu\tau : \wfcxt{\Gamma} \to \wfcxt[\Delta'']{\Gamma''}$. Then
\begin{align*}
& (\subex\nu{\kappa'}\kappa)^*\pret{\pair{j}{\id_{\Gamma'}}} \circ \pret{\pair{\subex\nu{\kappa'}\kappa}{\sigma}}  \\
 & = (\subex\nu{\kappa'}\kappa)^*\pret{\pair{j}{\id_{\Gamma'}}} \circ 
 (\subex\mu{\kappa'}{\kappa})^*\pret{\pair{\basicsub{\kappa''}{\kappa'''}}{\basicsub\tickA\tickc}} 
 \circ  \pret{\pair{\subex\mu{\kappa'}\kappa}{\tau}}\\
 & = (\subex\mu{\kappa'}\kappa)^*(\basicsub{\kappa''}{\kappa'''}^*\pret{\pair{j}{\id_{\Gamma'}}} \circ 
\pret{\pair{\basicsub{\kappa''}{\kappa'''}}{\basicsub\tickA\tickc}} )
 \circ  \pret{\pair{\subex\mu{\kappa'}\kappa}{\tau}}
\end{align*}
Let $k : \Delta'' \to \Delta'', \kappa''$ be the inclusion. Then $\subex k{\kappa''}{\kappa''} = j$, and 
\begin{align*}
& \basicsub{\kappa''}{\kappa'''}^*\pret{\pair{j}{\id_{\Gamma'}}} \circ 
\pret{\pair{\basicsub{\kappa''}{\kappa'''}}{\basicsub\tickA\tickc}} \\
& = 
 \basicsub{\kappa''}{k(\kappa''')}^*\pret{\pair{\subex k{\kappa''}{\kappa''}}{\subex{\id_{\Gamma''}}\tickA\tickA}} \circ 
\pret{\pair{\basicsub{\kappa''}{k(\kappa''')}}{\basicsub\tickA\tickc}} 
\end{align*}
By the induction hypothesis
\begin{align*}
  (\subex k{\kappa'}{\kappa'})^*\pret{\pair{j'}{\id_{\Gamma''}}} \circ \pret{\pair{\subex k{\kappa'}{\kappa'}}{\id_{\Gamma''}}} 
  =  j^*\pret{\pair{k}{\id_{\Gamma''}}} \circ \pret{\pair{j}{\id_{\Gamma''}}}
\end{align*}
Where $j'$ is the inclusion $\Delta'' \to \Delta'',\kappa'$. We have then that
\[(\subex k{\kappa'}{\kappa'})^*\inv{\pret{\pair{j'}{\id_{\Gamma''}}}}\circ j^*\pret{\pair{k}{\id_{\Gamma''}}} = 
 \pret{\pair{\subex k{\kappa'}{\kappa'}}{\id_{\Gamma''}}} \circ \inv{\pret{\pair{j}{\id_{\Gamma''}}}}\]
By Lemma \ref{lem:weak:sub:clock:extension} we have then that
\begin{align*}\basicsub{\kappa''}{k(\kappa''')}^*\pret{\pair{\subex k{\kappa''}{\kappa''}}{\subex{\id_{\Gamma''}}\tickA\tickA}} \circ &
\pret{\pair{\basicsub{\kappa''}{k(\kappa''')}}{\basicsub\tickA\tickc}}\\& = k^*\pret{\pair{\basicsub{\kappa''}{\kappa''')}}{\basicsub\tickA\tickc}} \circ
\pret{\pair{k}{\id_{\Gamma''}}}
\end{align*}

So 
\begin{align*}
& (\subex\nu{\kappa'}\kappa)^*\pret{\pair{j}{\id_{\Gamma'}}} \circ \pret{\pair{\subex\nu{\kappa'}\kappa}{\sigma}}  \\
& = (\subex\mu{\kappa'}\kappa)^*(k^*\pret{\pair{\basicsub{\kappa''}{\kappa''')}}{\basicsub\tickA\tickc}} \circ
\pret{\pair{k}{\id_{\Gamma''}}})
 \circ  \pret{\pair{\subex\mu{\kappa'}\kappa}{\tau}} \\
& = (\subex\mu{\kappa'}\kappa)^*k^*\pret{\pair{\basicsub{\kappa''}{\kappa''')}}{\basicsub\tickA\tickc}} \circ
(\subex\mu{\kappa'}\kappa)^*\pret{\pair{k}{\id_{\Gamma''}}})
 \circ  \pret{\pair{\subex\mu{\kappa'}\kappa}{\tau}}\\
& = i^*\mu^*\pret{\pair{\basicsub{\kappa''}{\kappa''')}}{\basicsub\tickA\tickc}} \circ
(\subex\mu{\kappa'}\kappa)^*\pret{\pair{k}{\id_{\Gamma''}}})
 \circ  \pret{\pair{\subex\mu{\kappa'}\kappa}{\tau}}
\end{align*}
By the induction hypothesis this equals
\begin{align*}
&  i^*\mu^*\pret{\pair{\basicsub{\kappa''}{\kappa''')}}{\basicsub\tickA\tickc}} \circ
i^*\pret{\pair{\mu}{\tau}}\circ \pret{\pair i{\id_{\Gamma}}} \\
& =  i^*(\mu^*\pret{\pair{\basicsub{\kappa''}{\kappa''')}}{\basicsub\tickA\tickc}} \circ
\pret{\pair{\mu}{\tau}})\circ \pret{\pair i{\id_{\Gamma}}} \\
& =  i^*(\pret{\pair{\nu}{\subex\tau{\tickA}{\tickc}}})
\circ \pret{\pair i{\id_{\Gamma}}} \\
& =  i^*(\pret{\pair{\nu}{\sigma}})
\circ \pret{\pair i{\id_{\Gamma}}} \qedhere
\end{align*}
\end{itemize}
\end{proof}

\subsection{Variable introduction}

We now turn to the general substitution lemma. This subsection proves the case of variable introduction. We must prove that if 
$\pair\nu\sigma : \wfcxt\Gamma \to \wfcxt[\Delta']{\Gamma'}$ and $\Gamma'$ is of the form
$\Gamma'_0, x : A, \Gamma'_1$, then 
\[
\nu^*(\pret{\hastype[\Delta']{\Gamma'}xA})[\pret{\pair\nu\sigma}] = \pret{\hastype{\Gamma}{\sigma(x)}{A\pair\nu\sigma}}
\]
But we have
\begin{align*}
\nu^*(\pret{\hastype[\Delta']{\Gamma'}xA})[\pret{\pair\nu\sigma}] &=(\nu^*(\q_{\pret A}[\p_{\pret{\Gamma_1}}])[\pret{\pair \nu \sigma}]\\
&=\q_{\nu^*\pret{A}} [\nu^* \p_{\pret{\Gamma_1}}][\pret{\pair \nu \sigma}]\\
&= \q_{\nu^*\pret{A}} [\nu^* \p_{\pret{\Gamma_1}} \circ \pret{\pair \nu \sigma}]\\
\end{align*}
We know that $\sigma$ has the form $\tau[x\mapsto \sigma(x)]\tau'$ where $\pair\nu \tau:\wfcxt\Gamma \to \wfcxt[\Delta']{\Gamma'_0}$ and $\sigma(x)$ is of type $A\pair\nu\tau = A\pair\nu\sigma$. We have then that $\q_{\nu^*\pret{A}} [\nu^* \p_{\pret{\Gamma_1}} \circ \pret{\pair \nu \sigma}] =\q_{\nu^*\pret{A}} [\nu^* \p_{\pret{\Gamma_1}} \circ \pret{\pair \nu {\tau[x\mapsto \sigma(x)]\tau'}}]$. By Lemma~\ref{lem:sub:comp:proj}
\begin{align*}
\q_{\nu^*\pret{A}} [\nu^* \p_{\pret{\Gamma_1}} \circ \pret{\pair \nu {\tau[x\mapsto \sigma(x)]\tau'}}] &= \q_{\nu^*\pret{A}} [\pret{\pair \nu {\tau[x\mapsto \sigma(x)]}}]\\
&= \q_{\nu^*\pret{A}} [\cpair {\pret{\pair \nu \tau}} {\pret{\sigma(x)}} ]\\
&=\pret{\sigma(x)}
\end{align*}

\subsection{Tick application and abstraction}


The case of tick abstraction is almost precisely the same as the case of $\latbind\tickA\kappa A$ presented in Section~\ref{sec:substitutions}:
\begin{align*}
 \nu^*\pret{\istypeshort[\Delta']{\Gamma'}{\tabs\tickA\kappa t}{}}[\pret\sigma] & = 
 \nu^*\left(\left(\tlater^{\kappa}\pret{\istypeshort[\Delta']{\Gamma', \tickA : \kappa}{t}}\right)[\eta^\kappa_{\pret{\Gamma'}}]\right)[\pret\sigma] \\
& = \left(\tlater^{\nu(\kappa)}\nu^*\pret{\istypeshort[\Delta']{\Gamma', \tickA : \kappa}{t}}\right)[\nu^*\eta^\kappa_{\pret{\Gamma'}}][\pret\sigma] \\
& = \left(\tlater^{\nu(\kappa)}\nu^*\pret{\istypeshort[\Delta']{\Gamma', \tickA : \kappa}{t}}\right)[\tlater^{\nu(\kappa)}e^\kappa \circ \eta^{\nu(\kappa)}_{\nu^*\pret{\Gamma'}} \circ \pret\sigma] \\
& = \left(\tlater^{\nu(\kappa)}\nu^*\pret{\istypeshort[\Delta']{\Gamma', \tickA : \kappa}{t}}\right)[\tlater^{\nu(\kappa)}(e^\kappa \circ \tearlier^{\nu(\kappa)}\pret\sigma)\circ \eta^{\nu(\kappa)}_{\nu^*\pret{\Gamma'}}] \\
& = \left(\tlater^{\nu(\kappa)}\left(\nu^*\pret{\istypeshort[\Delta']{\Gamma', \tickA : \kappa}{t}}\right)[e^\kappa \circ \tearlier^{\nu(\kappa)}\pret\sigma]\right) [\eta^{\nu(\kappa)}_{\nu^*\pret{\Gamma'}}] \\
& = \left(\tlater^{\nu(\kappa)}\pret{\istypeshort{\Gamma, \tickB : \nu(\kappa)}{t\pair{\nu}{\subex\sigma\tickA\tickB}}}\right) [\eta^{\nu(\kappa)}_{\nu^*\pret{\Gamma'}}] \\
& = \pret{\istypeshort{\Gamma}{\tabs \tickB{\nu(\kappa)}{(t\pair{\nu}{\subex\sigma\tickA\tickB})}}} \\
& = \pret{\istypeshort{\Gamma}{(\tabs \tickA{\kappa}t)\pair{\nu}\sigma}} 
\end{align*}

For the case of application of a term to a tick variable, we are given 
$\hastype[\Delta']{\Gamma_0'}t{\latbind \tickA\kappa A}$ and a substitution
$\pair\nu\sigma :  \wfcxt{\Gamma} \to \wfcxt[\Delta']{\Gamma_0',\tickA : \kappa, \Gamma_1'}$ and must show that
\[
\nu^*(\pret{\istypeshort[\Delta']{\Gamma_0',\tickA : \kappa, \Gamma_1'}{\tapp t}})[\pret{\pair\nu\sigma}]
= \pret{\istypeshort[\Delta]{\Gamma}{(\tapp t)\pair\nu\sigma}}
\]
Note first that since $\sigma$ must be of the form $\tau\tau'$ where 
$\pair\nu\tau :  \wfcxt{\Gamma} \to \wfcxt[\Delta']{\Gamma_0',\tickA : \kappa}$ and 
\begin{align*}
\nu^*(\pret{\istypeshort[\Delta']{\Gamma_0',\tickA : \kappa, \Gamma_1'}{\tapp t}})[\pret{\pair\nu\sigma}]
& = \nu^*(\pret{\istypeshort[\Delta']{\Gamma_0',\tickA : \kappa}{\tapp t}}[\p_{\Gamma'_1}])[\pret{\pair\nu\sigma}] \\
& = \nu^*(\pret{\istypeshort[\Delta']{\Gamma_0',\tickA : \kappa}{\tapp t}})[\nu^*(\p_{\Gamma'_1})\pret{\pair\nu\sigma}] \\
& = \nu^*(\pret{\istypeshort[\Delta']{\Gamma_0',\tickA : \kappa}{\tapp t}})[\pret{\pair\nu{\tau}}] 
\end{align*}
and $(\tapp t)\pair\nu\sigma = (\tapp t)\pair\nu\tau$ it suffices to consider the case where $\Gamma'_1$ is empty. We will
thus assume this from now on. 

By the type of the substitution we know that $\pair\nu\sigma$ must either be of the form $\pair\nu{\subex\tau\tickA\tickB}$ or 
of the form $\pair\nu{\subex\tau\tickA\tickc}$. We first consider the former case.

If $\tau = \subex\tau\tickA\tickB$, the context $\Gamma$ must be of the form $\Gamma_0, \tickB : \nu(\kappa), \Gamma_1$. Then
\begin{align*}
 \nu^*(\pret{\istypeshort[\Delta']{\Gamma_0',\tickA : \kappa}{\tapp t}})[\pret{\pair\nu\sigma}] 
 & = \nu^*(\transp{\pret{\istypeshort[\Delta']{\Gamma_0'}{t}}})[\pret{\pair\nu\sigma}] \\
 & = \nu^*(\transp{\pret{\istypeshort[\Delta']{\Gamma_0'}{t}}})[\pret{\pair\nu{\subex\tau\tickA\tickB}}] \\
 & = \nu^*(\transp{\pret{\istypeshort[\Delta']{\Gamma_0'}{t}}})[\pret{\pair\nu{\subex\tau\tickA\tickB}}] [\p_{\Gamma_1}]
\end{align*}
where in the last equality $\pair\nu{\subex\tau\tickA\tickB}$ is considered an equality from 
$\wfcxt{\Gamma_0, \beta : \nu(\kappa)}$. Since also 
\[
  \pret{\istypeshort[\Delta]{\Gamma}{(\tapp t)\pair\nu\sigma}} = 
  \pret{\istypeshort[\Delta]{\Gamma_0, \tickA : \kappa}{(\tapp t)\pair\nu{\subex\tau\tickA\tickB}}}[\p_{\Gamma_1}]
\]
it suffices to prove 
\[
\nu^*(\transp{\pret{\istypeshort[\Delta']{\Gamma_0'}{t}}})[\pret{\pair\nu{\subex\tau\tickA\tickB}}] 
= \pret{\istypeshort[\Delta]{\Gamma_0, \tickA : \kappa}{(\tapp t)\pair\nu{\subex\tau\tickA\tickB}}} 
\]
To prove this, it suffices to show that their correspondents under the 
bijective correspondence of Lemma~\ref{lem:bijectivecorresp} agree. 
The correspondent to 
$\pret{\istypeshort[\Delta]{\Gamma_0, \tickA : \kappa}{(\tapp t)\pair\nu{\subex\tau\tickA\tickB}}}$ is simply
$\pret{\istypeshort[\Delta]{\Gamma_0}{t\pair\nu{\tau}}}$. 
The correspondent to 
$\nu^*(\transp{\pret{\istypeshort[\Delta']{\Gamma_0'}{t}}})[\pret{\pair\nu{\subex\tau\tickA\tickB}}]$ is 
\begin{align*}
 \tlater^{\nu(\kappa)}(\nu^*(\transp{\pret{\istypeshort[\Delta']{\Gamma_0'}{t}}})[\pret{\pair\nu{\subex\tau\tickA\tickB}}])
 [\eta]
& = \,  \tlater^{\nu(\kappa)}(\nu^*(\transp{\pret{\istypeshort[\Delta']{\Gamma_0'}{t}}}))[\tlater^{\nu(\kappa)}
(\pret{\pair\nu{\subex\tau\tickA\tickB}}) \circ \eta]  \\
& = \, \tlater^{\nu(\kappa)}(\nu^*(\transp{\pret{\istypeshort[\Delta']{\Gamma_0'}{t}}}))[\tlater^{\nu(\kappa)}
(e^{\kappa, \nu} \circ \tearlier^\kappa(\pret{\pair\nu{\tau}})) \circ \eta] \\
& = \, \tlater^{\nu(\kappa)}(\nu^*(\transp{\pret{\istypeshort[\Delta']{\Gamma_0'}{t}}}))[\tlater^{\nu(\kappa)}
(e^{\kappa, \nu}) \circ \eta \circ \pret{\pair\nu{\tau}}] \\
& = \, \tlater^{\nu(\kappa)}(\nu^*(\transp{\pret{\istypeshort[\Delta']{\Gamma_0'}{t}}}))[\nu^*(\eta) \circ \pret{\pair\nu{\tau}}] \\
& =  \nu^*(\tlater^\kappa(\transp{\pret{\istypeshort[\Delta']{\Gamma_0'}{t}}}))[\nu^*(\eta) \circ \pret{\pair\nu{\tau}}] \\
& =  \nu^*(\tlater^\kappa(\transp{\pret{\istypeshort[\Delta']{\Gamma_0'}{t}}})[\eta])[\pret{\pair\nu{\tau}}] \\
& =  \nu^*(\pret{\istypeshort[\Delta']{\Gamma_0'}{t}})[\pret{\pair\nu{\tau}}]  \\
& =  \pret{\istypeshort[\Delta]{\Gamma_0}{t\pair\nu{\tau}}}
\end{align*}
proving the case. 

Finally, we must consider the case of $\pair\nu\sigma = \pair\nu{\subex\tau\tickA\tickc}$. In this case there is a $\kappa$ and
$\Delta''$ such
that $\Delta' = \Delta'', \kappa$ and $\nu = \subex\mu\kappa{\mu(\kappa')}$ for some $\mu : \Delta'' \to \Delta$ and 
$\kappa' \in \Delta''$. Now,
\begin{align*}
 & \nu^*(\pret{\istypeshort[\Delta']{\Gamma_0',\tickA : \kappa}{\tapp t}})[\pret{\pair\nu\sigma}] \\
 & =  \nu^*\pret{\istypeshort[\Delta']{\Gamma_0',\tickA : \kappa}{\tapp t}}
 [\mu^*(\pret{\pair{\basicsub{\kappa}{\kappa'}}{\basicsub\tickA\tickc}})\circ\pret{\pair\mu\tau}]  
 \end{align*}
 By Lemma~\ref{lem:sub:clock:extension} this equals
\begin{align*}
 & \nu^*(\pret{\istypeshort[\Delta']{\Gamma_0',\tickA : \kappa}{\tapp t}})
 [\basicsub{\kappa}{\mu(\kappa')}^*\pret{\pair{\subex\mu{\kappa}\kappa}{\subex\tau\tickA\tickA}}
  \circ \pret{\pair{\basicsub{\kappa}{\mu(\kappa')}}{\basicsub\tickA\tickc}} ]  \\
   & = \basicsub{\kappa}{\mu(\kappa')}^*((\subex\mu{\kappa}{\kappa})^*
   \pret{\istypeshort[\Delta']{\Gamma_0',\tickA : \kappa}{\tapp t}})
 [\basicsub{\kappa}{\mu(\kappa')}^*\pret{\pair{\subex\mu{\kappa}\kappa}{\subex\tau\tickA\tickA}}
  \circ \pret{\pair{\basicsub{\kappa}{\mu(\kappa')}}{\basicsub\tickA\tickc}} ]  \\
   & = \basicsub{\kappa}{\mu(\kappa')}^*((\subex\mu{\kappa}{\kappa})^*
   \pret{\istypeshort[\Delta']{\Gamma_0',\tickA : \kappa}{\tapp t}}
   [\pret{\pair{\subex\mu{\kappa}\kappa}{\subex\tau\tickA\tickA}}])
   [\pret{\pair{\basicsub{\kappa}{\mu(\kappa')}}{\basicsub\tickA\tickc}} ] 
\end{align*}
By the case we just proved above, this equals
\begin{align*}
  & \basicsub{\kappa}{\mu(\kappa')}^*(
   \pret{\istypeshort[\Delta,\kappa]{\Gamma,\tickA : \kappa}{\tapp t\pair{\subex\mu{\kappa}\kappa}{\subex\tau\tickA\tickA}}})
   [\pret{\pair{\basicsub{\kappa}{\mu(\kappa')}}{\basicsub\tickA\tickc}} ] \\
  & = \basicsub{\kappa}{\mu(\kappa')}^*(
   \pret{\istypeshort[\Delta,\kappa]{\Gamma,\tickA : \kappa}{\tapp {(t\pair{\subex\mu{\kappa}\kappa}{\tau})}}})
   [\pret{\pair{\basicsub{\kappa}{\mu(\kappa')}}{\basicsub\tickA\tickc}} ] 
\end{align*}
 which by definition equals
\begin{align*}
   \pret{\istypeshort[\Delta]{\Gamma}{\tappc {(t\pair{\subex\mu{\kappa}{\mu(\kappa')}}{\tau})}}}
   = \pret{\istypeshort[\Delta]{\Gamma}{\tappc {(t\pair{\nu}{\tau})}}}
\end{align*}
as required. 

\section{Interpretation of $\forall\kappa$}
\label{app:forall}

\begin{proofof}{Lemma~\ref{lem:(i,id):invertible}}
  The proof is by induction on $\Gamma$. In the case where $\Gamma=\Gamma',x:A$ we have 
  \[\pret{\pair{i}{\id_{\Gamma', x : A}}} = \cpair{\pret{\pair{i}{\id_{\Gamma'}}} \circ \p_A}{\q_A}\]
  Given a pair $\pair{\gamma}{a} \in \pret{\wfcxt[\Delta,\kappa]{\Gamma'}}_{\triple {\Theta,\#}{\vartheta[\#\mapsto n]}{f[\kappa\mapsto \#]}}$. We have that
  \[\cpair{\pret{\pair{i}{\id_{\Gamma'}}} \circ \p_A}{\q_A} \pair{\gamma}{a} = \pair{\pret{\pair{i}{\id_{\Gamma'}}} (\gamma)} {a}\in \pret{\wfcxt[\Delta]{\Gamma',x:A}}_{\triple{\Theta,\#}{\vartheta[\#\mapsto n]}{\iota f}}\]
  By the induction hypothesis $\pret{\pair{i}{\id_{\Gamma'}}}$ has a two sided inverse $\inv{\pret{\pair{i}{\id_{\Gamma'}}}}$. We get have then that $\cpair{\pret{\pair{i}{\id_{\Gamma'}}} \circ \p_A}{\q_A}$ has the two sided inverse mapping $\pair{\gamma'}{a'} \in \pret{\Gamma',x:A}{\triple {\Theta,\#}{\vartheta[\#\mapsto n]}{\iota f}}$ to  $\pair{\inv{\pret{\pair{i}{\id_{\Gamma'}}}} (\gamma')}{a'}$.
  

In the case where $\Gamma = \Gamma', \tickA:\kappa'$ (where $\kappa' \in \Delta$), the map $\pret{\pair{i}{\id_{\Gamma, \tickA : \kappa'}}}$ is defined as the
 composition 
 \[
 \tearlier^{\kappa'}\pret{\wfcxt[\Delta, \kappa]{\Gamma}} \xrightarrow{\tearlier^{\kappa'}\pret{\pair{i}{\id_{\Gamma}}}}
 \tearlier^{\kappa'}i^*\pret{\wfcxt[\Delta]{\Gamma}} \xrightarrow{e^{\kappa', i}_{\pret\Gamma}}
 i^*\tearlier^{\kappa'}\pret{\wfcxt[\Delta]{\Gamma}}
 \]
 The component at $\triple{\Theta, \#}{\vartheta\basicsub\# n}{f\basicsub{\kappa}{\#}}$ of the first map has domain 
\begin{align*}
  & \displaystyle\coprod_{\kappa'\in X \subset (f\basicsub{\kappa}{\#})^{-1}(f\basicsub{\kappa}{\#}(\kappa'))} \pret{\wfcxt[\Delta, \kappa]{\Gamma}}
 \triple{\Theta,\#,\#'}{\vartheta,\#\mapsto n, \#'\mapsto m+1}{f\basicsub{\kappa}{\#}[X\mapsto \#']} \\
  = & \displaystyle\coprod_{\kappa'\in X \subset f^{-1}(f(\kappa'))} \pret{\wfcxt[\Delta, \kappa]{\Gamma}}
 \triple{\Theta,\#',\#}{\vartheta,\#'\mapsto m+1,\#\mapsto n}{f[X\mapsto \#']\basicsub{\kappa}{\#}}
\end{align*}
 where $m = \vartheta(f(\kappa'))$. Note that this calculation uses crucially that $\kappa$ is never in 
 $(f\basicsub{\kappa}{\#})^{-1}(f\basicsub{\kappa}{\#}(\kappa'))$. The codomain can be be computed likewise as
 \[
 \displaystyle\coprod_{\kappa'\in X \subset f^{-1}(f(\kappa'))} i^*\pret{\wfcxt[\Delta]{\Gamma}}
 \triple{\Theta,\#',\#}{\vartheta,\#'\mapsto m+1,\#\mapsto n}{f[X\mapsto \#']\basicsub{\kappa}{\#}}
 \]
 The component of $\tearlier^{\kappa'}\pret{\pair{i}{\id_{\Gamma}}}$ at 
 ${\triple{\Theta, \#}{\vartheta\basicsub\# n}{f\basicsub{\kappa}{\#}}}$ is 
 \[
 \tearlier^{\kappa'}\pret{\pair{i}{\id_{\Gamma}}}_{\triple{\Theta, \#}{\vartheta\basicsub\# n}{f\basicsub{\kappa}{\#}}}
 = \displaystyle\coprod_{\kappa'\in X \subset f^{-1}(f(\kappa'))}
\pret{\pair{i}{\id_{\Gamma}}}_{\triple{\Theta,\#',\#}{\vartheta,\#'\mapsto m+1,\#\mapsto n}{f[X\mapsto \#']\basicsub{\kappa}{\#}}}
 \]
 By the induction hypothesis, each of the maps of the coproduct is an isomorphism, and thus also  
 \[ \tearlier^{\kappa'}\pret{\pair{i}{\id_{\Gamma}}}_{\triple{\Theta, \#}{\vartheta\basicsub\# n}{f\basicsub{\kappa}{\#}}}\]
 is an isomorphism. 
 
 The component of $e^{\kappa', i}_{\pret\Gamma}$ at $\triple{\Theta, \#}{\vartheta\basicsub\# n}{f\basicsub{\kappa}{\#}}$ 
 has domain 
  \[
 \displaystyle\coprod_{\kappa'\in X \subset f^{-1}(f(\kappa'))} \pret{\wfcxt[\Delta]{\Gamma}}
 \triple{\Theta,\#',\#}{\vartheta,\#'\mapsto m+1,\#\mapsto n}{f[X\mapsto \#']}
 \]
 and the codomain is the same. It maps a pair $(X,\gamma)$ to $\pair{X\cap \Delta}\gamma$ and since $X\cap\Delta = X$,
 this is the identity and thus an isomorphism.

 The last statement of the lemma follows Lemma~\ref{lem:subst:comp:i:kappa:to:kappa'}.
\end{proofof}

\begin{proofof}{Lemma~\ref{lem:beta:eta:clock:quant}}
 Suppose $\hastype[\Delta, \kappa]{\Gamma}tA$ and $\wfcxt\Gamma$, for the $\beta$ rule, we must show that
 if $\kappa'\in \Delta$ then 
 \[
 \pret{\hastype{\Gamma}{(\Lambda \kappa. t) [\kappa']}{A\subst{\kappa}{\kappa'}}} 
 = \pret{\hastype{\Gamma}{t\subst{\kappa}{\kappa'}}{A\subst{\kappa}{\kappa'}}}
 \]
 Let $\triple\Theta\vartheta f$ be an object of $\catT$ and $\gamma \in \pret{\wfcxt\Gamma}_{\triple\Theta\vartheta f}$.
 Let $n = \vartheta(f(\kappa'))$. Then 
\begin{align*}
 \pret{\hastype{\Gamma}{(\Lambda \kappa. t) [\kappa']}{A\subst{\kappa}{\kappa'}}}_{\triple\Theta\vartheta f}(\gamma)
 & = \basicsub{\#}{f(\kappa')}\cdot(\pret{\hastype{\Gamma}{\Lambda \kappa. t }{\forall\kappa . A}}_{\triple\Theta\vartheta f}(\gamma))_n \\
 & = \basicsub{\#}{f(\kappa')}\cdot(\pret{\hastype[\Delta, \kappa]{\Gamma}{t}{A}}_{\triple{\Theta, \#}{\subex\vartheta{\#}{n}} {\subex f\kappa\#}}(\inv{\pret{\pair i{\id}}}(\iota\cdot \gamma))) \\ 
 & = \pret{\hastype[\Delta, \kappa]{\Gamma}{t}{A}}_{\triple{\Theta}{\vartheta} {\subex f\kappa{f(\kappa')}}}(\basicsub{\#}{f(\kappa')}\cdot{\pret{\pair {\basicsub{\kappa}{\kappa'}}{\id}}}(\iota\cdot \gamma)) \\ 
 & = \pret{\hastype[\Delta, \kappa]{\Gamma}{t}{A}}_{\triple{\Theta}{\vartheta} {\subex f\kappa{f(\kappa')}}}({\pret{\pair {\basicsub{\kappa}{\kappa'}}{\id}}}(\basicsub{\#}{f(\kappa')}\cdot\iota\cdot \gamma)) \\ 
 & = \pret{\hastype[\Delta, \kappa]{\Gamma}{t}{A}}_{\triple{\Theta}{\vartheta} {\subex f\kappa{f(\kappa')}}}({\pret{\pair {\basicsub{\kappa}{\kappa'}}{\id}}}(\gamma)) \\ 
 & = \basicsub{\kappa}{\kappa'}^*\pret{\hastype[\Delta, \kappa]{\Gamma}{t}{A}}_{\triple{\Theta}{\vartheta} {f}}({\pret{\pair {\basicsub{\kappa}{\kappa'}}{\id}}}(\gamma)) \\ 
 & = \pret{\hastype{\Gamma}{t\subst{\kappa}{\kappa'}}{A\subst{\kappa}{\kappa'}}}_{\triple{\Theta}{\vartheta} {f}}(\gamma)
\end{align*}

 For the $\eta$ rule, suppose $\hastype[\Delta]{\Gamma}t{\forall \kappa . A}$, we must show that for every
 object $\triple\Theta\vartheta f$ of $\catT$, element $\gamma \in \pret{\wfcxt\Gamma}_{\triple\Theta\vartheta f}$
 and natural number $n$
\[
  (\pret{\hastype[\Delta]{\Gamma}{\Lambda \kappa . (t [\kappa])}{\forall \kappa . A}}_{\triple\Theta\vartheta f}(\gamma))_n
  = (\pret{\hastype[\Delta]{\Gamma}t{\forall \kappa . A}}_{\triple\Theta\vartheta f}(\gamma))_n
\]
  We compute  
\begin{align*}
 (\pret{\hastype[\Delta]{\Gamma}{\Lambda \kappa . (t [\kappa])}{\forall \kappa . A}}_{\triple\Theta\vartheta f}(\gamma))_n
 & = \pret{\hastype[\Delta,\kappa]{\Gamma}{t [\kappa]}{A}}_{\triple{\Theta,\#}{\subex{\vartheta}{\#}n}{\subex{f}{\kappa}{\#}}}(\inv{\pret{\pair i{\id}}}(\iota \cdot \gamma)) \\
  & = \basicsub{\#'}{\#} \cdot (\pret{\hastype[\Delta,\kappa]{\Gamma}{t}{\forall\kappa . A}}_{\triple{\Theta,\#}{\subex{\vartheta}{\#}n}{\subex{f}{\kappa}{\#}}}(\inv{\pret{\pair i{\id}}}(\iota \cdot \gamma)))_n
\end{align*}
where $\#'$ is the chosen clock name fresh for $\Theta, \#$. 

Now, by the substitution lemma 
\begin{align*}
& \pret{\hastype[\Delta,\kappa]{\Gamma}{t}{\forall\kappa . A}}_{\triple{\Theta,\#}{\subex{\vartheta}{\#}n}{\subex{f}{\kappa}{\#}}}(\inv{\pret{\pair i{\id}}}(\iota \cdot \gamma)) \\
& = i^*(\pret{\hastype[\Delta]{\Gamma}{t}{\forall\kappa . A}})_{\triple{\Theta,\#}{\subex{\vartheta}{\#}n}{\subex{f}{\kappa}{\#}}}(\pret{\pair i{\id}} (\inv{\pret{\pair i{\id}}}(\iota \cdot \gamma))) \\
& = \pret{\hastype[\Delta]{\Gamma}{t}{\forall\kappa . A}}_{\triple{\Theta,\#}{\subex{\vartheta}{\#}n}{f}}(\iota \cdot \gamma) \\
& = \iota \cdot (\pret{\hastype[\Delta]{\Gamma}{t}{\forall\kappa . A}}_{\triple{\Theta}{\vartheta}{f}}(\gamma))
\end{align*}
So 
\begin{align}
 (\pret{\hastype[\Delta]{\Gamma}{\Lambda \kappa . (t [\kappa])}{\forall \kappa . A}}_{\triple\Theta\vartheta f}(\gamma))_n
  & = \basicsub{\#'}{\#} \cdot (\iota \cdot (\pret{\hastype[\Delta]{\Gamma}{t}{\forall\kappa . A}}_{\triple{\Theta}{\vartheta}{f}}(\gamma)))_n \label{eq:eta:clocks:subresult}
\end{align}
Here 
\[
\basicsub{\#'}{\#}  : \triple{\Theta,\#, \#'}{\subex{\vartheta}{\#}n\basicsub{\#'}{n}}{\subex{f}{\kappa}{\#'}} \to 
\triple{\Theta,\#}{\subex{\vartheta}{\#}n}{\subex{f}{\kappa}{\#}}
\]
and 
\[
  \iota : \triple{\Theta}{\vartheta}{f} \to 
  \triple{\Theta,\#}{\subex{\vartheta}{\#}n}{\subex{f}{\kappa}{\#}}
\]
is the inclusion. By definition of the presheaf action on the interpretation of $\forall\kappa . A$ we have
 \[
 (\iota \cdot (\pret{\hastype[\Delta]{\Gamma}{t}{\forall\kappa . A}}_{\triple{\Theta}{\vartheta}{f}}(\gamma)))_n
 = 
 \basicsub{\#}{\#'}((\pret{\hastype[\Delta]{\Gamma}{t}{\forall\kappa . A}}_{\triple{\Theta}{\vartheta}{f}}(\gamma))_n)
 \]
 where 
\[
\basicsub{\#}{\#'}  : \triple{\Theta,\#}{\subex{\vartheta}{\#}n}{\subex{f}{\kappa}{\#}}
\to \triple{\Theta,\#, \#'}{\subex{\vartheta}{\#}n\basicsub{\#'}{n}}{\subex{f}{\kappa}{\#'}} 
\]
Since $\basicsub{\#'}{\#} \circ \basicsub{\#}{\#'}$ is the identity, we conclude from (\ref{eq:eta:clocks:subresult}) that
\begin{align*}
 (\pret{\hastype[\Delta]{\Gamma}{\Lambda \kappa . (t [\kappa])}{\forall \kappa . A}}_{\triple\Theta\vartheta f}(\gamma))_n
  & = (\pret{\hastype[\Delta]{\Gamma}{t}{\forall\kappa . A}}_{\triple{\Theta}{\vartheta}{f}}(\gamma)))_n
\end{align*}
concluding the proof. 
\end{proofof}

\subsection{Substitution lemma cases}

We just show the case of the substitution lemma for the type constructor.

Suppose $\pair\nu\sigma : \wfcxt{\Gamma} \to \wfcxt[\Delta']{\Gamma'}$, let $\triple\Theta\vartheta f$ be an object of
$\catT$ and let $\gamma \in \pret{\wfcxt{\Gamma}}_{\triple\Theta\vartheta f}$. We must show that
\[
\nu^*{\pret{\istypeshort[\Delta']{\Gamma'}{\forall \kappa' . A}}}_{\triple\Theta\vartheta f}
(\pret{\pair\nu\sigma}_{\triple\Theta\vartheta f}(\gamma)) 
= \pret{\istypeshort[\Delta]{\Gamma}{\forall \kappa . (A \pair\nu{\subex\sigma{\kappa'}\kappa})}}_{\triple\Theta\vartheta f}(\gamma)
\]
for $\kappa$ fresh for $\Delta$. The left hand side of this equation is a set of families $(\omega_n)_n$ where each $\omega_n$
lives in
\begin{align} \label{eq:sub:case:forall:lhs}
 \pret{\istypeshort[\Delta', \kappa']{\Gamma'}{A}}_{\triple{\Theta,\#}{\subex\vartheta\#n}{\subex{\iota f\nu}{\kappa'}\#}}
 (\inv{\pret{\pair{j}{\id_{\Gamma'}}}}_{\triple{\Theta,\#}{\subex\vartheta\#n}{\subex{\iota f\nu}{\kappa'}{\#}}}
 (\iota \cdot\pret{\pair\nu\sigma}_{\triple\Theta\vartheta f}(\gamma)))
\end{align}
where $j : \Delta' \to \Delta', \kappa'$ is the inclusion. The right hand side is likewise a set of families where each $\omega_n$ 
lives in
\begin{align}\label{eq:sub:case:forall:rhs}
 \pret{\istypeshort[\Delta, \kappa]{\Gamma}{(A \pair{\subex\nu{\kappa'}\kappa}{\sigma})}}_{\triple{\Theta,\#}{\subex\vartheta\#n}{\subex{\iota f}{\kappa}\#}}
 (\inv{\pret{\pair{i}{\id_{\Gamma}}}}_{\triple{\Theta,\#}{\subex\vartheta\#n}{\iota f}}
 (\iota \cdot\gamma))
\end{align}
where $i : \Delta \to \Delta, \kappa$ is the inclusion. Since the naturality requirements for these families are the same, 
we will just show that (\ref{eq:sub:case:forall:lhs}) and (\ref{eq:sub:case:forall:rhs}) are equal. 

\begin{align*}
 & \inv{\pret{\pair{j}{\id_{\Gamma'}}}}_{\triple{\Theta,\#}{\subex\vartheta\#n}{\subex{\iota f\nu}{\kappa'}{\#}}}
 (\iota \cdot\pret{\pair\nu\sigma}_{\triple\Theta\vartheta f}(\gamma)) \\
 & = 
 \inv{\pret{\pair{j}{\id_{\Gamma'}}}}_{\triple{\Theta,\#}{\subex\vartheta\#n}{\subex{\iota f\nu}{\kappa}{\#}}}
 (\pret{\pair\nu\sigma}_{\triple{\Theta,\#}{\subex\vartheta\#n}{\iota f}}(\iota \cdot \gamma)) \\
  & = 
 \inv{\pret{\pair{j}{\id_{\Gamma'}}}}_{\triple{\Theta,\#}{\subex\vartheta\#n}{\subex{\iota f\nu}{\kappa'}{\#}}}
 (i^*\pret{\pair\nu\sigma}_{\triple{\Theta,\#}{\subex\vartheta\#n}{\subex{\iota f}{\kappa}\#}}(\iota \cdot \gamma)) \\
  & = 
 (\subex\nu{\kappa'}{\kappa})^*\inv{\pret{\pair{j}{\id_{\Gamma'}}}}_{\triple{\Theta,\#}{\subex\vartheta\#n}{\subex{\iota f}{\kappa}{\#}}}
 (i^*\pret{\pair\nu\sigma}_{\triple{\Theta,\#}{\subex\vartheta\#n}{\subex{\iota f}{\kappa}\#}}(\iota \cdot \gamma))
\end{align*}
Which by Lemma~\ref{lem:subst:comp:clock:weak} equals
\begin{align} \label{eq:subst:forall:subresult1}
 \pret{\pair{\subex\nu{\kappa'}\kappa}\sigma}_{\triple{\Theta,\#}{\subex\vartheta\#n}{\subex{\iota f}{\kappa'}\#}}
 (\inv{\pret{\pair{j}{\id_{\Gamma}}}}_{\triple{\Theta,\#}{\subex\vartheta\#n}{\iota f}}(\iota \cdot \gamma)) 
\end{align}
So (\ref{eq:sub:case:forall:lhs}) equals $(\subex\nu{\kappa'}\kappa)^*\pret{\istypeshort[\Delta', \kappa']{\Gamma'}{A}}_{\triple{\Theta,\#}{\subex\vartheta\#n}{\subex{\iota f}{\kappa}\#}}$ applied to (\ref{eq:subst:forall:subresult1})
which by the induction hypothesis equals
\begin{align*}
\pret{\istypeshort[\Delta, \kappa]{\Gamma}{(A \pair\nu{\subex\sigma{\kappa'}\kappa})}}_{\triple{\Theta,\#}{\subex\vartheta\#n}{\subex{\iota f}{\kappa'}\#}}
(\inv{\pret{\pair{j}{\id_{\Gamma}}}}_{\triple{\Theta,\#}{\subex\vartheta\#n}{\iota f}}(\iota \cdot \gamma)) 
\end{align*}
which equals  (\ref{eq:sub:case:forall:rhs}), completing the proof.

\section{Interpretation of $\tickc$}

For the proof of Lemma~\ref{lem:dfixapptick} we need the following. 

\begin{lemma}
\label{lem:dfixproperty}
If $\hastype{\Gamma}{t}{\later^\kappa A \to A}$ and $\tickA$ not in $\Gamma$ then 
\[\pret{\hastype{\Gamma}{\lambda (\alpha\of \kappa) (t\, (\dfix \,t))}{\later^\kappa A}} 
= \pret{\hastype{\Gamma}{\dfix\,t}{\later^\kappa A}}\]
\end{lemma}
\begin{proof}
By the substitution lemma
\[
\pret{\hastype{\Gamma, \tickA : \kappa}{t\, (\dfix \,t)}{\later^\kappa A}} 
= \pret{\hastype{\Gamma}{t\, (\dfix \,t)}{\later^\kappa A}} \pret{\pair{\id_{\Delta}}{p}}
\]
where $p$ is the projection. It is easy to see that $\pret{\pair{\id_{\Delta}}{p}} = \p_{\tearlier^\kappa}$, so that
\begin{align*}
 \pret{\hastype{\Gamma}{\lambda (\alpha\of \kappa) (t\, (\dfix \,t))}{\later^\kappa A}} 
 & = \transp{\pret{t\,\dfix\,t} [\p_{\tearlier^\kappa}}] \\
 & = \, \tlater^\kappa(\pret{t\,\dfix\,t} [\p_{\tearlier^\kappa}])[\eta] \\
 & = (\tlater^\kappa \pret{t\,\dfix\,t})[\transp{\p_{\tearlier^\kappa}}]
\end{align*}
and
\[
((\tlater^\kappa \pret{t\,\dfix\,t})[\transp{\p_{\tearlier^\kappa}}])_{\triple \Theta \vartheta f} (\gamma) = 
  \begin{cases} 
    (\tlater^\kappa \pret{t\,\dfix\,t})_{\triple \Theta \vartheta f} (\ast)  = \ast & \vartheta(f(\kappa)) = 0\\
    \pret{t\,\dfix\,t}_{\triple \Theta {\vartheta[f(\kappa)-]} f} (\gamma|_{\triple \Theta {\vartheta[f(\kappa)-]} f}) & \text{Otherwise}
  \end{cases}
\] 
This is by definition equal to $\pret{\hastype{\Gamma}{\dfix\,t}{\later^\kappa A}}_{\triple \Theta \vartheta f}(\gamma)$
\end{proof}

\begin{proofof}{Lemma~\ref{lem:dfixapptick}}
For the equality to be well typed we have $t = u[\kappa'/\kappa]$
\begin{align*}
\pret{\hastype[\Delta]{\Gamma}{\tappc{(\dfix^{\kappa'} \, t)}}{A}} &= \pret{\hastype[\Delta]{\Gamma}{\tappc{(\dfix^{\kappa} \, u)\subst{\kappa}{\kappa'}}} {A}}\\
&= ([\kappa \mapsto \kappa']^* \pret{\hastype[\Delta,\kappa]{\Gamma,\alpha\of\kappa}{\tapp {(\dfix^{\kappa}\, u)}}{A}}) \pret{\pair{[\kappa\mapsto \kappa']}{\basicsub{\tickA}\tickc}}
\end{align*}
By Lemma~\ref{lem:dfixproperty} $\tapp{(\dfix^\kappa \,u)} = u\,(\dfix^\kappa\,u)$ thus
\begin{align*}\pret{\hastype[\Delta]{\Gamma}{\tappc{(\dfix^{\kappa'} \, t)}}{A}} = ([\kappa \mapsto \kappa']^* \pret{\hastype[\Delta,\kappa]{\Gamma,\alpha\of\kappa}{u\,(\dfix^\kappa\,u)}{A}}) \pret{\pair{[\kappa\mapsto \kappa']}{\basicsub{\tickA}\tickc}} \end{align*}
Which by the substitution Lemma is equal to 
\[\pret{\hastype[\Delta]{\Gamma}{u\,(\dfix^\kappa\,u) \pair{[\kappa\mapsto \kappa']}{\basicsub{\tickA}\tickc}}{A}}\]
Which, since $\alpha$ is not free in $u$, is equal to 
\[\pret{\hastype[\Delta]{\Gamma}{u\subst{\kappa'}{\kappa}\,(\dfix^{\kappa'}\,u\subst{\kappa'}{\kappa})}{A}}\]
\end{proofof}


\subsection{Substitution lemma case}
The case of the substitution lemma for application to $\tickc$ is as follows:
Suppose $\wfcxt[\Delta']{\Gamma'}$ and $\kappa' \in \Delta'$ and $\hastype[\Delta',\kappa]{\Gamma'}{t}{\latbind\tickA\kappa A}$.
Suppose moreover that $\pair\nu\sigma : \wfcxt{\Gamma} \to \wfcxt[\Delta']{\Gamma'}$. Assume for simplicity that 
$\kappa\notin \Delta$. We must show that $\pret{\pair\nu\sigma}^*(\pret{\tappc{t\subst\kappa{\kappa'}}}) = \pret{\tappc{t\subst\kappa{\kappa'}}\pair\nu\sigma}$. 

The induction hypothesis will state that 
\[
(\subex\nu\kappa\kappa)^*\pret t[\pret{\pair{\subex\nu\kappa\kappa}\sigma}] = \pret{t\pair{\subex\nu\kappa\kappa}\sigma}
\] 
from which one can prove 
\[
(\subex\nu\kappa\kappa)^*\pret{\istypeshort[\Delta',\kappa]{\Gamma', \tickA : \kappa}{\tapp t}}
[\pret{\pair{\subex\nu\kappa\kappa}{\subex{\sigma}\tickA\tickA}}]
= \pret{\istypeshort[\Delta,\kappa]{\Gamma, \tickA : \kappa}{\tapp{(t\pair{\subex\nu\kappa\kappa}{\sigma})}}}
\] 
as in the case of the substitution lemma for application to a tick variable. For the proof we need the following lemma.

\begin{lemma} \label{lem:sub:clock:extension}
 Let $\pair\nu\sigma : \wfcxt\Gamma\to \wfcxt[\Delta']{\Gamma'}$ be a \clott\ substitution, let $\kappa$ be fresh and 
 $\kappa'\in \Delta$. Then
the two compositions
\[
\pret{\wfcxt\Gamma} \xrightarrow{\pret{\pair{\nu}{\sigma}}} \nu^*\pret{\wfcxt[\Delta']{\Gamma'}} 
\xrightarrow{\nu^*(\pret{\pair{[\kappa\mapsto \kappa']}{\basicsub{\tickA}\tickc}})} 
\nu^*\basicsub{\kappa}{\kappa'}^* \pret{\wfcxt[\Delta',\kappa]{\Gamma', \tickA : \kappa}} 
\]
and $ \basicsub{\kappa}{\nu(\kappa')}^* \pret{\pair{\subex\nu\kappa\kappa}{\subex\sigma\tickA\tickA}} \circ \pret{\pair{[\kappa\mapsto \nu(\kappa')]}{\basicsub{\tickA}\tickc}}$ of type
\[
\pret{\wfcxt\Gamma}\to 
 \basicsub{\kappa}{\nu(\kappa')}^*\pret{\wfcxt[\Delta,\kappa]{\Gamma, \tickA : \kappa}} 
 \to 
 (\basicsub{\kappa}{\nu(\kappa')})^*\subex\nu{\kappa}{\kappa}^* \pret{\wfcxt[\Delta',\kappa]{\Gamma', \tickA : \kappa}} 
\]
are equal. 
\end{lemma}
\begin{proof}
 The statement follows from Lemma~\ref{lem:weak:sub:clock:extension} and Lemma~\ref{lem:subst:comp:clock:weak}.
\end{proof}
From this, the case now follows:
\begin{align*}
& \nu^*(\pret{\tappc{t\subst\kappa{\kappa'}}})[\pret{\pair\nu\sigma}] \\
 & = \nu^*\left((\subex{\id_{\Delta'}}\kappa{\kappa'}^*\pret{\tapp t}) \pret{\pair{[\kappa\mapsto \kappa']}{\basicsub{\tickA}\tickc}}\right) [\pret{\pair\nu\sigma}] \\
 & = \left(\nu^*(\subex{\id_{\Delta'}}\kappa{\kappa'}^*\pret{\tapp t})\right) [\nu^*\pret{\pair{[\kappa\mapsto \kappa']}{\basicsub{\tickA}\tickc}} \circ \pret{\pair\nu\sigma}] \\
 & = \left(\nu^*(\subex{\id_{\Delta'}}\kappa{\kappa'}^*\pret{\tapp t})\right) [\basicsub{\kappa}{\nu(\kappa')}^* \pret{\pair{\subex\nu\kappa\kappa}{\subex\sigma\tickA\tickA}} \circ \pret{\pair{[\kappa\mapsto \nu(\kappa')]}{\basicsub{\tickA}\tickc}}] \\
 & = \left((\basicsub\kappa{\nu(\kappa')}^*(\subex\nu\kappa\kappa)^*\pret{\tapp t})\right) [\basicsub{\kappa}{\nu(\kappa')}^* \pret{\pair{\subex\nu\kappa\kappa}{\subex\sigma\tickA\tickA}}] [\pret{\pair{[\kappa\mapsto \nu(\kappa')]}{\basicsub{\tickA}\tickc}}] \\
 & = \left(\basicsub\kappa{\nu(\kappa')}^*((\subex\nu\kappa\kappa)^*\pret{\tapp t}[\pret{\pair{\subex\nu\kappa\kappa}{\subex\sigma\tickA\tickA}}]\right) [\pret{\pair{[\kappa\mapsto \nu(\kappa')]}{\basicsub{\tickA}\tickc}}] \\
 & = \left(\basicsub\kappa{\nu(\kappa')}^*\pret{\tapp{(t\pair{\subex\nu\kappa\kappa}{\sigma})}}\right) [\pret{\pair{[\kappa\mapsto \nu(\kappa')]}{\basicsub{\tickA}\tickc}}] \\
 & = \pret{\tappc{(t\pair{\subex\nu\kappa{\nu(\kappa')}}{\sigma})}} \\
 & = \pret{(\tappc{t\subst{\kappa}{\kappa'}})\pair{\nu}{\sigma}}
 \end{align*}


\end{document}